\documentclass[11pt]{article}
\usepackage{styling}
\usepackage[ruled,linesnumbered]{algorithm2e}
\usepackage{booktabs}

\date{}
\author{Praneeth Kacham\footnote{CMU. Email: \href{mailto:pkacham@cs.cmu.edu}{\texttt{pkacham@cs.cmu.edu}}} \and Rasmus Pagh\footnote{BARC, University of Copenhagen. Email: \href{mailto:pagh@di.ku.dk}{\texttt{pagh@di.ku.dk}}}  \and Mikkel Thorup\footnote{BARC, University of Copenhagen. Email: \href{mailto:mikkel2thorup@gmail.com}{\texttt{mikkel2thorup@gmail.com}}} \and David P. Woodruff\footnote{CMU. Email: \href{mailto:dwoodruf@cs.cmu.edu}{\texttt{dwoodruf@cs.cmu.edu}}}}

\title{Pseudorandom Hashing for Space-bounded Computation\\ with Applications in Streaming} 

\begin{document}

\maketitle
\thispagestyle{empty}
\begin{abstract}
    We revisit Nisan's classical pseudorandom generator (PRG) for space-bounded computation (STOC 1990) and its applications in streaming algorithms.
    We describe a new generator, \FastPRG, that can be thought of as a symmetric version of Nisan's generator over larger alphabets.
    Our generator allows a trade-off between seed length and the time needed to compute a given block of the generator's output.
    \FastPRG can be used to obtain derandomizations with much better update time and \emph{without sacrificing space} for a large number of data stream algorithms, for example:
    \begin{itemize}
    \item Andoni's $F_p$ estimation algorithm for constant $p>2$ (ICASSP, 2017) assumes a random oracle, but achieves optimal space and constant update time. Using \FastPRG's time-space trade-off we eliminate the random oracle assumption while preserving the other properties. Previously no time-optimal derandomization was known. Using similar techniques, we give an algorithm for a relaxed version of $\ell_p$ sampling in a turnstile stream. Both of our algorithms use $\tilde{O}(d^{1-2/p})$ bits of space and have $O(1)$ update time.
    \item For $0 < p < 2$, the $1 \pm \varepsilon$ approximate $F_p$ estimation algorithm of Kane et al., (STOC, 2011) uses an optimal $O(\varepsilon^{-2}\log d)$ bits of space but has an update time of $O(\log^2(1/\varepsilon)\log\log(1/\varepsilon))$. Using \FastPRG, we show that if $1/\sqrt{d} \le \varepsilon \le 1/d^{c}$ for an arbitrarily small constant $c > 0$, then we can obtain a $1 \pm \varepsilon$ approximate $F_p$ estimation algorithm that uses the optimal $O(\varepsilon^{-2}\log d)$ bits of space and has an update time of $O(\log d)$ in the Word RAM model, which is more than a quadratic improvement in the update time. We obtain similar improvements for entropy estimation. 
    \item CountSketch, with the fine-grained error analysis of Minton and Price (SODA, 2014).
    For derandomization, they suggested a direct application of Nisan's generator, yielding a logarithmic multiplicative space overhead. With \FastPRG we obtain an efficient derandomization yielding the same asymptotic space as when assuming a random oracle.
    Our ability to obtain a time-efficient derandomization makes crucial use of \FastPRG's symmetry. We also give the first derandomization of a recent private version of CountSketch. 
    \end{itemize}
For a $d$-dimensional vector $x$ being updated in a turnstile stream, we show that $\linf{x}$ can be estimated up to an additive error of $\varepsilon\opnorm{x}$ using $O(\varepsilon^{-2}\log(1/\varepsilon)\log d)$ bits of space. Additionally, the update time of this algorithm is $O(\log 1/\varepsilon)$ in the Word RAM model. We show that the space complexity of this algorithm is optimal up to constant factors. However, for vectors $x$ with $\linf{x} = \Theta(\opnorm{x})$, we show that the lower bound can be broken by giving an algorithm that uses $O(\varepsilon^{-2}\log d)$ bits of space which approximates $\linf{x}$ up to an additive error of $\varepsilon\opnorm{x}$. We use our aforementioned derandomization of the CountSketch data structure to obtain this algorithm, and using the time-space trade off of $\FastPRG$, we show that the update time of this algorithm is also $O(\log 1/\varepsilon)$ in the Word RAM model.
\end{abstract}
\newpage
\thispagestyle{empty}
\tableofcontents
\newpage
\setcounter{page}{1}
\section{Introduction}

Space-efficient algorithms are a central theme in computer science.
In many cases, the best such algorithms are \emph{randomized}, which raises the question how to obtain the random bits.
Nisan's classical pseudorandom generator~\cite{nisan} shows that it is possible to expand a small random seed into a longer pseudorandom string that is essentially as good as full randomness if used in a space-bounded computation.
In the context of \emph{streaming algorithms} Nisan's generator has been used not only to reduce the need for random bits, but also because \emph{storing} the seed allows us to re-create random values when they are needed, which is essentially a type of hashing~\cite{charikar2002similarity,indyk2006stable,ahn2012graph,dasgupta2010sparse,jowhari2011tight,gilbert2003one,kane2010exact,frahling2005sampling}.
Ideally we would like the space for the seed to be smaller than the space of the streaming algorithm, but a black-box application of Nisan's generator does not quite live up to this ideal:
the seed length needed is larger than the space of the streaming algorithm by a multiplicative logarithmic factor.
Furthermore, retrieving a block of the string output by the generator requires time proportional to the length of the seed, introducing a significant slowdown in many settings. Hence, obtaining streaming algorithms that are simultaneously space optimal and have a very fast update time---the time required to process each update in the stream---is challenging.

In a turnstile stream, a vector $x \in \R^{d}$ is initially set to $0^d$ and receives updates of the form $(i_1,v_1), (i_2, v_2), \ldots, (i_m, v_m) \in [d] \times \set{-M, \ldots, M}$. On receiving an update of the form $(i_j, v_j)$, the vector $x$ is updated as follows: $x_{i_j} \gets x_{i_{j}} + v_j$. Given a function $f$ with domain $\R^{d}$, at the end of the stream we want to output an approximation to $f(x)$ using space sublinear in $d$ while processing the stream. Some examples of $f$ are (i) the $F_p$ moments $\sum_{i=1}^d |x_i|^p$ and (ii) the number of distinct elements in the stream, often denoted by $\|x\|_0$. 
Turnstile streaming algorithms typically apply a randomized linear map $\bS : \R^{d} \rightarrow \R^{D}$ to the vector $x$ and show that $\bS x$ can be used to approximate $f(x)$ at the end of processing the stream. The advantage of $\bS$ being a linear map (or linear sketch) is that on receiving an update $(i_j, v_j)$ in the stream, the sketch $\bS x$ can be updated by simply adding $(\bS)_{*i_j}v_j$ to the current sketch. Here $(\bS)_{*i_j}$ denotes the $i_j$-th column of the matrix $\bS$. Note that to obtain sublinear space algorithms, we cannot store the full matrix $\bS$ in memory. One of the techniques here is to describe the entries of the matrix $\bS$ using hash functions that are $k$-wise independent for a small value of $k$. For example, the CountSketch matrix of \cite{Charikar2004finding} can be described by using $4$-wise independent hash functions and thus can be stored efficiently. Further, for any $j \in [d]$, the column $(\bS)_{*j}$ can be generated efficiently using the hash functions, thereby allowing for a fast update of the sketch in the stream.

Unfortunately, it is not always easy to show that a matrix $\bS$ generated using hash functions sampled from hash families with limited independence is sufficient to approximate $f(x)$. Indyk \cite{indyk2006stable} showed that we can assume full independence when constructing $\bS$ and later derandomize\footnote{We use the term ``derandomize'' to denote any procedure that \emph{lowers} the randomness required by the algorithm for example by  replacing a uniform random string with a string sampled from a Pseudorandom generator that uses a smaller uniform random seed.} the construction of $\bS$ using pseudorandom generators for small-space computation. The crucial idea of Indyk is that the final state of a linear sketch depends only on the vector $x$ at the end of the stream and not on the sequence of updates that result in the vector $x$. Thus, any algorithm that assumes the columns of $\bS$ are sampled independently can be derandomized using a pseudorandom generator that fools small space algorithms as follows: suppose an algorithm needs $r$ uniform random bits to sample a column of $\bS$. Fixing a vector $x \in \R^d$, we construct a small space algorithm that makes a single pass over an $r \cdot d$ length uniform random string reading $r$ uniform random bits at a time, sampling the column $\bS_{*j}$, and updating the stored sketch by adding $\bS_{*j}x_j$. If each coordinate of $\bS x$ can be stored using $t$ bits, such an algorithm only uses $D \cdot t$ bits of space. As the columns of $\bS$ were sampled independently, we use the analysis assuming full independence to conclude that the sketch computed by the algorithm has certain desired properties with a certain probability. 
Now, the small space algorithm can be \emph{fooled}\footnote{Formally, we say an algorithm using fully-random bits is fooled by a Pseudo Random Generator if the total variation distance between the distribution of outputs of the algorithm when using a string of fully random bits and when using a string sampled from the PRG is small.} using Nisan's PRG with a seed length of $\Theta(Dt \log(rd/Dt))$, which is $\Theta(Dt \log d)$ in many cases as $Dt$ is often much smaller than $d$. The argument essentially shows that sampling columns of $\bS$ using blocks of bits in the pseudorandom string is sufficient to ensure that the properties of $\bS x$, which were proved assuming independent sampling of columns, still hold.

To implement the derandomized algorithm in a stream, given an update $(i_j, v_j)$, we need to generate the column $\bS_{*i_j}$ on the fly. Nisan's PRG allows us to generate a block of bits of the pseudorandom string by sequentially evaluating $O(\log d)$ hash functions $h : \set{0,1}^{Dt} \rightarrow \set{0,1}^{Dt}$ on a $Dt$ length random seed. Thus, if the hash functions $h : \set{0,1}^{Dt} \rightarrow \set{0,1}^{Dt}$ used by the generator can be evaluated in time $T$ on any input, then the block of bits necessary to generate the column $\bS_{*i_j}$ can be computed in time $O(T \log d)$. We think of the pseudorandom string as a ``hash function'' mapping an index $i \in [d]$ to the block of pseudorandom bits needed to generate the column $\bS_{*i}$.

The above argument of Indyk gives a \emph{black box} way to derandomize a streaming algorithm albeit with a logarithmic space blowup ($Dt$ bits to $Dt\log d$ bits) and an update time of $O(T \log d)$. We use the standard Word RAM model with a word size of $w=\Omega(\log d)$ bits to measure the time complexity of the algorithm. If $Dt \gg \log d$, then the hash functions $h : \set{0,1}^{Dt} \rightarrow \set{0,1}^{Dt}$ are slow to compute, making the update time $O(T \log d)$ very large.
We will mostly express space usage of algorithms in \emph{bits}, but with some Word RAM upper bounds having space measured in \emph{words}.

Thus a na\"ive application of Nisan's PRG results in a sub-optimal algorithm space-wise, as well as a large update time. As numerous data stream algorithms use Nisan's PRG for derandomization, it is important to improve upon this. We significantly improve on the space overhead and update time for derandomization with our construction \FastPRG. First, we show that by a careful analysis, in many problems, the pseudorandom string only has to fool an $O(\log d)$ space algorithm instead of the full $O(Dt)$ space algorithm. We show that such a reduction can be performed for the $F_p$ moment estimation algorithms both for $p \in (0,2)$ and for $p > 2$. We further show that using a symmetry property of \FastPRG, we can reduce the derandomization of CountSketch to fooling an $O(\log d)$ space algorithm. We note that CountSketch in our context cannot be derandomized by an application of Nisan's PRG to fooling an $O(\log d)$ space algorithm; see below for further discussion. 
We, however, are able to give a reduction to fooling an $O(\log d)$ space algorithm which enjoys two properties (i) the space complexity of the derandomized algorithm will be $O(Dt + \log^2 d)$,  which is $O(Dt)$ when $Dt = \Omega(\log^2 (d))$ and (ii) the hash functions that are to be evaluated to generate a block of pseudorandom bits that correspond to a column of $\bS$ will now map the domain $\set{0,1}^{O(\log d)}$ to a range $\set{0,1}^{O(\log d)}$. Such hash functions can be evaluated in $O(1)$ time in the Word RAM model. 

Even with our reduction to a PRG needing to fool only an $O(\log d)$ space algorithm, the need to evaluate $O(\log d)$ hash functions one after another to compute a block of pseudorandom bits presents a barrier to obtaining fast update time. We show that in the case of $Dt = d^{\Omega(1)}$, the space-vs-time tradeoff of \FastPRG lets us trade seed length for fast update time, by varying the number of hash functions needed in order to compute a block of pseudorandom bits.

\subsection{Our Results}\label{sec:results}
We construct a new pseudorandom generator, which we call \FastPRG, that satisfies the same guarantees as Nisan's PRG but with an additional symmetry property. Our construction also allows a space-vs-time tradeoff and lets us compute any block of pseudorandom bits quickly if we increase the seed length.
\begin{theorem}[Informal]
    There is a constant $c > 0$ such that for any positive integers $n$, $b$ and $k$ satisfying $b^{k} \le 2^{cn}$, there exists  a pseudorandom generator parameterized by $n$, $b$ and $k$ that converts a random seed of length $O(bkn)$ bits to a bitstring of length $b^k \cdot n$ that cannot be distinguished from truly random bits by any Space($cn$) algorithm making a single pass over the length $b^k \cdot n$ bitstring.
    A given $n$-bit block of this generator can be computed by evaluating $k$ $2$-wise independent hash functions mapping $\set{0,1}^n$ to $\set{0,1}^n$.
\end{theorem}
Our generator uses the random seed of length $O(bkn)$ to sample $b \cdot k$ hash functions from a $2$-wise independent hash family $\calH = \set{h : \set{0,1}^n \rightarrow \set{0,1}^n}$ such as \cite[Theorem~3(b)]{Dietzfelbinger96universal} and uses $n$ bits as an additional random seed. 

In the Word RAM model with a word size $\Omega(n)$, an arbitrary block of $n$ bits in the pseudorandom string generated by our generator can be computed in time $O(k)$. Fixing a value of $b^{k} = t$, we obtain that \FastPRG needs a seed of size $O(t^{1/k} k n)$ to be able to generate a pseudorandom string of length $t \cdot n$ supporting the computation of an arbitrary block of pseudorandom bits in time $O(k)$. By varying $k$, we get a space versus time trade off. Even more informally, our result can be stated as follows:
\begin{theorem}[Informal, Compare with {\cite[Theorem~1]{nisan}}]
A space $S$ algorithm making a single pass over a length $R \le \exp(S)$ random string can be fooled by a pseudorandom generator with a seed length $O((R/S)^{1/k} \cdot k \cdot S)$ where $k$ is an integer parameter of $\FastPRG$. A block of $S$ random bits in the pseudorandom string can be computed by sequentially evaluating $k$ hash functions mapping $\set{0,1}^S$ to $\set{0,1}^S$.
\end{theorem}
Setting $k = \log(R/S)$ in the above theorem, we recover Nisan's result.

Another nice property of \FastPRG is that the distribution of the pseudorandom string is symmetric in the following specific way. Let $\bgamma$ be sampled from $\FastPRG$. Let $\bgamma$ be written as $\bgamma_0 \circ \bgamma_1 \circ \cdots \circ \bgamma_{b^{k}-1}$ where $\circ$ denotes concatenation and each $\bgamma_i$ is a length-$n$ block. For arbitrary $\ell \in \set{0,\ldots,b^{k}-1}$, define $\bgamma^{\oplus \ell} := \bgamma_{0 \oplus \ell} \circ \bgamma_{1 \oplus \ell} \circ \cdots \circ \bgamma_{(b^{k}-1) \oplus \ell}$ where $\oplus$ denotes the bitwise xor operation. The construction of \FastPRG ensures that $\bgamma^{\oplus \ell}$ has the same distribution as $\bgamma$. In all of our algorithms, we use a block of $\bgamma$ to generate appropriate random variables for the corresponding coordinate of the vector $x$ that is being streamed. To analyze the properties of our streaming algorithms, we create \emph{abstract algorithms}\footnote{Algorithms that are created only for analysis and are not implemented.} that make a single pass over the pseudorandom string such that the behavior of the streaming algorithm is the same as the abstract algorithm. Since the distribution of the string $\bgamma$ is the same as the distribution of $\bgamma^{\oplus \ell}$, the distribution of outputs of the abstract algorithm when run on $\bgamma^{\oplus \ell}$ is the same as the distribution of the outputs when run directly on $\bgamma$. So given a fixed $\ell$, the abstract algorithm can ``know'', throughout its execution, the block of pseudorandom bits corresponding to the coordinate $\ell$ that our original streaming algorithm uses. ``Knowing'' these bits is important in our analysis of the CountSketch data structure.

\paragraph{Applications.}
We use \FastPRG to obtain space-optimal algorithms for constant factor $F_p$ ($p > 2$) estimation algorithms with an $O(1)$ update time in turnstile streams. 
Recall the turnstile stream setting: a vector $x \in \R^{d}$ is being maintained in the stream. The vector $x$ is initialized to $0^d$ and receives a stream of updates $(i_1, v_1), (i_2, v_2), \ldots, (i_m, v_m) \in [d] \times \set{-M, \ldots, M}$. Unless otherwise specified, we assume in all of our results that $m, M \le \poly(d)$. For $0 < p < \infty$, we define $F_p (x) := \sum_{i=1}^{d} |x_i|^p$ and $\ell_p(x) := (\sum_{i=1}^d |x_i|^p)^{1/p}$. 
\begin{theorem}
Given $p > 2$, there is a turnstile streaming algorithm that uses $O(d^{1-2/p}\log d)$ words of space and outputs a constant factor approximation to $\lp{x}$. Further, the streaming algorithm processes each update to the stream in $O(1)$ time in the Word RAM model on a machine with $\Omega(\log d)$ word size.
\end{theorem}
We also show that for $p > 2$, similar techniques used in obtaining the above theorem can be used to obtain an algorithm that performs a relaxed version of the approximate $\ell_p$ sampling in the stream. See Section~\ref{subsec:lp-sampling}.

We next give an algorithm to estimate $F_p$ moments for $0 < p < 2$ in the high accuracy regime. We show that the algorithm of \cite{fast-moment-estimation-optimal-space} can be implemented using \FastPRG \emph{without a space blowup} and also ensure that the algorithm has a faster update time.
\begin{theorem}
    Let $0 < p < 2$ be a parameter and $1/\sqrt{d} \le \varepsilon \le 1/d^{c}$ be the desired accuracy for a constant $0 < c \le 1/2$. There is a streaming algorithm that uses $O(\varepsilon^{-2})$ words of space and outputs a $1 \pm \varepsilon$ approximation for $\lp{x}^p$. The streaming algorithm processes each update to the stream in $O(\log d)$ time in the Word RAM model on a machine with $\Omega(\log d)$ word size.
\end{theorem}
Using the reduction in  \cite{fast-moment-estimation-optimal-space} (Appendix A of the conference version), we can obtain a similarly improved update time for both additive and multiplicative entropy approximation in a stream. 

We then derandomize the tighter analysis of CountSketch that \cite{Minton2014improved} show assuming fully random hash functions. We prove that such tighter guarantees can be obtained even if the hash functions and sign functions in the CountSketch are generated from the pseudorandom string sampled using \FastPRG, improving on the black-box derandomization using Nisan's generator.
\begin{theorem}[Informal]
    Given a table size $t$, number $r$ of repetitions, and word size $w = \Omega(\log d)$, there exists a derandomization of CountSketch, $\CS_{\FastPRG} : \set{-M, \ldots, M}^{d} \rightarrow \set{-2^{w}, \ldots, 2^{w}}^{tr}$, with the following properties:
    \begin{enumerate}
        \item The space complexity of the data structure is $O(tr + \log d)$ words.
        \item Given an update $(i,v)$, the data structure can be updated in time $O(r \log d)$ in the Word RAM model.
        \item For any $\alpha \in [0,1]$, given any index $\ell$, an estimate of $x_l$ given by $\hat{x}_{\ell}$ can be constructed from $\CS_{\FastPRG}(x)$ such that
        \begin{align*}
            \Pr[|x_{\ell} - \hat{x}_{\ell}| \ge \alpha \Delta] < 2\exp(-\alpha^2 r) + 2^{-Cw}
        \end{align*}
        for $\Delta = \opnorm{\text{tail}_t(x)}/\sqrt{t}$, where $\text{tail}_t(x) \in \R^d$ denotes the vector obtained after zeroing out the $t$ entries with the highest absolute value in $x$.
    \end{enumerate}
\end{theorem}
We crucially use the symmetry property of \FastPRG to obtain the above result. Note that when $t \cdot r \ge \log d$, the derandomization presents no asymptotic space blowup and retains the tighter analysis of estimation errors from \cite{Minton2014improved}. Further, depending on the parameters $t,r$ it is possible to obtain faster update time using the time-vs-space tradeoff offered by \FastPRG. 
We similarly derandomize the utility analysis of Private CountSketch from \cite{Pagh2022improved} to obtain:
\begin{theorem}[Informal]
    Consider private CountSketch, $\PCS(x) = \CS_{\FastPRG}(x) + \bnu$ where $\bnu \sim N(0,\sigma^2)^{D}$, with $\countsketchrepetitions$ repetitions and table size $\countsketchrange$, where $\FastPRG$ uses block size $w$.
    Given $\ell \in [d]$ we can compute an estimate $\hat{x}_{\ell}$ from $\PCS(x)$ such that for every $\alpha \in [0,1]$ and $\Delta = \opnorm{\text{tail}_t(x)}/\sqrt{t}$,
    \begin{align*}
        \Pr[|\hat{x}_{\ell} - x_{\ell}| \ge \alpha \max(\Delta, \sigma)] \le 2 \exp(-\Omega(\alpha^2 r)) + O(2^{-cw}) \enspace .
    \end{align*}
\end{theorem}

We show the following tight result for estimating $\linf{x}$ in a turnstile stream.
\begin{theorem}[Informal]
    Let $x$ be a an arbitrary $d$ dimensional vector being maintained in a turnstile stream. Assuming that the coordinates of $x$ are integers bounded in absolute value by $\poly(d)$, any streaming algorithm that estimates $\linf{x}$ up to an additive error of $\varepsilon\opnorm{x}$ for $\varepsilon > ((\log d)/d)^{1/4}$ must use $\Omega(\varepsilon^{-2}\log d \log 1/\varepsilon)$ bits. Matching this lower bound, there is an algorithm that uses $O(\varepsilon^{-2}\log d\log 1/\varepsilon)$ bits and outputs an approximation to $\linf{x}$ up to an additive error of $\varepsilon\opnorm{x}$. The update time of this algorithm is $O(\log 1/\varepsilon)$ in the Word RAM model with a word size $\Omega(\log d)$.
\end{theorem}
Our upper bound beats the previous best result of \cite[Theorem~10]{BGW20}. Their algorithm uses $O(\varepsilon^{-2}\log d(\log\log d + \log 1/\varepsilon))$ bits of space. Our matching lower bound shows that our result cannot be improved without making additional assumptions on the vector $x$.

When $\linf{x} = \Theta(\opnorm{x})$, we show that it is possible to break the lower bound in the above result by giving an algorithm that uses $O(\varepsilon^{-2}\log d)$ bits of space and estimates $\linf{x}$ up to an additive error of $\varepsilon\opnorm{x}$. The algorithm uses our derandomization of CountSketch with tight guarantees given in \cite{Minton2014improved}.
\begin{theorem}[Informal]
    Given a $d$ dimensional vector $x$ being updated in a turnstile stream, if $\linf{x} = \Theta(\opnorm{x})$ there is a streaming algorithm that uses $O(\varepsilon^{-2} \log d)$ bits and approximates $\linf{x}$ up to an additive error of $\varepsilon\opnorm{x}$ with probability $\ge 9/10$. The update time of this algorithm is $O(\log 1/\varepsilon)$ in the Word RAM model with a word size $\Omega(\log d)$.
\end{theorem}

\subsection{Previous Work}\label{sec:previous}
\paragraph{$F_p$ estimation for $p > 2$} The problem of estimating moments in a stream has been heavily studied in the streaming literature and has been a source of lot of techniques both from the algorithms and lower bounds perspective. A lower bound of $\Omega(d^{1-2/p})$ bits on the space complexity of a constant factor approximation algorithm was shown in \cite{chakrabarti2003near, bar2004information,j09}. On the algorithms side, \cite{andoni2011streaming} gives a sketching algorithm with $m = O(d^{1-2/p}\log d)$ rows using a technique called \emph{precision sampling}; this improves additional polylogarithmic factors of earlier work \cite{IW05,BGKS06}. For linear sketches, \cite{andoni2013tight} shows a lower bound of $m = \Omega(d^{1-2/p}\log d)$ on the number of rows in the sketch, thereby proving that the algorithm of \cite{andoni2011streaming} is tight up to constant factors for linear sketching algorithms. All upper and lower bounds mentioned here are for constant factor approximation algorithms. See \cite{andoni2017high, andoni2011streaming} and references therein for the upper and lower bounds for $(1 \pm \varepsilon)$-approximate algorithms for $F_p$ estimation.

Later, Andoni \cite{andoni2017high} gave a simpler linear sketch using $m = O(d^{1-2/p}\log d)$ rows for $F_p$ moment estimation. Andoni uses the min-stability property of exponential random variables to embed $\ell_p$ into $\ell_{\infty}$ and then uses a CountSketch data structure to estimate the maximum absolute value in the vector obtained by sketching $x$ with scaled exponential random variables. The analysis assumes that the exponential random variables are sampled independently. To derandomize the algorithm, Andoni uses the pseudorandom generator of Nisan and Zuckerman \cite{nisan-zuckerman} which shows that any randomized algorithm that uses space $s$ and $\poly(s)$ random bits can be simulated using $O(s)$ random bits. Since the space complexity of Andoni's algorithm is $d^{\Omega(1)}$, it can be derandomized with at most a constant factor blowup in space complexity. A major drawback of using the pseudorandom generator of Nisan and Zuckerman is that the update time of the sketch in the stream is $d^{\Omega(1)}$, which is prohibitively large. In this work, we show that we can derandomize Andoni's algorithm while having an update time of $O(1)$ in the Word RAM model.

We note that there are many other algorithms for $F_p$-moment estimation, such as \cite{IW05,BGKS06,BO10,MW10,G15,GW18}, which are based on subsampling the input vector in $O(\log d)$ scales and running an $\ell_2$-heavy hitters algorithm at each scale. Although it may be possible to amortize the update time of the $O(\log d)$ levels of subsampling, such algorithms cannot achieve an optimal $O(1)$ update time since all known algorithms for $\ell_2$-heavy hitters in the turnstile streaming model require $\Omega(\log d)$ update time. 

\paragraph{$F_p$ estimation for $p < 2$} Indyk \cite{indyk2006stable} showed how to estimate $F_p$ moments for $p \in (0,2]$ up to a factor $1 \pm \varepsilon$ in turnstile streams using a space of $O(\varepsilon^{-2}\log d)$ words, which translates to $O(\varepsilon^{-2}\log^2(d))$ bits with our assumption on the values of $m, M$. This work used \emph{$p$-stable distributions} and as discussed earlier, introduced the influential technique to derandomize streaming algorithms using pseudorandom generators. Li \cite{ping-li} used $p$-stable distributions to define the geometric mean estimator, which can be used to give an unbiased estimator for $\lp{x}^p$ with low variance and such that the mean of $\Theta(\varepsilon^{-2})$ independent copies of the estimator gives a $1 \pm \varepsilon$ estimate for $\lp{x}^p$. Li's algorithm can also be derandomized to use  $O(\varepsilon^{-2}\log^2 d)$ bits.

The upper bound was then improved to $O(\varepsilon^{-2}\log d)$ bits by Kane, Nelson and Woodruff \cite{kane2010exact}. They avoid the $O(\log d)$ factor blowup which is caused when derandomizing using Nisan's PRG by showing that Indyk's algorithm can be derandomized using $k$-wise independent random variables for a small value of $k$. They also mention that a ``more prudent analysis'' of the seed length required in Nisan's generator to derandomize Indyk's algorithm makes the space complexity $O(\varepsilon^{-2}\log d + (\log d)^2)$ bits. The idea there was to use Nisan to fool a median of dot products, but it was not able to exploit fast update time, as we do for $0 < p < 2$, without our large alphabet improvement to Nisan's PRG. 

They also show that any algorithm that $1 \pm \varepsilon$ approximates $\lp{x}^p$ in a turnstile stream must use $\Omega(\varepsilon^{-2}\log d)$ bits of space, hence resolving the space complexity of $F_p$ moment estimation in turnstile streams. Although their algorithm uses an optimal amount of space, the update time of their algorithm is $\tilde{O}(\varepsilon^{-2})$ which is non-ideal when $\varepsilon$ is small. Concurrent to their work, \cite{andoni2011streaming} gave a streaming algorithm that has a fast update time of $O(\log d)$ per stream element but uses a sub-optimal space of $O(\varepsilon^{-2-p}\log^2 d)$ bits for $p \in [1,2]$.

\cite{fast-moment-estimation-optimal-space} then made progress by giving algorithms that are space-optimal while having a fast update time. They give an algorithm that uses $O(\varepsilon^{-2}\log d)$ bits of space and with an update time of $O(\log^{2}(1/\varepsilon)\log\log(1/\varepsilon))$ per stream element. They use multiple techniques such as estimating the contribution of heavy and light elements to $F_p$ separately by using new data structures and hash functions drawn from limited independent hash families; they buffer updates and use fast multi-point evaluation of polynomials and amortize the time over multiple updates to obtain fast update times.

When $\varepsilon < 1/d^{c}$ for a small enough constant $c$, which is when the $\tilde{O}(\varepsilon^{-2})$ update time of earlier algorithms becomes prohibitive, we show that we can derandomize the algorithm of \cite{fast-moment-estimation-optimal-space} using \FastPRG. We show that there is a streaming algorithm using an optimal $O(\varepsilon^{-2}\log d)$ bits of space with an update time of $O(\log d)$ per stream element. Our algorithm updates the sketch immediately, removing the need to buffer updates and the use of fast multi-point polynomial evaluation from their algorithm. Our update time thus improves the previous update time of \cite{fast-moment-estimation-optimal-space} from $O(\log^2 d \log \log d)$ to $O(\log d)$ for any polynomially small $\varepsilon$, making the first progress on this problem in over 10 years. 
%

\paragraph{Estimation Error bounds with CountSketch} The CountSketch data structure \cite{Charikar2004finding} can be used to compute an estimate $\hat{x}_{\ell}$ of $x_{\ell}$ for each $\ell \in [d]$. In \cite{Charikar2004finding}, the authors show that with high probability the maximum estimation error $\|x - \hat{x}\|_{\infty} \le \Delta$ where $\Delta = \opnorm{\text{tail}_k(x)}/\sqrt{k}$ if the table size $\countsketchrange = O(k)$ and the number of repetitions $\countsketchrepetitions = O(\log d)$. Minton and Price \cite{Minton2014improved} observed that even though the worst-case estimation error follows the above law, most coordinates in $\hat{x}$ have asymptotically smaller estimation error. Concretely, they show that for any $\alpha \in [0,1]$ and any coordinate $\ell \in [d]$,
$    \Pr[|\hat{x}_{\ell} - x_{\ell}| \ge \alpha\Delta] \le O(\exp(-\alpha^2 \countsketchrepetitions)).$

They also show other applications of the above tighter analysis. While proving the above result, they assume that the sign functions used to construct the CountSketch data structure are fully random. They argue that their construction can be derandomized by incurring a factor $O(\log d)$ overhead in space using Nisan's PRG via the \emph{black box} approach we described earlier. 
We derandomize CountSketch using \FastPRG and show that even for the derandomized CountSketch construction, the above estimation error holds, albeit with an additional additive term from the failure of the pseudorandom generator. 
Our derandomized CountSketch data structure uses $O(\countsketchrepetitions \cdot \countsketchrange + \log d)$ words of space. Note that $O(r \cdot t)$ words of space is anyway required to store the sketched vector and therefore when $\log d = O(\countsketchrepetitions \cdot \countsketchrange)$, our derandomization increases the storage cost by at most a constant factor. Our derandomized CountSketch data structure has an update time of $O(r \log d)$ per stream element. We crucially use the symmetry property of \FastPRG to reduce the problem to derandomizing a small space algorithm. While Minton and Price say that a derandomization of their algorithm with Nisan's PRG incurs an $O(\log d)$ factor space blow-up in the size of the CountSketch data structure, we observe that we can avoid the blow-up by a more careful analysis of derandomization using Nisan's PRG. We use the fact that Nisan's PRG is resilient to multiple passes \cite{david2011strong} as well to obtain the derandomization. See Section~\ref{subsec:alternate-derandomizations-nisan} for a discussion on how to derandomize CountSketch using Nisan's PRG and how the derandomization with Nisan's PRG compares against derandomization with \FastPRG.

In the case of $rt = o(\log d)$, our derandomization is not ideal as the asymptotic space complexity of derandomized CountSketch is larger than a constant factor as compared to $O(\countsketchrepetitions \cdot \countsketchrange)$ words of space. Jayaram and Woodruff \cite{Jayaram2018perfect} give an alternate derandomization of CountSketch with strong estimation error guarantees and show that their derandomized CountSketch needs $O(\countsketchrepetitions \cdot \countsketchrange (\log \log d)^2)$ words of space. When $r \cdot t = o(\log d/ (\log \log d)^2)$, their derandomization has a smaller space complexity than ours. They use a half-space fooling pseudorandom generator of Gopalan, Kane and Meka \cite{gopalan2018pseudorandomness} to derandomize CountSketch. However, the CountSketch data structure derandomized in \cite{Jayaram2018perfect} is a slight modification of the standard CountSketch data structure, and combined with the pseudorandom generator of \cite{gopalan2018pseudorandomness},  leads to worse update times as compared to our derandomization of the standard CountSketch data structure. We also stress that in a number of applications of CountSketch, such as to the $\ell_2$-heavy hitters problem, one has $r \cdot t = \Omega(\log d)$, and in this regime our derandomization is space-optimal, and not only significantly improves the update time, but also removes the additional $(\log \log d)^2$ factors in the space of \cite{Jayaram2018perfect}. 

\paragraph{Estimating $\linf{x}$} Given a $d$ dimensional vector $x$ being maintained in a turnstile stream, it can be shown that approximating $\linf{x}$ up to a multiplicative factor $C < \sqrt{2}$ requires $\Omega(d)$ bits of space by reducing from the INDEX problem. Notably, the lower bound for multiplicative approximation holds even in the stronger addition only model. Hence the problem of approximating $\linf{x}$ up to an additive error of $\varepsilon\opnorm{x}$ has gained more interest. Cormode, in the 2006 IITK workshop on data streams\footnote{Follow this \href{https://www.semanticscholar.org/paper/OPEN-PROBLEMS-IN-DATA-STREAMS-AND-RELATED-TOPICS-ON-Agarwal-Baswana/5394ab5bf4b66bfb52f111525d6141a3226ba883}{link} for the list of open problems.}, asked if it was possible to approximate $\linf{x}$ to an additive error of $\varepsilon\opnorm{x}$ using fewer than $O(\varepsilon^{-2}\log^2 d)$ bits of space. Later \cite{BCIW16} for insertion only streams and \cite{BGW20} for general turnstile streams answered the question in affirmative by giving an algorithm that uses only $O(\varepsilon^{-2}\log d\log \log d)$ bits of space.

\subsection{Technical Overview}

\paragraph{\FastPRG} We will briefly describe our construction and show how the construction and analysis differs from Nisan's. Given parameters $n, b, k$, our construction samples $b \cdot k$ independent hash functions $\bh_i^{(j)} : \set{0,1}^n \rightarrow \set{0,1}^n$ for $(i, j) \in \set{0,\ldots,k-1} \times \set{0,1,\ldots,b-1}$ from a $2$-wise independent hash family. We then sample a uniform random string $\br \sim \set{0,1}^n$, then the $b^k \cdot n$ pseudorandom string output by the generator is defined by $G_k(\br, \bh_1,\ldots,\bh_k)$, where for any $x$,
\begin{align*}
     G_0(x) &\coloneqq x\\
     G_k(x, \bh_0,\ldots,\bh_{k-1}) &\coloneqq 
     G_{k-1}(\bh_{k-1}^{(0)}(x), \bh_0,\ldots,\bh_{k-2}) \circ \cdots \circ G_{k-1}(\bh_{k-1}^{(b-1)}(x), \bh_1,\ldots,\bh_{k-1}).
\end{align*}
Thus for $i \in \set{0,\ldots,b^k-1}$ if $i$ is written as $(i_{k-1} i_{k-2} \cdots i_0)$ in base $b$, then the $i$-th block of $n$ bits in the pseudorandom string $G_k(\br, \bh_1,\ldots,\bh_k)$ is given by
\begin{align*}
    \bh_0^{(i_0)}( \cdots (\bh_{k-1}^{(i_{k-1})}(\br))).
\end{align*}

Our analysis of \FastPRG is based on a new, simpler, and more precise analysis of Nisan's generator~\cite{nisan} (see Appendix~\ref{sec:nisan} for the definition of Nisan's PRG). Note that Nisan's generator
corresponds to the special case
of our generator where $b=2$ and
where for all $i,x$, we 
deterministically fix 
$\bh_i^{(0)}(x) := x$. With $\bh_i^{(0)}$ the identity function, Nisan only defines
a single hash function
$h_i=\bh_i^{(1)}$ for each $i$.

We will now pinpoint where our
analysis diverges from a more
natural generalization of Nisan's \cite{nisan}. The first main difference is in Definition \ref{def:set-independence}.
It plays a role similar to Nisan's \paragraph{Definition 3} Let 
$A\subset \set{0,1}^{n}$, 
$B\subset \set{0,1}^{m}$, $h:\set{0,1}^{n}\to\set{0,1}^{m}$,
and $\varepsilon>0$. We say that
$h$ is $(\varepsilon,A,B)$-independent if $|\Pr_{x\in \set{0,1}^{n}}[x\in A\textnormal{ and }h(x)\in B]-\varrho(A)\varrho(B)|\leq\varepsilon,$
where $\varrho(A)=|A|/2^n$ and $\varrho(B)=|B|/2^m$.\medskip

\noindent 
However, our Definition \ref{def:set-independence} instead generalizes the following definition:
\paragraph{Definition 3'} Let 
$A\subset \set{0,1}^{n+m}$, 
$h:\set{0,1}^{n}\to\set{0,1}^{m}$,
and $\varepsilon>0$. We say that
$h$ is $(\eps,A)$-independent if
$|\Pr_{x\in \set{0,1}^{n}}[(x,h(x))\in A]-\varrho(A)|\leq\varepsilon,$ 
where $\varrho(A)=|A|/2^{n+m}$.

\medskip

\noindent It turns out that
replacing Definition 3 with Definition 3' in Nisan's proof
yields a slightly tighter
result. In the proof of 
\cite[Lemma 2]{nisan}, in step 2, all triplets of
state nodes $i,l,j$ are considered, where $i$ is
a start state, $l$ is an intermediate state and $j$ is an end state. Referring to Definition~3, Nisan uses $A=B_{i,l}^{h_1,\ldots,h_{k-1}}\subset \set{0,1}^n$
and $B=B_{l,j}^{h_1,\ldots,h_{k-1}} \subset \set{0,1}^n$.  If  instead we base
the analysis
on Definition~3', then we only 
need to consider pairs
of state nodes $i,j$ and use
$A=B_{i,j}^{h_1,\ldots,h_{k-1}} \subset \set{0,1}^{2n}$. This affects the whole analysis of Nisan, but it turns out that it only gets simpler, and this was the starting point for our generalized analysis. Avoiding the
intermediate state $l$ was crucial for us because we would need $b-1$
intermediate states $l_1\ldots l_{b-1}$, and this would lead to much worse bounds with larger $b$.

Our construction also implies that the pseudorandom generator has a symmetry property. To see the symmetry, suppose we define $\bh_0',\ldots,\bh_{k-1}'$ with $(\bh_0')^{(0)} = \bh_0^{(1)}$, $(\bh_0')^{(1)} = \bh_0^{(0)}$ and $(\bh'_i)^{(j)} = \bh_i^{(j)}$ in all other cases. Then the string $G_k(\br, \bh_0', \ldots, \bh_{k-1}')$ is obtained by an appropriate permutation of blocks in the string $G_k(\br, \bh_0, \ldots, \bh_{k-1})$. As both the strings $G_k(\br, \bh_0, \ldots, \bh_{k-1})$ and $G_k(\br, \bh_0', \ldots, \bh_{k-1}')$ are just as likely when sampling from \FastPRG, we obtain the symmetry property. 

Extending the above switching argument, we get the following family of transformations that preserve the distribution of the pseudorandom string. Consider an arbitrary $\ell \in \set{0, \ldots, b^k - 1}$. Let $\ell_{k-1} \cdots \ell_0$ be the base $b$ representation of $\ell$. Suppose $\bh_0, \ldots, \bh_{k-1}$ be the hash functions sampled in the construction of $\FastPRG$. Define $\bh_0',\ldots,\bh_{k-1}'$ as
\begin{align*}
    (\bh_i')^{(j)} \coloneqq \bh_i^{(j \oplus \ell_i)}
\end{align*}
for all $i \in \set{0, \ldots, k-1}$ and $j \in \set{0, \ldots, b-1}$. Clearly, the joint distribution of $(\bh_0, \ldots, \bh_{k-1})$ is the same as $(\bh_0', \ldots, \bh_{k-1}')$. Now, for any fixed $x$, we have
\begin{align*}
    G_k(x, \bh_0, \ldots, \bh_{k-1})^{\oplus \ell} = G_k(x, \bh_0', \ldots, \bh_{k-1}').
\end{align*}
Hence if $\bgamma \sim \FastPRG$, then $\bgamma^{\oplus \ell}$ has the same distribution as $\bgamma$.

\paragraph{$F_p$ estimation for $p > 2$} We derandomize Andoni's algorithm for $F_p$ moment estimation \cite{andoni2017high}. Andoni's algorithm can be seen as sketching in two stages: let $\bz \in \R^{d}$ be a random vector such that $\bz_i = \bE_{i}^{-1/p}x_i$ where $\bE_1,\ldots,\bE_d$ are independent standard exponential random variables. By the min-stability property of exponential random variables, we obtain that $\linf{\bz}^p \sim \lp{x}^p/\bE$, where $\bE$ is also a standard exponential random variable. Thus, the coordinate of maximum absolute value in $\bz$ can be used to estimate $\lp{x}$. Notice that we have not yet performed any dimensionality reduction. 
Then a CountSketch matrix $\bS$ with $O(d^{1-2/p}\log d)$ rows is constructed using $O(\log d)$-wise independent hash functions, and this is applied to the vector $\bz$ to obtain $\boldf = \bS \bz$. Andoni argues using the properties of exponential random variables that all of the following hold true simultaneously with a large constant probability: (i) $\linf{\bz}$ is a constant factor approximation to $\lp{x}$, (ii) there are only $O((\log d)^{p})$ coordinates in $\bz$ with absolute value greater than $\lp{x}/c\log d$ and (iii) $\opnorm{\bz}^2 \le O(d^{1-2/p} \opnorm{\bx}^2)$. Conditioned on these properties of $\bz$, Andoni argues that $\linf{\boldf} \approx \linf{\bz} \approx \lp{x}$ with a large probability. Hence, $\linf{\boldf}$ is a constant factor approximation for $\lp{x}$ with a large constant probability. 

We only have to derandomize the exponential random variables as $\bS$ is constructed using $O(\log d)$-wise independent hash functions which can be efficiently stored and evaluated \cite{CPT15}. So, we want to show that each of the three properties of the vector $\bz$ hold even when the exponential random variables are sampled using a pseudorandom string. Now fix a vector $x$. There is an $O(\log d)$ space algorithm that makes a single pass over the string used to generate exponential random variables and (1) computes $\linf{\bz}$, (2) computes the number of coordinates in $\bz$ with absolute value at least $\lp{x}/c \log d$, and (3) computes $\opnorm{\bz}^2$. The algorithm is simple: it goes over a block of the string to generate $\bE_1$, sets $\linf{\bz} = |\bE_1^{-1/p}x_1|$, increases a counter if $|\bE_1^{-1/p}x_1| \ge \lp{x}/(c \log d)$ and sets $\opnorm{\bz}^2 = (\bE_1^{-1/p}x_1)^2$, and proceeds to read the next block of bits to generate $\bE_2$ and update the variables using the value of $\bE_2$ accordingly and continues so on. 
Now, using any PRG that fools an $O(\log d)$ space algorithm, we obtain that the random variables $\bE_i$ constructed using the pseudorandom string also make the vector $\bz$ have each of the three properties. Using the time-vs-space trade-off of \FastPRG, we obtain that any block of the pseudorandom string can be obtained in $O(1)$ time if the seed length of \FastPRG is $d^{\epsilon}$ for a small constant $\epsilon > 0$, thus making the time to compute a block of pseudorandom bits $O(1)$ in the Word RAM model. 
Further, the $O(\log d)$-wise independent hash functions necessary to map the vector $\bz$ to $\boldf$ when sampled from the constructions of \cite{CPT15} allow for the hash functions to be evaluated in $O(1)$ time. The data structures that allow the hash value to be computed in $O(1)$ time can be stored in $d^{\epsilon}$ bits of space as well for any constant $\epsilon > 0$. Thus, the overall update time of the algorithm is $O(1)$ in the Word RAM model. If $\epsilon$ is chosen smaller than $1 - 2/p$, then the asymptotic space complexity of our derandomized $F_p$ moment estimation algorithm remains $O(d^{1-2/p}\log d)$ words, while having a very fast $O(1)$ update time. 

We note that the $F_p$ estimation problem for $p > 2$ is ideally suited for \FastPRG{} --- it is precisely because the algorithm uses a large amount of memory already that we are able to use our generator over a large alphabet without a space overhead, which then allows us to remove the $O(\log d)$ factor in the update time and achieve constant time. Moreover, it is critical that we can keep track of various quantities needed to fool the algorithm with only $O(1)$ words of memory, as this is also needed for fast update time in our derandomization.  

\paragraph{$F_p$ estimation for $p < 2$} We give an alternate derandomization of the algorithm of \cite{fast-moment-estimation-optimal-space}. Their algorithm is based on Li's geometric mean estimator \cite{ping-li} using $p$-stable random variables. They introduce two new data structures they call \HighEnd{} and \LightEstimator{}. They also concurrently run an $\ell_p$ heavy hitters algorithm and at the end of processing the stream, they use the heavy hitters data structure to find a set $L \subseteq [d]$ of coordinates such that $\setbuilder{i}{|x_i|^p \ge (\alpha/2)\lp{x}^p} \supseteq L \supseteq \setbuilder{i}{|x_i|^p \ge \alpha\lp{x}^p}$. They then use the \HighEnd{} data structure to estimate $\lp{x_L}^p$ up to a factor of $1 \pm \varepsilon$. By definition of $L$, all the coordinates in $x_{[d] \setminus L}$ have a small magnitude. Using this fact, they show that their \LightEstimator{} data structure can be used to give a low variance estimator for $\lp{x_{[d] \setminus L}}^p$. By running multiple independent copies of \LightEstimator{} concurrently, they obtain an accurate estimate for $\lp{x_{[d] \setminus L}}^p$ and then output the sum of estimates of $\lp{x_{L}}^p$ and $\lp{x_{[d] \setminus L}}^p$. 

The \HighEnd{} estimator can be maintained using $O(\varepsilon^{-2}\log d)$ bits and has an update time of $O(\log d)$. Hence, we focus on giving a more efficient derandomization of \LightEstimator{}. At a high level, they hash coordinates of $x$ into $O(1/\alpha)$ buckets and for each bucket they maintain Li's estimator for the $F_p$ moment of coordinates hashed into that bucket. At the end of the stream, the set $L$ is revealed and they output the sum of Li's estimators of the buckets into which none of the elements of $L$ are hashed into. They scale the sum appropriately to obtain an unbiased estimator to $\lp{x_{[d] \setminus L}}^p$. They show that hashing of the coordinates can be quickly performed using the hash family of \cite{pagh2008uniform}. The only thing that remains is to derandomize Li's estimator in  each individual bucket. They prove that the $p$-stable random variables in Li's estimator can be derandomized by using $O(1/\varepsilon^p)$-wise independent random variables and show that this is sufficient to obtain algorithms with an optimal space complexity of $O(\varepsilon^{-2}\log d)$ bits and an update time of $O(\log^2(1/\varepsilon)\log\log(1/\varepsilon))$. While the other parts of their algorithm have fast update times, the \LightEstimator{} derandomized using limited independent $p$-stable random variables leads to a slow update time. We give an alternate derandomization of \LightEstimator{} using \FastPRG. 

Fix a particular bucket $b$ and the hash function $h$ that hashes the coordinates into one of the buckets. We give an $O(\log d)$ space algorithm that makes a single pass over the string used to generate $p$-stable random variables and computes Li's estimator for bucket $b$. The algorithm simply makes a pass over the string, and if a coordinate $i$ gets hashed into bucket $b$, it uses the block of bits in the string corresponding to the $i$-th coordinate to generate $p$-stable random variables and updates Li's estimator for bucket $b$ using the generated $p$-stable random variables. The existence of such an algorithm implies that the expectation of Li's estimator is fooled by \FastPRG and therefore the expectation of the sum of Li's estimators over all the buckets is fooled as well by \FastPRG. We now need to bound the variance of the sum of estimators of the $b$ buckets. Here to fool the variance, that is, we only have to fool $\E[\Est_b \cdot \Est_{b'}]$ for pairs of buckets $b, b'$. The idea of using a PRG to fool the variance for a streaming problem has also been used in \cite{anonymous}. Here $\Est_{b}$ denotes the result of Li's estimator for bucket $b$. Now, for any fixed pair of buckets $b, b'$ we again have that there is an $O(\log d)$ space algorithm that computes Li's estimators for buckets $b$ and $b'$ simultaneously while making a single pass over the string used to generate $p$-stable random variables. This shows that $\FastPRG$ fools $\E[\Est_{b} \cdot \Est_{b'}]$ and hence the overall variance of the estimator. As we only need to fool the mean and variance, we obtain a derandomization of $p$-stable random variables using \FastPRG. Additionally, when $\varepsilon < 1/d^{c}$, we can use the time-vs-space tradeoff of \FastPRG to obtain an update time of $O(\log d)$ without changing the asymptotic space complexity of the algorithm. Note that the case of $\varepsilon$ being small is actually the setting for which we would most like to improve the update time of \cite{fast-moment-estimation-optimal-space}.

\paragraph{Entropy estimation} An improved update time for $F_p$ moment estimation for $p \in (0,2)$ also leads to improved update time for entropy estimation using the algorithm of \cite{harvey2008sketching}. See Section A in the conference version of \cite{fast-moment-estimation-optimal-space} for a discussion on how the update time of $F_p$ moment estimation algorithms translate to update time of approximate entropy estimation algorithms.

\paragraph{CountSketch} CountSketch is a randomized linear map $\CS : \R^{d} \rightarrow \R^{D}$ defined by two parameters: (i) the table size $\countsketchrange$ and (ii) the number of repetitions $\countsketchrepetitions$. Here $D = \countsketchrepetitions \cdot \countsketchrange$. For each $i \in [\countsketchrepetitions]$, we have a hash function $\bg_i : [d] \rightarrow [\countsketchrange]$ and a sign function $\bs_i : [d] \rightarrow \set{+1, -1}$. Indexing the coordinates of $\CS(x)$ by $(i,j) \in [\countsketchrepetitions] \times [\countsketchrange]$, we define $\CS(x)_{i,j} = \sum_{\ell \in [d]} [\bg_i(\ell) = j]\bs_i(\ell)x_{\ell}$. Thus,  for each repetition $i$, the coordinate $x_{\ell}$ is multiplied with a sign $\bs_i(\ell)$ and added to the $\bg_i(\ell)$-th bucket. For each $\ell \in [d]$, we can define 
$    \hat{x}_{\ell} = \text{median}(\setbuilder{\bs_{i}(\ell) \cdot \CS(x)_{i, \bg_i(\ell)}}{i \in [\countsketchrepetitions]})$.
The randomness in a CountSketch data structure is from the hash functions $\bg_i$ and the sign functions~$\bs_i$. Assuming that the hash functions $\bh_i$ and $\bs_i$ for $i \in [\countsketchrange]$ are drawn independently from 2-wise and 4-wise independent hash families respectively, \cite{Charikar2004finding} showed that if $\countsketchrepetitions = O(\log d)$, then with probability at least $1 - 1/\poly(d)$, $\linf{x - \hat{x}} \le \Delta$ for $\Delta = \opnorm{\text{tail}_k(x)}/\sqrt{k}$ if $\countsketchrange = O(k)$. 
As discussed in Section \ref{sec:previous}, Minton and Price assume that the hash functions $\bg_i$ and the sign functions $\bs_i$ are fully random to give probability bounds on estimation error for any particular index $\ell \in [d]$. We derandomize their construction by using a pseudorandom generator to sample the hash functions $\bg_i$ and the sign functions $\bs_i$. Suppose we treat a bitstring $\gamma$ as $\countsketchrange \cdot d$ blocks of equal length. We index the blocks of $\gamma$ by $(i, \ell) \in [\countsketchrepetitions] \times [d]$. We use the block $\gamma_{i, \ell}$ to define $\bg_i(\ell)$ and $\bs_i(\ell)$ in the natural way. If $\gamma$ is sampled uniformly at random, then clearly we have that $\bg_i$ and $\bs_i$ constructed using the string $\gamma$ are fully random and the CountSketch data structure constructed using such hash functions satisfies the guarantees given by Minton and Price. 
We need to define which block of the string corresponds to which $(i,\ell)$, since when we receive an update in the stream, we need to be able to extract the corresponding block from the pseudorandom string. Now, we want to show that even if $\gamma$ is sampled from \FastPRG, we obtain similar guarantees on the estimation error. Fix a vector $x$ and coordinate $\ell$. Our strategy to derandomize has been to give an algorithm using space $O(\log d)$ bits that makes a single pass over the randomness and computes the quantity of interest. Now, in a black box way, we can conclude that the distribution of the quantity of interest does not change much even if a string sampled from \FastPRG is used instead of a fully random string. 

We shall try to employ the same strategy here. Fixing an $x \in \R^d$, $\ell \in [d]$ and $\alpha \in [0,1]$, we want to walk over the string $\gamma$ and count the number of repetitions $i \in [\countsketchrepetitions]$ for which the value $\bs_{i}(\ell) \cdot \CS(x)_{i, \bg_i(\ell)} > x_{\ell} + \alpha \Delta$ and the number of repetitions $i \in [\countsketchrepetitions]$ for which the value $\bs_{i}(\ell) \cdot \CS(x)_{i, \bg_i(\ell)} < x_{\ell} - \alpha \Delta$. Clearly the estimator $\hat{x}_{\ell} \in [x_{\ell} - \alpha \Delta, x_{\ell} + \alpha \Delta]$ if and only if both counts are smaller than $r/2$ (for odd $r$).
Hence, if both the counts can be computed by a small space algorithm, we can derandomize CountSketch using \FastPRG. We immediately hit a roadblock while trying to design such an algorithm. As we fixed an index $\ell \in [d]$, for repetition $i \in [\countsketchrepetitions]$, the quantity of interest is $\CS(x)_{i, \bg_i(\ell)} \coloneqq \sum_{\ell' \in [d]}[\bg_i(\ell') = \bg_i(\ell)]\bs_i(\ell')x_{\ell'}$. As the string $\gamma$ is ordered in increasing order of $\ell'$ for each repetition $i$, the algorithm does not know the value of $\bg_i(\ell)$ until it gets to the block $\gamma_{i, \ell}$. So, for indices $\ell' < \ell$, the algorithm is not aware if $\bg_i(\ell')$ equals $\bg_i(\ell)$ or not and hence cannot track the value of $\CS(x)_{i, \bg_i(\ell)}$ with a single pass over $\gamma$. To solve the problem, we use the symmetric property unique to \FastPRG. As we mentioned in Section \ref{sec:results}, if $\bgamma \sim \FastPRG$, then for any integer $m$, the string $\bgamma^{\boxplus m}$ has the same distribution as $\bgamma$. Making use of this symmetry property, we can now assume that the block $\bgamma_{i,j}$ actually corresponds to the hash values of the index $j \boxplus \ell$ for iteration $i$ as the joint distribution of hash and sign values defined by the earlier ordering of the blocks will be the same as the new ordering of the blocks by the symmetry property.

Thus, reading the block $\bgamma_{i,1}$, the algorithm immediately knows which bucket the index $\ell$ gets hashed into in the $i$-th repetition of CountSketch. Now in a single pass over the blocks $\bgamma_{i,1},\ldots,\bgamma_{i, d}$, an $O(\log d)$ space algorithm can compute $\bs_i(\ell)\CS(x)_{i, \bg_i(\ell)}$ and at the end of traversing the blocks corresponding to the $i$-th iteration, the algorithm checks if $\bs_i(\ell)\CS(x)_{i, \bg_i(\ell)}$ is $> x_{\ell} + \alpha \Delta$ or $< x_{\ell} - \alpha \Delta$ and updates the corresponding counters accordingly. Thus, we have an $O(\log d)$ space algorithm and \FastPRG fools an $O(\log d)$ space algorithm and can be used to derandomize CountSketch with the stronger guarantees as given by Minton and Price without incurring a space blowup when $\log d = O(\countsketchrepetitions \cdot \countsketchrange)$. 

For each update in the stream, and for each of the $r$ repetitions, we need to compute a block of the pseudorandom string which takes $O(\log d)$ time in the Word RAM model with a word size $w = \Omega(\log d)$. Hence, our derandomized CountSketch with strong guarantees from \cite{Minton2014improved} has an update time of $O(r \log d)$. When $r \cdot t = d^{\Omega(1)}$, using the space-vs-time trade-off of \FastPRG, we can obtain an update time of $O(r)$.

\paragraph{Private CountSketch} In a recent work, Pagh and Thorup \cite{Pagh2022improved} gave an improved analysis of the estimation error of differentially private CountSketch. After computing $\CS(x)$ in the stream, they compute $\PCS(x) = \CS(x) + \bnu$ where $\bnu \sim N(0,\sigma^2)^{D}$. They show that for an appropriate value of $\sigma$, $\PCS(x)$ is $(\varepsilon, \delta)$-differentially private. They also show that for any $\ell \in [d]$, the estimator $\hat{x}_{\ell}$ computed using $\PCS(x)$ also concentrates heavily around $x_{\ell}$, and gave similar concentration bounds to that of Minton and Price \cite{Minton2014improved}. The analysis of \cite{Pagh2022improved} assumes that the hash functions and sign functions are fully random -- the same assumption as in \cite{Minton2014improved}. We derandomize the Private CountSketch construction using \FastPRG, which is similar to our derandomization of  \cite{Minton2014improved}.

\paragraph{Approximating $\linf{x}$} We consider the problem of approximating $\linf{x}$ up to an additive error $\varepsilon\opnorm{x}$. Directly using the CountSketch data structure, we can approximate each coordinate of the vector $x$ up to an additive error of $\varepsilon\opnorm{x}$ using $O(\varepsilon^{-2}(\log d)^2)$ bits. In the space complexity, a $\log d$ factor comes from the fact that we have to store integers with magnitudes $\poly(d)$ and other $\log d$ factor is because of $O(\log d)$ repetitions of CountSketch to be able to take a union bound over the reconstruction error of all $\poly(d)$ coordinates. We show that there is a simple linear sketching technique to reduce the dimension from $d$ to $\poly(1/\varepsilon)$ while preserving the $\linf{x}$ up to an additive error of $\varepsilon\opnorm{x}$. Then the CountSketch data structure, using a space of $O(\varepsilon^{-2}\log 1/\varepsilon \log d)$ bits, can be used to approximate the $\linf{\cdot}$ of the sketched vector thereby approximating $\linf{x}$. We use the randomized map $\bL : \R^{d} \rightarrow \R^t$ defined as follows to reduce the dimension: let $\bh : [d] \rightarrow [t]$ be drawn at random from a $2$-wise independent hash family and $\bs : [d] \rightarrow \set{+1,-1}$ be drawn at random from a $4$-wise independent hash family. Define $(\bL x)_i = \sum_{j \in [d]: \bh(i) = j}\bs(j)x_j$. Using simple variance computations, we show that if $t = \poly(1/\varepsilon)$, with a high probability, $\linf{\bL x} = \linf{x} \pm \varepsilon\opnorm{x}$ and that $\opnorm{\bL x} \le 2\opnorm{x}$. Note that a turnstile update to one coordinate of $x$ simply translates to a turnstile update to one coordinate of $\bL x$  and hence the two stage sketching algorithm can be efficiently implemented in a stream. A similar universe reduction was employed in \cite{fast-moment-estimation-optimal-space} but the technique was previously not explored in the context of $\ell_{\infty}$ estimation.

We further show that this simple algorithm is tight by showing an $\Omega(\varepsilon^{-2}\log 1/\varepsilon \log d)$ bits lower bound on any algorithm that approximates $\linf{x}$ up to an additive $\varepsilon\opnorm{x}$. We define a communication problem called ``Augmented Sparse Set-Disjointness''. In this one-way communication problem, Alice receives sets $A_1,\ldots, A_t$ with the property that $|A_i| = k$ and $A_i \subseteq [n]$ for all $i \in [t]$. Bob similarly receives the sets $B_1, \ldots, B_t$ with the same properties. Given an index $i \in [t]$ and the sets $A_1,\ldots,A_{i-1}$, using a single message $M$ (possibly randomized) from Alice, Bob has to compute if $B_i \cap A_i = \emptyset$ or not. In the case of $t=1$, \cite{dasgupta2010sparse} show that if $n \ge k^2$, then the sparse set disjointness problem has a communication lower bound of $\Omega(k\log k)$. Note that for $t=1$, simply sending the entire set $A_1$ to Bob requires a communication of $O(k\log n)$ bits. Surprisingly, they show that there is a protocol using $O(k \log k)$ bits to solve the problem. Extending their ideas, we show that the Augmented Sparse Set-Disjointness problem has a communication lower bound of $\Omega(tk\log k)$. Embedding an instance of the Sparse Set-Disjointness into approximating $\linf{x}$ for an appropriate vector $x$, we show a lower bound of $\Omega(\varepsilon^{-2}\log 1/\varepsilon\log d)$ bits.

The hard distribution in the lower bound has the property that the vector $x$ constructed satisfies $\linf{x} = O(\varepsilon\opnorm{x})$. We show that it is possible to break the lower bound if we assume that $\linf{x} \ge c\opnorm{x}$ for a constant $c$. This is the case when the coordinates of $x$ follow Zipf's law where the $i$-th largest coordinate has a value approximately $i^{-\alpha}$ for $\alpha > 0.5$. The algorithm is again the two stage sketch we described but in the second stage, instead of using CountSketch as described in \cite{Charikar2004finding} using constant wise independent hash functions, we use the tighter CountSketch guarantees that we obtain by derandomizing the analysis of \cite{Minton2014improved} using \FastPRG. We show that after the first level sketching, only a few large coordinates need to be estimated to large accuracy while the rest of the coordinates can have larger estimation errors. Using this insight, we show that $O(\varepsilon^{-2}\log d)$ bits of space is sufficient to estimate $\linf{x}$ up to an additive error of $\varepsilon\opnorm{x}$.
\section{Preliminaries}\label{sec:preliminaties}

\paragraph*{Notation.}
For an integer $n \ge 1$, let $[n]$ denote $\{1, 2, \ldots, n\}$. 
For a predicate $P$ let $[P]$ have the value~$1$ when $P$ is true and the value~$0$ when $P$ is false.
Let $\text{tail}_{\countsketchrange}(x)$ denote the vector derived from $x$ by changing the $\countsketchrange$ entries with largest absolute value to zero. We use bold symbols such as $\bh, \bx, \bN, \ldots$ to denote that these objects are explicitly sampled from an appropriate distribution.

For a real-valued matrix $M$ define  $\|M\|= \sup_{x\ne 0}\|x M\|_{1}/\|x\|_{1}$ where $\|x\|_{1} = \sum_i |x_i|$. 
(This matrix norm is sometimes written $\|M\|_1$, but since we do not need other matrix norms we omit the subscript.)

\paragraph{Model of Computation.} All of our running times are in the Word RAM model with a word size $O(\log d)$ unless otherwise mentioned. We assume that all elementary operations on words can be performed in $O(1)$ time.

\begin{definition}[$k$-wise independence]
	A family of hash functions $\calH = \set{h : [u] \rightarrow [v]}$ is said to be $k$-wise independent if for any $k$ distinct keys $x_1, x_2,\ldots,x_k \in [u]$ and not necessarily distinct values $y_1,y_2,\ldots,y_k \in [v]$ 
\begin{align*}
	\Pr_{\bh \sim \calH}[\bh(x_1) = y_1\land \cdots \land \bh(x_k) = y_k] = \frac{1}{v^k}.
\end{align*}
\end{definition}
The definition states that if $\bh \sim \calH$, then for any $x \in [u]$, the random variable $\bh(x)$ is uniformly distributed over $[v]$ and for any $k$ distinct keys $x_1, \ldots, x_k \in [u]$, the random variables $\bh(x_1), \ldots, \bh(x_k)$ are independent. 

We now state a randomized construction of a hash family from \cite{CPT15} that lets us evaluate the sampled hash function quickly on any input. 
 \begin{theorem}[Corollary 3 in \cite{CPT15}]
 There exists a randomized data structure that takes as input positive integers $u, v = u^{O(1)}, t, k = u^{O(1/t)}$ and  selects a family of function $\calF$ from $[u]$ to $[v]$. In the Word RAM model with word length $\Theta(\log u)$ the data structure satisfies the following:
 \begin{enumerate}
 	\item The space used to represent the family $\calF$ as well as a function $f \in \calF$ is $O(ku^{1/t}t)$ bits.
 	\item The evaluation time of any function $f \in \calF$ on any input is $O(t\log t)$.
 	\item With probability $\ge 1 - u^{-1/t}$, we have that $\calF$ is a $k$-wise independent family.
 \end{enumerate}
 \label{thm:fast-hash-functions}
 \end{theorem}
 Throughout the paper, we use this construction with a constant $t$ and $k$ at most $O(\log u)$.
 
 We now state the guarantees of the hash family construction from \cite{pagh2008uniform}. Whereas the above construction gives a randomized hash family $\calF$ that is $k$-wise independent with some probability, the following construction gives a randomized hash family $\calH$ that when restricted to any fixed subset $S$ of a certain size is a uniform hash family with some probability. We use this construction when we want to ensure that all elements of a small underlying set (but unknown to the streaming algorithm) are hashed to uniformly random locations.
\begin{theorem}[Theorem~1 of \cite{pagh2008uniform}]
Let $S \subseteq U = [u]$ be a set of $z > 1$ elements and let $V = [v]$ for any $1 \le v \le u$. Suppose the machine word size is $\Omega(\log u)$. For any constant $C > 0$, there is a Word RAM algorithm that, using $\log(z)(\log v)^{O(1)}$ time and $O(\log z + \log\log u)$ bits of space, selects a family of hash functions $\calH$ from $U$ to $V$ (independent of $S$) such that
\begin{itemize}
    \item $\calH$ is $z$-wise independent (in other words, uniform) when restricted to $S$, with probability $1 - O(1/z^C)$.
    \item Any function $h \in \calH$ can be represented using $O(z\log v)$ bits and $h$ can can be evaluated on any $x \in U$ in $O(1)$ time in the Word RAM model. The data structure (or representation) of a random function from the family $\calH$ can be constructed in $O(z)$ time.
\end{itemize}
\label{thm:pagh-and-pagh}
\end{theorem}
We now define $\varepsilon$ pseudorandom generators for any given class of algorithms $\calC$.
\begin{definition}
    A generator $G : \set{0,1}^{n} \rightarrow \set{0,1}^m$ is called an $\varepsilon$ pseudorandom generator for the class of algorithms $\calC$ if for any $C \in \calC$,
    \begin{align*}
        |\Pr_{\bx \sim U_{m}}[C(\bx)\, \text{accepts}] - \Pr_{\by \sim U_n}[C(G(\by))\, \text{accepts}]| < \varepsilon.
    \end{align*}
\end{definition}
Here $U_m$ denotes the uniform distribution over $\set{0,1}^m$. For a generator $G$, the quantity $n$ is called the \emph{seed length}. In this work $\calC$ is taken to be the set of space bounded algorithms with an appropriate space parameter $s$. The algorithms have a read/write space of $s$ bits and a read only tape that contains an input. The algorithms are allowed to only stream over the $m$ length random/pseudorandom string. Although the above definition is in terms of accept/reject, it can extended to general functions by instead considering the total variation distance between the distributions of the outputs.

\section{\FastPRG}\label{sec:hashprg}

In this section we present a new pseudorandom generator for space-bounded computations, which is going to be our main tool for derandomizing streaming algorithms.
The starting point is the classical generator of Nisan \cite{nisan} (summarized in Appendix~\ref{sec:nisan}), which we extend to provide a trade-off between seed length and the time to compute an arbitrary output block.
To make our treatment easy to access for readers familiar with Nisan's generator, we will follow the same proof outline, but make crucial changes to avoid a union bound over all possible intermediate states.
An advantage of our generator over Nisan's, even in the setting where the seed lengths are the same, is that it has a certain \emph{symmetry} property that we need in our applications.

\subsection{\FastPRG Construction}
Let $b\geq 2$ be an integer.
Consider $\bh_0,\ldots,\bh_{k-1}$ where $\bh_i := (\bh_i^{(0)},\ldots,\bh_i^{(b-1)})$ is a vector of $b$ hash functions with each $\bh_i^{(j)} : \set{0,1}^n \rightarrow \set{0,1}^n$ being a hash function drawn independently from a 2-wise independent hash family $\calH$. We slightly abuse terminology and also refer to the vectors $\bh_i$ as hash functions.

Using hash functions $\bh_0,\ldots,\bh_{k-1}$, define a generator $G_k : \set{0,1}^n \rightarrow \set{0,1}^{n \cdot b^k}$ recursively:
\begin{align*}
    G_0(x) &= x\\
    G_k(x, \bh_0,\ldots,\bh_{k-1}) &=
     G_{k-1}(\bh_{k-1}^{(0)}(x), \bh_0,\ldots, \bh_{k-2}) \circ \cdots \circ G_{k-1}(\bh^{(b-1)}_{k-1}(x), \bh_0,\ldots,\bh_{k-2}).
\end{align*}
Here $\circ$ denotes concatenation. Note that Nisan's generator can be obtained by setting $b = 2$ and deterministically setting $\bh_i^{(0)}(x) := x$ for all $i,x$.

By construction we have, for $x \in \set{0,1}^n$, that $G_k(x,\bh_0,\ldots,\bh_{k-1})$ is a bitstring of length $n \cdot b^k$. We look at the string $G_k(x, \bh_0,\ldots,\bh_{k-1})$ as concatenation of $b^k$ chunks each of length $n$, chunks indexed by $0,\ldots, b^k-1$. Let $j \in \set{0, \ldots, b^k-1}$ be written as $j_{k-1} \cdots j_0$ in base~$b$. Then the $j$th chunk of the string $G_k(x,\bh_0,\ldots,\bh_{k-1})$ is given by
\begin{align*}
    \bh_0^{(j_0)}(\bh_1^{(j_1)}(\cdots \bh_{k-1}^{(j_{k-1})}(x))).
\end{align*}
To define the power of our pseudorandom generator we need the following notation. 
Let $Q$ be an arbitrary finite state machine with $2^{w}$ states over the alphabet $\set{0,1}^n$. Let $D$ be any distribution over the strings of length $b^{k} \cdot n$, encoding $b^k$ steps of the FSM. Let $Q(D)$ be a $2^{w} \times 2^{w}$ matrix where $[Q(D)]_{ij}$ is the probability that the FSM starting in state $i$ goes to state $j$ after performing $b^{k}$ steps based on an input drawn from $D$.
Let $U_n$ denote the uniform distribution over $\set{0,1}^n$. We will prove the following lemma.

\begin{lemma}\label{lma:generalized-nisan}
There exists a constant $c > 0$, given integers $n$ and $w \le cn$ and parameters $b,k$ with $b^{k} \le 2^{cn}$, for any FSM $Q$ with at most $2^{w}$ states, if $\bh_0,\ldots,\bh_{k-1} : \set{0,\ldots,b-1} \rightarrow \set{0,1}^n \rightarrow \set{0,1}^n$ are drawn independently from ${\calH}^b$, where $\calH$ is any family of $2$-wise independent hash functions, then with probability $\ge 1 - 2^{-cn}$ (over the draw of $\bh_0,\ldots,\bh_{{k-1}}$),
\begin{align*}
    \|Q(G_k(*, \bh_0,\ldots,\bh_{k-1})) - Q((U_n)^{b^k})\| \le 2^{-cn}.
\end{align*}
Here $G_k(*, \bh_0,\ldots,\bh_{k-1})$ denotes uniform distribution over the set $\setbuilder{G_k(x, \bh_{0},\ldots,\bh_{k-1})}{x \in \set{0,1}^n}$.
\end{lemma}
By definition of the matrix norm $\|\cdot\|$ (see section~\ref{sec:preliminaties}) this implies that
that with probability at least $1 - 2^{-cn}$ over the hash functions $\bh_0,\ldots,\bh_{k-1}$, we have that the total variation distance between the distribution of final state using a random string drawn from $(U_n)^{b^k}$ and a random string drawn from $G_k(*, \bh_0,\ldots,\bh_{k-1})$ is at most $2^{-cn}$.

Using the above lemma, we will then prove the following theorem which we will use through out the paper.
\begin{theorem}
    There exists a constant $c > 0$ such that given any parameters $n$, $b$ and $k$ satisfying $b^{k} \le 2^{cn}$, there exists a generator which we call $\FastPRG : \set{0,1}^{O(bkn)} \rightarrow \set{0,1}^{b^k \cdot n}$ such that \FastPRG is an $O(2^{-cn})$ pseudorandom generator for the class of Finite State Machines over alphabet $\set{0,1}^n$ with at most $2^{cn}$ states. For any seed $\br$, the $i$-th block of bits in $\FastPRG(\br)$ can be computed in time $O(k)$ on a machine with Word Size $\Omega(n)$.
    \label{thm:nisan-main-theorem}
\end{theorem}
\begin{proof}
    Let $p_{i,j}^{h_0, \ldots, h_{k-1}}$ be $[Q(G_k(*, h_0,\ldots,h_{k-1}))]_{i,j}$. Let $p_{i,j}$ be the probability the FSM starting in state $i$ goes to a state $j$ after $b^k$ steps on an input $\br$ where $\bh_0,\ldots,\bh_{k-1} \sim \calH^{b}$ and $\br \sim G_k(*, \bh_0,\ldots,\bh_{k-1})$. Clearly, for any pair of states $i,j$,
    \begin{align*}
        p_{i,j} = \frac{1}{|\calH^{b}|^k}\sum_{(h_0,\ldots,h_{k-1}) \in (\calH^{b})^{k}} p_{i,j}^{h_0,\ldots,h_{k-1}}.
    \end{align*}
    Let $q_{i,j} = [Q((U_n)^{b^k})]_{i,j}$. Now, for any $i$,
    \begin{align*}
        \sum_j |p_{i, j} - q_{i,j}|
        &\le \sum_j \frac{1}{|\calH^b|^k}\sum_{(h_0,\ldots,h_{k-1}) \in (\calH^{b})^k}|p_{i,j}^{h_0,\ldots,h_{k-1}} - q_{i,j}|\\
        &= \frac{1}{|\calH^{b}|^k}\sum_{(h_0,\ldots,h_{k-1}) \in (\calH^{b})^k} (\sum_j |p_{i,j}^{h_0,\ldots,h_{k-1}} - q_{i,j}|).
    \end{align*}
    By Lemma~\ref{lma:generalized-nisan}, $|\setbuilder{(h_0,\ldots,h_{k-1})}{\sum_j |p_{i,j}^{h_0,\ldots,h_{k-1}} - q_{i,j}| \le 2^{-cn}}| \ge |\calH^{b}|^k (1 - 2^{-cn})$ and using the fact that for any $(h_0,\ldots,h_{k-1})$, $|p_{i,j}^{h_0,\ldots,h_{k-1}} - q_{i,j}| \le 1$, we obtain that
    \begin{align*}
        \sum_j |p_{i,j} - q_{i,j}| \le 2^{-cn} + 1 \cdot 2^{-cn} \le 2 \cdot 2^{-cn}.
    \end{align*}
    As the above holds for any $i$, we obtain that $\|Q(\FastPRG) - Q((U_n)^{b^k})\| \le 2 \cdot 2^{-cn}$. Where we overload the notation $\FastPRG$ to also denote the distribution of string $\bgamma$ obtained by first sampling $\bh_0,\ldots,\bh_{k-1}$ and then sampling $\bgamma \sim G_k(*, \bh_0,\ldots,\bh_{k-1})$. The $O(bkn)$ bits of random seed is used to sample $bk$ independent hash functions from the 2-wise independent hash family construction of Dietzfelbinger \cite[Theorem~3(b)]{Dietzfelbinger96universal}. The $bk$ independent hash functions are used to construct the hash functions $\bh_0,\ldots,\bh_{k-1}$ and $n$ bits from the random string are used to get an $\bx \sim \set{0,1}^n$. The output of \FastPRG is then $G_k(\bx, \bh_0,\ldots,\bh_{k-1})$. Clearly, by the definition of $G_k$, any block of bits in the output string can be computed in $O(k)$ time.
\end{proof}
A special case of the above theorem we use is when $n = O(\log d)$. By choosing $b = d^{\epsilon}$ and $k = O(1/\epsilon)$ for a small enough constant $\epsilon$, we obtain that \FastPRG with these parameters fools any FSM with $\poly(d)$ states over an alphabet $\set{0,1}^{n}$. And that on a machine with word size $\Omega(\log d)$, any block of bits in the output of \FastPRG can be computed in time $O(1)$ since $\epsilon$ is a constant.

\subsection{\FastPRG Analysis}

We will be reasoning about matrices that represent transition probabilities of $Q$ under different input distributions.
We start with some simple facts about norms of matrices.
Recall that $\|M\| \coloneqq \sup_{x :  \|x\|_{\infty} = 1}\|Mx\|_{\infty}$.
First, for any two matrices $A$ and $B$ we have $\|A+B\| \le \|A\|+\|B\|$ and $\|AB\| \le \|A\|\|B\|$. 
We will also need the following lemma about the norm of differences of matrix powers:
\begin{lemma}\label{lem:power-difference}
For integer $b\ge 1$ and square real matrices $M$ and $N$ with $\|M\|\le 1$ and $\|N\| \le 1$, $\|M^b - N^b\| \le b\|M-N\|$.
\end{lemma}
\begin{proof}
The proof is by induction on $b$. The statement clearly holds for $b=1$. 
For the induction step we have:
\begin{align*}
\|M^b - N^b\|
& \leq \|M^b - N^{b-1}M \| + \| N^{b-1}M - N^b\|\\
& = \|(M^{b-1}-N^{b-1})M\| + \|N^{b-1}(M-N)\|\\
& \le \|M^{b-1} - N^{b-1}\|\|M\| + \|N\|^{b-1}\|M-N\| \\
&\le (b-1)\|M-N\| + \|M-N\| = b\|M-N\| \enspace .
\end{align*}
This first uses triangle inequality, the second inequality uses that the norm of a product is bounded by the product of the norms, and the third inequality uses the induction hypothesis as well as the assumption $\|M\|\le 1$, $\|N\| \le 1$.
\end{proof}

For fixed hash functions $h_0,\ldots,h_{{k-1}}$, let $G_k(*,h_0,\ldots,h_{{k-1}})$ be the distribution of $G_k(\bx,h_0,\ldots,h_{{k-1}})$ when $\bx$ is drawn uniformly at random from $\set{0,1}^n$, denoted $\bx\sim U_n$. Let $(U_n)^{b^k}$ denote the uniform distribution over length $n \cdot b^k$ bitstrings. The aim is to show, akin to Nisan's generator, for any ``small-space computation'', with high probability over $\bh_0,\ldots,\bh_{{k-1}}$, the distribution $G_k(*, \bh_0,\ldots,\bh_{{k-1}})$ is indistinguishable from the uniform distribution $(U_n)^{b^k}$.

Lemma~\ref{lma:generalized-nisan} will be derived from this result:
For any $\varepsilon > 0$, if all $\bh_i^{(j)}$ are drawn independently from a 2-wise independent hash family,
\begin{align}
    \Pr\left[\|Q(G_k(*, \bh_0,\ldots,\bh_{{k-1}})) - Q((U_n)^{b^k})\| > \varepsilon \right]
    &\le \frac{(b^k-1)^2 k}{(b-1)^2\varepsilon^{2}} \cdot 2^{3w - n} \enspace .\label{eq:errorprob}
\end{align}

Note that  each $\bh_i$ can also be seen as function from $\set{0,1}^n$ to $\set{0,1}^{bn}$ defined as $\bh_i(x) = \bh_i^{(0)}(x) \circ \cdots \circ \bh_i^{(b-1)}(x)$.
We say $\bh_i : \set{0,1}^{n} \rightarrow \set{0,1}^{bn}$ is drawn from the hash family ${\calH}^b$.
Since each $\bh_i^{(j)}$ is sampled independently from a $2$-wise independent hash family, ${\calH}^b$ is a also $2$-wise independent. 

\begin{definition}\label{def:set-independence}
Let $A  \subseteq \set{0,1}^{bn}$, $h : \set{0,1}^n \rightarrow \set{0,1}^{bn}$, and $\varepsilon > 0$. We say that $h$ is \emph{$(\varepsilon, A)$-independent} if 
    $|\Pr_{\bx\sim U_n}[h(\bx) \in A] - \varrho(A)| \le \varepsilon$,
where $\varrho(A) := {|A|}/{2^{bn}}$ is the density of  set $A$.
\end{definition}
The above definition corresponds to Definition~3 in \cite{nisan} but with some important differences. The differences leads to subtle changes in the whole analysis, but the overall structure of the analysis remain the same.

We now have the following lemma that corresponds to Lemma~1 of \cite{nisan}.
\begin{lemma}
Let $A \subseteq \set{0,1}^{bn}$ and $\varepsilon > 0$. Then
    $\Pr_{\bh \sim \calH^{b}}[\text{$\bh$ is not $(\varepsilon, A)$-independent}] < \frac{\varrho(A)}{\varepsilon^2 2^n}$.
\label{lma:independence-of-hash-function}
\end{lemma}
\begin{proof}
Consider a matrix $M$ that has a row for each $x \in \set{0,1}^n$ and a column for each $h \in {\calH}^b$. Let $M(x,h) = 1$ if $h(x) \in A$ and $0$ otherwise. We define the function $f$ that expresses the probability that a function $h$ maps $\bx \sim U_n$ to a value in $A$:
\begin{align*}
    f(h) = \E_{\bx \sim U_n}[M(\bx,h)] = \Pr_{\bx \sim U_n}[h(\bx) \in A].
\end{align*}
For $\bh \sim {\calH}^b$ we have
\begin{align*}
    \E_{\bh \sim {\calH}^b} f(\bh) &= \E_{\bh \sim {\calH}^b,\,\bx \sim U_n}[ M(\bx,\bh)] = \E_{\bx \sim U_n} \Pr_{\bh \sim {\calH}^b}[\bh(\bx) \in A] = \varrho(A),
\end{align*}
where the last equality follows from $\Pr_{\bh \sim {\calH}^b}[\bh(\bx) \in A] = \varrho(A)$, since $\calH^{b}$ is $2$-wise independent. 
We next bound the variance of $f(\bh)$ to show that $f(\bh)$ is close to $\varrho(A)$ with high probability:
\begin{align*}
    \Var_{\bh \sim {\calH}^b} f(\bh) &= \E_{\bh \sim {\calH}^b} (f(\bh) - \varrho(A))^2\\
    &= \E_{\bh \sim {\calH}^b}(\E_{\bx \sim U_n} [M(\bx,\bh) - \varrho(A)])^2.
\end{align*}
If $\bx_1 \sim U_n$ and $\bx_2 \sim U_n$ are independently drawn we can expand
\begin{align*}
    &(\E_{\bx \sim U_n}[M(\bx,h) - \varrho(A)])^2\\
    &= \E_{\bx_1 \sim U_n,\, \bx_2 \sim U_n} (M(\bx_1, h) - \varrho(A)) (M(\bx_2, h) - \varrho(A)).
\end{align*}
Using the fact that for every $x$, $\E_{\bh \sim {\calH}^b}[M(x,\bh)] = \varrho(A)$ we obtain
\begin{align*}
    \Var_{\bh \sim {\calH}^b} f(\bh) = \E_{\bx_1 \sim U_n,\, \bx_2 \sim U_n}\E_{\bh \sim {\calH}^b}[M(\bx_1, \bh)M(\bx_2, \bh) - \varrho(A)^2].
\end{align*}
With probability $1 - 2^{-n}$ we have $\bx_1 \ne \bx_2$. Since $\bh$ is drawn from a $2$-wise independent hash family, $\bh(\bx_1)$ and $\bh(\bx_2)$ are independent in this case, so $\E_{\bh \sim {\calH}^b}[M(\bx_1,\bh)M(\bx_2,\bh) \mid \bx_1\ne\bx_2 ] = \varrho(A)^2$. With probability $1/2^n$, we have $\bx_1 = \bx_2$ and $\E_{\bh \sim {\calH}^b}[M(\bx_1, h)M(\bx_2, h)  \mid \bx_1 = \bx_2 ] = \E_{\bh \sim {\calH}^b}[M(\bx_1, \bh)] = \varrho(A)$. Hence,
\begin{align*}
    \Var_{\bh \sim {\calH}^b} f(\bh) = (1 - 2^{-n}) \varrho(A)^2 + 2^{-n}\varrho(A) - \varrho(A)^2 < \frac{\varrho(A)}{2^n} \enspace .
\end{align*}
By Chebyshev's inequality, 
\begin{align*}
    \Pr_{\bh \sim {\calH}^b}[|f(\bh) - \varrho(A)| \ge \varepsilon] \le \frac{\varrho(A)}{2^n \varepsilon^2}. & \qedhere
\end{align*}
\end{proof}
We now have the following definition that corresponds to Definition 4 of \cite{nisan}.
\begin{definition}\label{def:good}
Let $Q$ be a FSM on alphabet $\set{0,1}^n$ and pick $\varepsilon > 0$. For hash functions $h_0,\ldots,h_{{k-1}} : \set{0,1}^n \rightarrow \set{0,1}^{bn}$, we say $(h_0,\ldots,h_{{k-1}})$ is \emph{$(\varepsilon, Q)$-good} if
    $\|Q(G_k(*, h_0,\ldots, h_{{k-1}})) - Q((U_n)^{b^k})\| \le \varepsilon$.
\end{definition}
For small $\varepsilon$ the definition essentially says that the FSM $Q$ cannot distinguish between the distributions $G_k(*, h_0,\ldots,h_{{k-1}})$ and $(U_n)^{b^k}$. The following lemma corresponds to Lemma~2 of \cite{nisan}.
\begin{lemma}\label{lem:good-bound}
Let $Q$ be a FSM of size $2^w$ over alphabet $\set{0,1}^n$ and $k$ a nonnegative integer. Then for every $\varepsilon > 0$,
\begin{align*}
    &\Pr_{\bh_0,\ldots,\bh_{{k-1}} \sim {\calH}^b}[\text{$(\bh_0,\ldots,\bh_{{k-1}})$ is not $((b^k - 1)\varepsilon, Q)$-good}]\\
    &\le \frac{2^{3w}k}{(b-1)^2\varepsilon^2 2^n}.
\end{align*}
\end{lemma}
\begin{proof}
The proof is a careful translation of the proof of Lemma~2 of \cite{nisan}. The proof is by induction on $k$. For $k=0$, the statement is immediate since $G_0(x)\sim U_n$ which means that there is zero difference between the distributions. For the induction step assume that the statement holds for $k-1$. For every choice of $h_0,\ldots,h_{{k-2}}: \set{0,1}^n \rightarrow \set{0,1}^{bn}$ and every two states $i,j$ of~$Q$ we define the set of seeds $y^{(0)},\ldots,y^{(b-1)}$ that when used one by one with $G_{k-1}$ produces a string that takes $Q$ from state $i$ to state $j$:
\begin{align*}
    B_{i,j}^{h_0,\ldots,h_{{k-2}}}
    &= \{(y^{(0)},\ldots,y^{(b-1)}) \in \set{0,1}^{bn} \mid  G_{k-1}(y^{(0)},h_0,\ldots,h_{{k-2}}) \circ \cdots  \circ G_{k-1}(y^{(b-1)},h_0,\ldots,h_{{k-2}})\, \\
    &\qquad \text{takes $Q$ from $i$ to $j$}\}.
\end{align*}
Now sample $\bh_0,\ldots,\bh_{{k-1}}\sim {\calH}^b$ independently and consider the following events:
\begin{enumerate}
    \item $(\bh_0,\ldots,\bh_{{k-2}})$ is $((b^{k-1}-1)\varepsilon, Q)$-good (see Definition~\ref{def:good}), and
    \item $\bh_{{k-1}}$ is $(\varepsilon(b-1)/2^w, B_{i,j}^{\bh_0,\ldots,\bh_{{k-2}}})$-independent for all states $i,j$ (see Definition~\ref{def:set-independence}).
\end{enumerate}
We will see that these events happen simultaneously with probability at least $1 - \frac{2^{3w}k}{(b-1)^2\varepsilon^2 2^n}$.
Furthermore, we will show that when this happens, $(\bh_0,\ldots,\bh_{{k-1}})$ is $((b^k-1)\varepsilon, Q)$-good, completing the induction step.

\medskip

By the induction hypothesis, the probability of event 1 \emph{not} happening is at most $\frac{2^{3w}(k-1)}{(b-1)^2\varepsilon^2 2^n}$.
Lemma~\ref{lma:independence-of-hash-function} tells us that for every choice of $h_0,\ldots,h_{{k-2}}$ and states $i,j$ of $Q$, the probability that $\bh_{{k-1}}$ is \emph{not} $(\varepsilon(b-1)/2^w, B_{i,j}^{h_0,\ldots,h_{{k-2}}})$-independent is at most
\begin{align*}
    \frac{\varrho(B_{i,j}^{h_0,\ldots,h_{{k-2}}})2^{2w}}{\varepsilon^2 (b-1)^2 2^n} \enspace .
\end{align*}
Using a union bound, the probability there exists a pair of states $i$ and $j$ such that $\bh_{{k-1}}$ is not $(\varepsilon(b-1)/2^{w}, B_{i,j}^{h_0,\ldots,h_{{k-2}}})$-independent is at most
\begin{align*} &\sum_{i,j}\frac{\varrho(B_{i,j}^{h_0,\ldots,h_{{k-2}}})2^{2w}}{\varepsilon^2 (b-1)^2 2^n}= \frac{2^{2w}}{\varepsilon^2(b-1)^2 2^n}\sum_{i,j} \varrho(B_{i,j}^{h_0,\ldots,h_{{k-2}}})\le \frac{2^{3w}}{\varepsilon^2(b-1)^2 2^n} \enspace .
\end{align*}
The last inequality follows from the fact that for every fixed $i$, the sets ${B_{i,j}^{h_0,\ldots,h_{{k-2}}}}$, where $j$ ranges over the states of $Q$, partition $\set{0,1}^{bn}$ and thus their $\varrho$-values sum up to~1.
A union bound shows us that the probability of either event 1 or 2 not holding is at most $2^{3w}k/(b-1)^2\varepsilon^2 2^n$, as claimed.

Let $(G_{k-1}(*, \bh_0,\ldots,\bh_{{k-2}}))^b$ denote the distribution over bitstrings of length $n \cdot b^k$ obtained by concatenating $b$ \emph{independent} samples drawn from $G_{k-1}(*, \bh_0,\ldots,\bh_{{k-2}})$.
Conditioning on events 1 and 2 we now bound $\|Q(G_k(*, \bh_0,\ldots,\bh_{{k-1}}) - Q((U_n)^{b^k})\|$. Using the triangle inequality,
\begin{align*}
    \|Q(G_k(*, \bh_0,\ldots,\bh_{{k-1}})) - Q((U_n)^{b^k})\|& \le \|Q(G_{k-1}(*, \bh_0,\ldots,\bh_{k-1})) - Q((G_{k-1}(*,\bh_0,\ldots,\bh_{{k-2}}))^b)\|\\
    &\quad + \|Q((G_{k-1}(*,\bh_0,\ldots,\bh_{{k-2}}))^b) - Q((U_n)^{b^k})\|.
\end{align*}
Below we will bound the first term by $(b-1)\varepsilon$ and second term by $(b^{k}-b)\varepsilon$, proving that the hash functions are $((b^k-1)\varepsilon, Q)$-good as claimed.

Consider the matrix $Q(G_k(*, \bh_0,\ldots,\bh_{{k-1}}))$. By definition, entry $(i,j)$ of the matrix is equal to
    $\Pr_{\bx \sim U_n}[\bh_{{k-1}}(\bx) \in B_{i,j}^{\bh_0,\ldots,\bh_{{k-2}}}]$.
Now consider $Q((G_{k-1}(*,\bh_0,\ldots,\bh_{k-2}))^b)$ where entry $(i,j)$ is
\begin{align*}
    \Pr_{\by^{(0)},\ldots,\by^{(b-1)} \sim U_n}  [G_{k-1}(\by^{(0)},\bh_0,\ldots,\bh_{k-2}) \circ \cdots \circ G_{k-1}(\by^{(b-1)}, \bh_0,\ldots,\bh_{k-2}) \text{ takes $Q$ from $i$ to $j$]} \enspace .
\end{align*}
By definition, the above quantity is exactly $\varrho(B_{i,j}^{\bh_0,\ldots,\bh_{k-2}})$. Since event $2$ holds, we have that the absolute value of every entry of the matrix $Q(G_k(*, \bh_0,\ldots,\bh_{k-1})) - Q((G_{k-1}(*,\bh_0,\ldots,\bh_{k-2}))^b)$ is at most $\varepsilon(b-1)/2^w$. Since each row has at most $2^w$ entries we obtain
\begin{align*}
    \|Q(G_k(*, \bh_0,\ldots,\bh_{k-1})) - Q((G_{k-1}(*,\bh_0,\ldots,\bh_{k-2}))^b)\|\le (b-1)\varepsilon \enspace .
\end{align*}
To bound the second term we define the transition matrices $M = Q(G_{k-1}(*,\bh_0,\ldots,\bh_{k-2}))$ and $N = Q((U_n)^{b^{k-1}})$ that describe transition probabilities using $G_{k-1}$ and uniform inputs, respectively.
Since the $b$ parts of the input are independent we can express the matrices in the second term as powers of $M$ and $N$:
\begin{align*}
    Q((G_{k-1}(*,\bh_0,\ldots,\bh_{k-1}))^b) = M^b
\quad \text{and} \quad
    Q((U_n)^{b^k}) = N^b \enspace .
\end{align*}

We can now invoke Lemma~\ref{lem:power-difference}.
Since event 1 holds we have $\|M-N\| \le (b^{k-1}-1)\varepsilon$, and thus
\begin{align*}
    \|Q((G_{k-1}(*,\bh_0,\ldots,\bh_{k-2}))^b) - Q((U_n)^{b^k})\| \le b(b^{k-1}-1)\varepsilon = (b^k-b)\varepsilon.
\end{align*}
Conditioning on events 1 and 2 we thus have
    $\|Q(G_k(*, \bh_0,\ldots,\bh_{{k-1}})) - Q((U_n)^{b^k})\| \le (b^k - 1)\varepsilon$. \qedhere
\end{proof}

The proof of Lemma~\ref{lma:generalized-nisan} now follows by choosing a small constant $c$ and setting $w = cn$ and $\varepsilon = {2^{-cn}}/{(b^{k}-1)}$ in the above lemma. From Proposition 1 of \cite{nisan}, it follows that any space($cn$) algorithm that uses $nb^k$ uniform random bits, for $b^k \le 2^{cn}$, can use the bits from the pseudorandom generator and with probability at least $1 - 2^{-cn}$, the total variation distance of the final state of the algorithm using pseudorandom bits to the final state of the algorithm using uniform random bits is at most $2^{-cn}$. 
\section{Moment Estimation for \texorpdfstring{$p > 2$}{p > 2}}
Using min-stability of exponential random variables, Andoni \cite{andoni2017high} gave an algorithm (Algorithm~\ref{alg:andoni}) to estimate $F_p$ moments in the stream. We state their result below and describe their algorithm. 

Consider the vector $x$ obtained at the end of the stream applying all the updates sequentially to the starting vector $0$. As Algorithm~\ref{alg:andoni} is linear, we have that the final state of the algorithm depends only on the final vector and not on the order of the updates. To estimate the $\ell_p$ norm of a $d$-dimensional vector $x$, the algorithm first samples $d$ independent standard exponential random variables $\bE_1,\ldots,\bE_d$ and creates a random vector $\bz \in \R^d$ such that $\bz_i = \bE_i^{-1/p} x_i$. Andoni shows that $\linf{\bz} = \Theta(\lp{x})$ and therefore estimating the value of the coordinate in $\bz$ with maximum absolute value gives a constant factor approximation for $\lp{x}$. Andoni further shows that with high probability there are only at most $O(\log d)^p$ coordinates in $\bz$ with absolute value $\ge \lp{x}/(c\log d)$ for a constant $c$ hence showing that the vector $\bz$ has only a few coordinates with large values. The final property of $\bz$ that Andoni uses is that with high probability
    $\opnorm{\bz}^2 = O(d^{1-2/p}\lp{x}^2)$.

Conditioned on the above properties of $\bz$, Andoni argues that if $\bS$ is a CountSketch matrix with $O(d^{1-2/p}\log d)$ rows formed using $O(\log d)$-wise independent hash functions, then with high probability $\linf{\bS \bz} = \Theta(\linf{\bz}) = \Theta(\lp{x})$. To implement the streaming algorithm in sub-linear space we cannot store the exponential random variables and the vector $\bz$ in the stream. So Andoni derandomizes his algorithm by using a pseudorandom generator of Nisan and Zuckerman \cite{nisan-zuckerman} and shows that exponential random variables $\bE_i$ generated using the pseudorandom bits are sufficient to ensure that the algorithm produces a constant factor approximation to $\lp{x}$. Also note that the algorithm does not need to explicitly store the $d$ dimensional vector $\bz$; it is enough just to update the $O(d^{1-2/p}\log d)$ dimensional vector $\boldf$ in the stream.

Although Andoni's algorithm is space optimal for linear sketches up to constant factors, the update time using the Nisan-Zuckerman pseudorandom generator is $\poly(d)$. We now show that Andoni's algorithm can be derandomized using our \FastPRG to obtain an algorithm that is both space optimal for linear sketches and has an update time of $O(1)$ in the Word RAM model with a word size $\Omega(\log d)$.
\begin{algorithm}
\KwIn{$p > 2$, $d \in \N$, a stream of updates $(i_1,v_1), \ldots, (i_m, v_m) \in [d] \times \set{-M, \ldots, M}$ for $m, M = \poly(d)$}
\KwOut{An approximation to $\lp{x}$ where $x \in \R^d$ is defined by the stream of updates}	
$\bE_1,\ldots,\bE_d \gets $ Independent standard exponential random variables\;
$d' \gets O(d^{1-2/p}\log d)$\;
$\bz \gets 0_d$, $\boldf \gets 0_s$\;
$\bh \gets $ $O(\log d)$-wise independent hash function from $[d]$ to $[d']$\;
$\bsigma \gets $ $O(\log d)$-wise independent hash function from $[d]$ to $\set{-1,+1}$\;
\For{$j = 1,\ldots,m$}{
	$\bz_{i_j} \gets \bz_{i_j} + \bE_i^{-1/p}v_j$\;
	$\boldf_{\bh(i_j)} \gets \boldf_{\bh(i_j)} + \bsigma(i)\bE_i^{-1/p}v_j$\;
}
\Return{$\linf{\boldf}$}\;
\caption{Andoni's Algorithm with Independent Exponentials \cite{andoni2017high}}
\label{alg:andoni}
\end{algorithm}

\subsection{Discretizing the Exponentials}
We first show that we can replace exponential random variables in Andoni's algorithm with a simple discrete random variable and obtain all the guarantees we stated above that $\bz$ satisfies even with this discrete random variable. For now, assume that $p=1$. We will later generalize the guarantees for all $p$.

Suppose $x \in \R^d$ with all the coordinates of $x$ being integers with absolute values $\le \poly_1(d)$. We can assume $x$ has no nonzero coordinates. Consider the discrete random variable $\bE$ that takes values in the set $\set{1, 2, \ldots, 2^{M}}$ for some $M = O(\log d)$ satisfying $2^M \ge d\poly_1(d)$. Let
\begin{align*}
    \Pr[\bE = 2^j] = \begin{cases}
                        1/2^{j+1} & 0 \le j < M\\
                        1/2^M & j = M.
                    \end{cases}
\end{align*}
We call the random variable $\bE$ a \emph{discrete exponential}\footnote{While $\bE$ models the inverse of a continuous exponential random variable, we use the term ``discrete exponential'' for brevity.}. Note that we can sample the random variable $\bE$ quite easily from a uniform random bitstring of length $M$ just based on the position of the first appearance of $1$ in the random bitstring. We mainly use the following properties of $\bE$: for any $t > 0$, $\Pr[\bE \ge t] \le 1/t$. The statement clearly holds for $t \le 1$. For $t \ge 1$, let $2^j$ be such that $t \le 2^j < 2t$. Now,
\begin{align*}
    \Pr[\bE \ge t] &= \Pr[\bE = 2^j] + \Pr[\bE = 2^{j+1}] + \cdots\\
    &= \frac{1}{2^{j+1}} + \frac{1}{2^{j+2}} + \cdots = \frac{1}{2^j} \le \frac{1}{t}.
\end{align*}
Similarly for any $t \le 2^M$, we have $\Pr[\bE \ge t] \ge \min(1,1/(2t))$.
We now prove the following lemma.
\begin{lemma}
Let $x \in \R^d$ be an arbitrary vector with integer entries of absolute value at most $\poly_1(d)$. Let $\bE_1,\ldots,\bE_d$ be independent \emph{discrete exponential} random variables. Then,
\begin{enumerate}
    \item With probability $\ge 95/100$,
    \begin{align*}
        \frac{\lone{x}}{16} \le \max_i |x_i|\bE_i \le 50\lone{x}.
    \end{align*}
    \item For any $T$,
    \begin{align*}
        \E[ |\setbuilder{i}{|x_i|\bE_i \ge \lone{x}/T}|] \le T.
    \end{align*}
\end{enumerate}
\end{lemma}
\begin{proof}
Without loss of generality, we can assume that all the coordinates of $x$ are nonzero. Using the distribution of the random variable $\bE_i$, we obtain that
\begin{align*}
	\Pr\left[\bE_i \ge 50 \frac{\lone{x}}{|x_i|}\right] \le \frac{|x_i|}{50\lone{x}}.
\end{align*}
Thus, by a union bound, with probability $\ge 1 - 1/50$, for all $i$, $\bE_i \le 50\lone{x}/|x_i|$. Hence with probability $\ge 49/50$,
$
	\max_i |x_i|\bE_i \le 50 \lone{x}.
$
Similarly for all $i$, 
\begin{align*}
	\Pr\left[\bE_i \ge \frac{\lone{x}}{16|x_i|}\right] \ge \min(8|x_i|/\lone{x}, 1).
\end{align*}
If there exists an index $i$ such that $|x_i| \ge \lone{x}/8$, then we already have that with probability $1$, $\max_i |x_i|\bE_i \ge \lone{x}/8$. Now assume that for all $i$, $|x_i| \le \lone{x}/8$. Using the independence of $\bE_i$'s, we obtain that with probability $\ge 99/100$, there exists an index $i$ such that $\bE_i \ge \lone{x}/(16|x_i|)$ and therefore with probability $\ge 99/100$, $\max_i |x_i|\bE_i \ge \lone{x}/16$. Hence, overall with probability $\ge 95/100$, we have 
\begin{align*}
	\frac{\lone{x}}{16} \le \max_i |x_i|\bE_i \le 50\lone{x}.
\end{align*}
Let $T > 0$ be arbitrary. For a fixed $i$, we have $\Pr[\bE_i \ge \lone{x}/(|x_i|T)] \le |x_i|T/\lone{x}$. Thus,
\begin{align*}
	&\E[|\setbuilder{i}{|x_i|\bE_i \ge \lone{x}/T}|]\\
 &= \sum_i \Pr[\bE_i\\
 &\ge \lone{x}/(|x_i|T)] \le T\sum_i |x_i|/\lone{x} = T. \qedhere
\end{align*}
\end{proof}
The following lemma extends the above properties to all $p > 2$ along with an additional property that Andoni uses in his proof.
\begin{lemma}
Let $p > 2$ be arbitrary and let $x \in \R^d$ be an arbitrary vector of integer entries with absolute value at most $(\poly_1(d))^{1/p}$. Let $\bE_1,\ldots,\bE_d$ be independent discrete exponential random variables. Then,
\begin{enumerate}
    \item With probability $\ge 95/100$,
    \begin{align*}
        	\frac{\lp{x}}{16^{1/p}} \le \max_i |x_i|\bE_i^{1/p} \le 50^{1/p}\lp{x}.
    \end{align*}
    \item For any $T > 0$, with probability $\ge 95/100$,
	\begin{align*}
		|\setbuilder{i}{|x_i|\bE_i^{1/p} \ge \lp{x}/T^{1/p}}| \le 20T.
	\end{align*}
	\item 
	\begin{align*}
	    \E[\sum_{i=1}^d (\bE_i^{1/p} x_i)^2] = O_p(d^{1-2/p}\lp{x}^2).
	\end{align*}
\end{enumerate}
\label{lma:guarantees-on-z-with-independence}
\end{lemma}
\begin{proof}
The first two properties follow from applying the previous lemma to the vector $x^{(p)} \in \R^d$ defined as $x^{(p)}_i = |x_i|^p$. To prove the last property, we have
\begin{align*}
    \E[\sum_{i=1}^d (\bE_i^{1/p}x_i)^2] = \sum_{i=1}^d x_i^2\E[\bE_i^{2/p}].
\end{align*}
Now $\E[\bE_i^{2/p}] = \sum_{j=0}^{M-1} 2^{2j/p}/2^{j+1} + 2^{2M/p}/2^M \le (1/2)\sum_{j=0}^{\infty} 2^{j(2/p-1)} + 2^{M(2/p-1)} \le 1/(2 - 2^{2/p}) + 1$. Hence
\begin{align*}
    \E[\sum_{i=1}^d (\bE_i^{1/p}x_i)^2] &\le \opnorm{x}^2 \left(\frac{1}{2 - 2^{2/p}}+1\right)\\
    &= O_p(d^{1-2/p}\lp{x}^2)
\end{align*}
where the constant we are hiding blows up as $p \rightarrow 2$.
\end{proof}
\subsection{\texorpdfstring{$F_p$}{Fp} Estimation With \FastPRG}
In the previous section, we discussed some of the properties satisfied by the random vector $\bz$ given by $\bz_i := \bE_i^{1/p}x_i$ for any arbitrary vector $x$ with integer entries of absolute value at most $(\poly_1(d))^{1/p}$. Now we show that the properties are still satisfied even when the random variables $\bE_i$ are generated using \FastPRG thereby showing that Algorithm~\ref{alg:fp-with-fast-prg} outputs a constant factor approximation to $\lp{x}$ with probability $\ge 7/10$. The success probability can be increased to $1-\delta$ by running $O(\log 1/\delta)$ independent copies of the algorithm and reporting the median.
\begin{theorem}
Let $p > 2$ be a parameter and for $m = O(\poly(d))$ let the vector $x := 0\in \R^d$ receive a stream of $m$ updates $(i_1, v_1), (i_2, v_2), \ldots, (i_m, v_m)$ with $|v_j| \le \poly(d)$ for all $j$. On receiving an update $(i_j, v_j)$, the vector $x$ is modified as $x_{i_j} \gets x_{i_j} + v_{j}$. The Algorithm~\ref{alg:fp-with-fast-prg} uses $O_p(d^{1-2/p}\log(d))$ words of space and at the end of the stream outputs a constant factor approximation to $\lp{x}$ with probability $\ge 7/10$. Further each update to $x$ is processed by the algorithm in $O(1)$ time in the Word RAM model on a machine with a word size $\Omega(\log d)$. 
\end{theorem}
\begin{proof}
First we condition on the event that the hash families from which $\bh, \bsigma$ are drawn are $O(\log d)$-wise independent. From Theorem~\ref{thm:fast-hash-functions}, the event holds with probability $\ge 99/100$. From Theorem~\ref{thm:nisan-main-theorem}, \FastPRG with parameters $n = O(\log d)$, $b = d^{\epsilon}$ and $k = O(1/\epsilon)$ fools a Finite State Machine with $\poly(d)$ states. Further the seed for \FastPRG can be stored using $O((1/\epsilon)d^{\epsilon}\log d) = o(d^{1-2/p})$ bits if $\epsilon < 1 - 2/p$. Theorem~\ref{thm:fast-hash-functions} shows that $\bh$ and $\bsigma$ can be stored using $O(d^{\epsilon}) = o(d^{1-2/p})$ bits of space and for any $i \in [d]$, $\bh(i)$ and $\bsigma(i)$ can be evaluated in $O_{1/\epsilon}(1)$ time. Therefore, each update in the stream is processed in $O_{1/\epsilon}(1)$ time.

Let $x$ be the vector at the end of the stream and for $i \in [d]$, $\bE_i$ be the discrete exponential random variable computed by the algorithm for coordinate $i$. Let $T := (O(\log d))^p$. Define a FSM $Q_x$ with a start state and other states being defined by the tuple $(i,j,t,v)$ with $1 \le i \le d$, $0 \le j \le d, 1 \le v \le \poly(d)$ and $t \in \set{\less, \correct, \more}$. The machine $Q_x$ clearly has $\poly(d)$ states. The FSM being in a state $(i,j,\less, v)$ denotes that it has processed the coordinates $x_1,\ldots,x_i$ until now and found that
\begin{align*}
|\setbuilder{i' \le i}{|x_{i'}|\bE_{i'}^{1/p} \ge \lp{a}f^{1/p}}| 	&= j, \\
	\max_{i' \le i}|x_{i'}|\bE_{i'}^{1/p} &< \frac{\lp{x}}{16^{1/p}}, 
\text{ and }\\
    \sum_{i' \le i} \text{round}(\bE_{i'}^{1/p})^2 \cdot (x_{i
    '})^2 &= v.
\end{align*}
Here $\text{round}(x)$ denotes $x$ rounded to the nearest integer. As $\bE_{i'}^{1/p} \ge 1$ if $x_{i'} \ne 0$, we have 
\begin{align*}
(\bE_{i'}^{1/p}x_{i
    '})^2 \le \text{round}(\bE_{i'}^{1/p})^2 \cdot (x_{i
    '})^2 \le 4(\bE_{i'}^{1/p}x_{i
    '})^2.
\end{align*}
 Note that using the bound on the absolute values of entries in $x$, the value $v$ can be at most $\text{poly}(d)$. Similarly, the state is in $(i,j,\correct, v)$ if a condition similar to above holds but instead we have
\begin{align*}
		\frac{\lp{x}}{16^{1/p}} \le \max_{i' \le i}|x_{i'}|\bE_{i'}^{1/p} \le 50^{1/p}{\lp{x}}.
\end{align*}
and in $(i,j,\more, v)$ if $\max_{i' \le i}|x_{i'}|\bE_{i'}^{1/p} > 50^{1/p}{\lp{x}}$. Here the value of $\bE_i$ is assigned based on the bitstring corresponding to an edge in the FSM. It is clear that we can construct such a Finite State Machine. Let $\calE$ denote the event that $\max_{i} |x_i|\bE_i^{1/p} \in [\lp{x}/16^{1/p}, 50^{1/p}\lp{x}]$, $|\setbuilder{i}{|x_i|\bE_i^{1/p} \ge \lp{x}/T^{1/p}}| \le 20T$ and $\sum_{i \le i}\text{round}((\bE_{i}^{1/p})x_{i})^2 = O_p(d^{1-2/p}\|x\|_p^2)$. By Lemma~\ref{lma:guarantees-on-z-with-independence}, the final state distribution of the FSM using uniform random edge at every state satisfies the event $\calE$ with probability $\ge 85/100$. By Theorem~\ref{thm:nisan-main-theorem}, if the random variables $\bE_1,\ldots,\bE_d$ are generated using the random string sampled from \FastPRG, then with probability $\ge 1 - 1/\poly(d)$, the final state of FSM $Q_x$ satisfies the event $\calE$ with probability $\ge 8/10$.

Thus, with probability $\ge 8/10$, the implicit vector $\bz \in \R^d$ in the algorithm defined as $\bz_i := \bE_i^{1/p}x_i$ satisfies all the properties Andoni requires of the vector obtained by multiplying coordinates of $x$ with scaled exponential random variables. Hence, with probability $\ge 75/100$, the maximum absolute value of the coordinate in $\boldf$ obtained by sketching $\bz$ with a CountSketch matrix is a constant factor approximation to $\lp{x}$.

Thus overall the algorithm outputs a constant factor approximation to $\lp{x}$ with probability $\ge 7/10$. 
\end{proof}

\begin{algorithm}
\KwIn{$p > 2$, $d \in \N$, a stream of updates $(i_1,v_1), \ldots, (i_m, v_m) \in [d] \times \set{-M, \ldots, M}$ for $m, M = \poly(d)$}
\KwOut{An approximation to $\lp{x}$ where $x \in \R^d$ is defined by the stream of updates}
$\epsilon \gets$ A constant smaller than $1 - 2/p$\;
$\bs \gets $ Pseudorandom string constructed using \FastPRG with parameters $n = O(\log d), b = d^{\epsilon}, k = O(1/\epsilon)$\;
\tcp{The string $\bs$ is only implicitly stored using the corresponding hash functions and the random seed to the generator}
$d' \gets O(d^{1-2/p}\log d)$\;
$\bh \gets O(\log d)$-wise independent hash function from $[d] \rightarrow [d']$\;
$\bsigma \gets O(\log d)$-wise independent hash
function from $[d] \rightarrow \set{+1, -1}$\; \tcp{Both $\bh$ and $\bsigma$ are drawn from hash family in Theorem~\ref{thm:fast-hash-functions} so that they can be stored using $O_{1/\epsilon}(d^{\epsilon})$ bits and evaluated in $O_{1/\epsilon}(1)$ time on any input}

$\boldf \gets 0_m$\;
\tcp{Stream Processing Begins}
\For{$j=1,\ldots,m$}{
    $\bE_{i_j} \gets \text{DiscreteExponential}(i_j\text{-th chunk of $\bs$})$\;
    $\boldf_{\bh(i_j)} \gets \boldf_{\bh(i_j)} + \bsigma(i_j) \cdot \text{round}(\bE_{i_j}^{1/p}) \cdot v_{j}$\;
}
\Return{$\linf{\boldf}$\;}
\caption{$F_p$ moment estimation using \FastPRG}
\label{alg:fp-with-fast-prg}
\end{algorithm}

\subsection{Comparison with Andoni's use of Nisan-Zuckerman PRG}
Andoni argues that his algorithm can be run in $O(d^{1-2/p}\log d)$ words of space using the Nisan-Zuckerman pseudorandom generator, which shows that an $S$ space algorithm using $\poly(S)$ random bits can be run with just $O(S)$ random bits. Nisan-Zuckerman's algorithm takes an $O(S)$ length uniformly random string and stretches it by a factor of $O(S^{\gamma})$ ($0 < \gamma < 1$) by computing by computing $O(S^{\gamma})$ blocks of $O(S)$ bits each. Each of the $O(S^{\gamma})$ blocks of pseudorandom bits takes time $\poly(S)$ time to compute which in our case is $\poly(d^{1-2/p}\log d)$ and hence prohibitive.

\subsection{\texorpdfstring{$\ell_p$}{lp} sampling}\label{subsec:lp-sampling}
As another application of $\FastPRG$, we give a simple $\ell_p$ sampling algorithm for $p>2$. Assume the same turnstile stream setting. At the end of the stream, $\ell_p$ sampling asks to output a coordinate $i$ of the underlying vector $x$ such that probability of sampling $i$ is proportional to $|x_i|^p/\lp{x}^p$. The problem has been widely studied (see \cite{cormode2019p, Jayaram2018perfect} and references therein). The perfect $\ell_p$ sampling algorithm of \cite{Jayaram2018perfect} for $p \in (0,2)$ uses the following property of exponential random variables: if $\bE_1,\ldots,\bE_d$ are independent standard exponential random variables and $i^* = \argmax_i x_i/\bE_i^{1/p}$, then $\Pr[i^* = i] = |x_i|^p/\|x\|_p^p$. This distribution exactly corresponds to what $\ell_p$ sampling asks. To implement the algorithm in a turnstile stream using a small amount of space, they first scale the coordinates with exponentials and then sketch the scaled vector using a data structure called \emph{count-max} and show that the count-max data structure allows to recover the max coordinate in the vector obtained by scaling $x$ with exponential random variables. Finally, they derandomize their construction using a half-space fooling pseudorandom generator.

We show that using $\FastPRG$, for $p > 2$, we obtain $\ell_p$ samplers that have a very fast update time. For simplicity, we discuss an algorithm that samples from the following distribution:
\begin{align*}
    \Pr[\text{$i$ is sampled}] \ge \frac{1}{1+\varepsilon}\frac{|x_i|^p}{\lp{x}^p} \pm 1/\poly(d).
\end{align*}
In the definition approximate perfect samplers, it is required that the $\Pr[\text{$i$ is sampled}]$ is $(1 \pm \varepsilon)|x_i|^p/\lp{x}^p$ up to an additive $1/\poly(d)$ error. We discuss the simpler version as $\ell_p$ samplers are not our main focus.

For this, we work with a finer approximation of exponential random variables than what we used for approximating $F_p$ moments. Assume that given a block of $O(\log d)$ uniform random bits, there is a fast way to convert the random bits into a fine enough discretization of the exponential random variables such that all the following property of the exponential random variables hold for this discretization:
\begin{enumerate}
    \item The probability that $i^*  = i$\ and $\bE_{i^*}^{-1/p}|x_{i^*}| \ge (1+\varepsilon)^{-1/p} \max_{i' \ne i} \bE_{i}^{-1/p}|x_i|]$ is  at least $ |x_i|^p/(1+\varepsilon)\lp{x}^p$, and
    \item with probability at least $1 - 1/\poly(d)$, \[\sum_{i=1}^d (\bE_{i}^{-1/p}|x_i|)^2 \le d^{1-2/p}(\bE_{i^*}^{-1/p}|x_{i^*}|)^2\polylog(d).\]
\end{enumerate}
It can be easily shown that both of these properties hold for \emph{continuous} exponential random variables and hence they hold for a suitable discretization of the continuous exponential random variables. The above properties crucially depend on $\bE_1,\ldots,\bE_d$ being independent but we can derandomize them by using $\FastPRG$. Fix a vector $x$ and design the following FSM for $x$: the FSM goes through each coordinate of $x$ sequentially. The FSM tracks $\max_i \bE_i^{-1/p}|x_i|$, the coordinate attaining the max, $\sum_i \bE_i^{-2/p}x_i^2$. Clearly, all these statistics can be tracked using an FSM with $\poly(d)$ states similar to how we derandomized the properties of exponential random variables for approximating $F_p$ moments. Using $n = O(\log d)$ and $b = d^{c}$ ($c < 1-2/p$) for $\FastPRG$, thus we obtain that if the random variables $\bE_i$ are generated using the pseudorandom string sampled from $\FastPRG$, then for all $i$,
\begin{small}
    \begin{align*}
    &\Pr_{(\bE_1,\ldots,\bE_d) \sim \FastPRG}[i^*  = i\ \text{and}\ \bE_{i^*}^{-1/p}|x_{i^*}| \ge (1+\varepsilon)^{-1/p} \max_{i' \ne i} \bE_{i}^{-1/p}|x_i|]\\
    &\ge \frac{|x_i|^p}{(1+\varepsilon)\lp{x}^p} - 1/\poly(d).
\end{align*}
\end{small}
and with probability $\ge 1 - 1/\poly(d)$, $\sum_{i=1}^d (\bE_{i}^{-1/p}|x_i|)^2 \le d^{1-2/p}(\bE_{i^*}^{-1/p}|x_{i^*}|)^2\polylog(d)$. Let $\bf$ be a vector such that $\boldf_i = \bE_i^{-1/p} x_i$. Condition on both the events. Then we have that the largest coordinate in $\boldf$ is at least a $1/(1+\varepsilon)$ factor larger than the second largest coordinate and that $\opnorm{\boldf}^2 = O(d^{1-2/p}\linf{\boldf}^2)$. Now hashing the coordinates of $\boldf$ into a CountSketch data structure with $O(d^{1-2/p}\polylog(d)/\varepsilon^2)$ rows using $O(\log d)$ wise preserves the large coordinate of $\boldf$ using the analysis in \cite{andoni2017high}. Using $O(1)$ independent CountSketch data structure and finding the coordinate $i \in [d]$ which hashes to the max bucket of the CountSketch data structure in each of the $O(1)$ repetitions, we can extract the coordinate $i$ and output it as the $\ell_p$ sample.

Note that the update time is $O(1)$ as a block of random bits from $\FastPRG$ with parameter $k = d^c$ can be obtained in $O(1)$ time and then the time to evaluate the hash functions of the CountSketch data structure is $O(1)$ when using the constructions from Theorem~\ref{thm:fast-hash-functions}. The overall space complexity of the data structure is $O(d^{1-2/p}\polylog(d)/\varepsilon^2)$ bits. However, we note that the final step of computing which $i \in [d]$ hashes to the max coordinate in all $O(1)$ copies of the CountSketch data structure takes $O(d)$ time.

\section{Moment Estimation for \texorpdfstring{$0 < p < 2$}{0 < p < 2}}
We assume as usual that a vector $x \in \R^d$ is being maintained in the stream. The vector $x$ is initialized to $0$ and then receives $m$ updates of the form $(i,v) \in [d] \times \set{-M, \ldots, M}$ where upon receiving $(i,v)$, we update $x_i \gets x_i + v$. We assume that both $m,M \le \poly(d)$. At the end of the stream, we want to approximate $\lp{x}^p$ up to a $1 \pm \varepsilon$ multiplicative factor. For $\varepsilon$ such that $1/\sqrt{d} \le \varepsilon \le 1/d^{c}$ for a small enough constant, we show that the algorithm of \cite{fast-moment-estimation-optimal-space} can be implemented in space of $O(\varepsilon^{-2}\log(d))$ \emph{bits} of space and $O(\log d)$ update time per stream element. We measure the time complexity of the update algorithm in the Word RAM model with a word size of at least $\Omega(\log d)$. We first give a high level overview of the moment estimation algorithm of \cite{fast-moment-estimation-optimal-space}. Throughout the section, we assume that $1/\sqrt{d} \le \varepsilon \le 1/d^c$.
\subsection{Overview of Moment Estimation Algorithm of \texorpdfstring{\cite{fast-moment-estimation-optimal-space}}{kane et al}}
Their 1-pass algorithm is based on the Geometric Mean estimator of Li \cite{ping-li}. Li gives an estimator to compute $F_p$ moment of a vector using $p$-stable random variables. Let $x$ be a fixed $d$ dimensional vector and $\bv_1,\bv_2,\bv_3 \in \R^d$ be random vectors with independent $p$-stable random variables. Then, Li showed that the estimator given by
\begin{align}
    \Est = \frac{(|\la \bv_1, x\ra||\la \bv_2, x\ra||\la \bv_3, x\ra|)^{p/3}}{\left(\frac{2}{\pi} \Gamma\left({p}/{3}\right) \Gamma\left(2/3\right) \sin \left({\pi p}/{6}\right)\right)^3}
    \label{eqn:lis-estimator}
\end{align}
satisfies $\E_{\bv_1,\bv_2,\bv_3}[\Est] = \lp{x}^p$ and $\Var_{\bv_1,\bv_2,\bv_3}[\Est] = O(\lp{x}^{2p})$. A $1 \pm \varepsilon$ approximate estimator can be obtained by averaging $O(1/\varepsilon^2)$ independent copies of the estimator but it makes the update time $\Omega(\varepsilon^{-2})$ which is prohibitive. To decrease the variance of the estimator, \cite{fast-moment-estimation-optimal-space} hash the coordinates of $x$ into buckets and estimate the contribution of each bucket to the $F_p$ moment of $x$ separately. When there are heavy coordinates in the vector $x$, variance of the estimator may still be too large. Therefore, they estimate the contribution of the \emph{heavy} elements using a different data structure they call \HighEnd{} and use the Li's estimator to only estimate the contribution of the light elements using a data structure they call \LightEstimator. We show that when $(1/\sqrt{d}) \le \varepsilon \le 1/d^c$ for a constant $c < 1/2$, we can implement their algorithm using $O(\varepsilon^{-2}\log d)$ bits of space and an update time of $O(\log d)$ per each element in the stream in the Word RAM model.

\subsection{The \HighEnd{} Data Structure}
As described above, the algorithm of \cite{fast-moment-estimation-optimal-space} estimates the $F_p$ moment of heavy entries and light entries separately. Their heavy entry moment estimation method, at the end of the stream, takes in $L \subseteq [d]$ satisfying the following conditions:
\begin{enumerate}
    \item $L \supseteq \setbuilder{i \in [d]}{|x_i|^p \ge \alpha\lp{x}^p}$,
    \item if $i \in L$, then $|x_i|^p \ge (\alpha/C)\lp{x}^p$ for a constant $C \ge 1$, and
    \item we know the sign of $x_i$ for all $i \in L$.
\end{enumerate}
We show in Appendix~\ref{sec:heavy-entries} how a set $L$ satisfying the above properties can be computed in a turnstile stream using $O(\alpha^{-1}\log^2d)$ bits of space and $O(\log d)$ update time per stream element. We use the CountSketch data structure \cite{Charikar2004finding} along with the ExpanderSketch data structure \cite{jowhari2011tight} to obtain $L$. Note that the set $L$ has size at most $O(1/\alpha)$.

Now we state the guarantees of the \HighEnd{} data structure. We first define \BasicHighEnd{} data structure and then define \HighEnd{} by taking independent copies of the \BasicHighEnd{} data structure. Let $\alpha$ be such that $1/\alpha = O(1/\varepsilon^2)$ and let $s = \Theta(1/\alpha)$ be a large enough power of $2$. Let $\bh : [d] \rightarrow [s]$ is picked at random from an $r_h$-wise independent hash family for $r_h = \Theta(\log 1/\alpha)$. Let $r = \Theta(\log 1/\varepsilon)$ be a sufficiently large power of $2$. Let $\bg: [d] \rightarrow [r]$ be drawn at random from an $r_g$-wise independent hash family for $r_g = r$. For each $i \in [d]$, we associate a random complex root of unity given by $\exp(\mathrm{i} 2\pi \bg(i)/r)$ where $\mathrm{i}$ denotes  $\sqrt{-1}$. We initialize $s$ counters $\bb_1,\ldots,\bb_s$ to $0$. Given an update of the form $(i,v)$, we set $\bb_{\bh(i)} \gets \bb_{\bh(i)} + \exp(\mathrm{i} 2\pi \bg(i)/r)v$.

The \HighEnd{} data structure is defined by taking $T$ independent copies of \BasicHighEnd{} data structure with $T = O(\max(\log(1/\varepsilon), \log(1/\alpha))) = O(\log d)$. Each of the copies of the \BasicHighEnd{} data structure is updated upon receiving an update $(i,v)$ in the stream. Let $(\bh^1,\bg^1), \ldots,(\bh^T, \bg^T)$ be the hash functions corresponding to each of the \BasicHighEnd{} data structures.

It is argued in \cite{fast-moment-estimation-optimal-space} that storing the coefficients of complex numbers up to a precision of $\Theta(\log(d))$ bits suffices if the number of updates is $\poly(d)$, the magnitude of each update is bounded by $\poly(d)$ and $1/\sqrt{d} \le \varepsilon \le 1/d^c$. Thus the space complexity of \HighEnd{} data structure (excluding the space required for storing the hash functions) is 
\begin{align*}
    O\left(\alpha^{-1}(\log d)^2\right)~\text{bits}.
\end{align*}

By Theorem~\ref{thm:fast-hash-functions}, for $t$ a large enough constant we can construct random hash families $\calH = \set{h : [d] \rightarrow [s]}$ and $\calG = \set{g : [d] \rightarrow [r]}$ such that with probability $\ge 1 - 1/\poly(d)$, the hash families $\calH$ and $\calG$ are $r_h$ and $r_g$ wise independent respectively. Now, if $\bh^{1}, \ldots, \bh^T$ are sampled independently from $\calH$ and $\bg^1,\ldots,\bg^T$ are sampled independently from $\calG$, they can be evaluated on any input in $O(1)$ time and therefore the update time per each stream element is $O(T) = O(\log d)$ in the Word RAM model. Each hash function $\bh^i$ and $\bg^i$ can be stored using $O(d^{c/2})$ bits and therefore the space required to store the hash functions necessary for the \HighEnd{} data structure is $o(d^c) = o(\varepsilon^{-2})$ bits.

At the end of stream, by Theorem~11 of \cite{fast-moment-estimation-optimal-space} we can use the \HighEnd{} data structure to compute a value $\Psi$ such that with probability $\ge 7/8$,
\begin{align*}
    |\Psi - \lp{x_L}^p| \le O(\varepsilon)\lp{x}^p.
\end{align*}
We then have the following lemma: 
\begin{lemma}
Given $1/\sqrt{d} \le \varepsilon \le 1/d^{c}$ for a small enough constant, and $\alpha$ such that $1/\alpha = O(1/\varepsilon^2)$, there is a streaming algorithm that takes $O(\alpha^{-1}\log^2(d))$ bits of space and has an update time of $O(\log d)$ per stream element in the Word RAM model satisfying:
\begin{enumerate}
    \item The algorithm outputs a set $L \subseteq [d]$ satisfying all the three properties stated above with probability $\ge 9/10$.
    \item Conditioned on the list $L$ satisfying those properties, the algorithm outputs a value $\Psi$ such that with probability $\ge 65/100$,
\begin{align*}
    |\Psi - \lp{x_L}^p| \le O(\varepsilon)\lp{x}^p.
\end{align*}
\end{enumerate}
By taking median of $\Psi$ output by $O(1)$ independent instances of the \HighEnd{} data structure, we have that conditioned on $L$ satisfying all the properties, we have an estimate of $\lp{x_L}^p$ with an additive error of $O(\varepsilon)\lp{x}^p$ with probability $\ge 99/100$. Hence, by a union bound with probability $\ge 8/10$, the algorithm outputs both a good list $L$ and a value $\Psi$ satisfying $|\Psi - \lp{x_L}^p| \le O(\varepsilon)\lp{x}^p$.
\label{lma:high-end}
\end{lemma}
Note that for $\alpha = \varepsilon^2\log(d)$, the algorithm uses $O(\varepsilon^{-2}\log d)$ \emph{bits} of space. For this setting of $\alpha$, we will now use the \LightEstimator{} data structure of \cite{fast-moment-estimation-optimal-space} to estimate $\lp{x_{[d] \setminus L}}^p$. We will now describe the \LightEstimator{} data structure and how it can be derandomized using \FastPRG.

\subsection{The \LightEstimator{} Data Structure}
\begin{algorithm}
\KwIn{A parameter $p \in (0,2)$, accuracy parameter $\varepsilon$ such that $1/\sqrt{d} \le \varepsilon \le 1/d^c$ for a small enough constant $c$, a parameter $\alpha$ such that $1/\alpha = O(1/\varepsilon^2)$, a stream of updates $(i_1, v_1), \ldots, (i_m, v_m) \in [d] \times \set{-M, \ldots, M}$, a list $L \subseteq [d]$ of heavy coordinates revealed at the end of the stream}
\KwOut{An estimate of $\lp{x_{[d] \setminus L}}^p$}
$s \gets O(1/\alpha)$\;
$\bh \gets $ A hash function sampled from the construction in Theorem~\ref{thm:pagh-and-pagh}\;
For $b \in [s]$ and $[j] \in [3]$, initialize $\bB_{b,j} \gets 0$\;
For $i \in [d]$ and $[j] \in [3]$, let $\bA_{i,j}$ be an independent $p$-stable random variable\;
\For{$j=1,\ldots,m$}{
    $\bB_{\bh(i_j), 1} \gets \bB_{\bh(i_j), 1} + \bA_{i_j, 1} v_j$\;
    $\bB_{\bh(i_j), 2} \gets \bB_{\bh(i_j), 2} + \bA_{i_j, 2} v_j$\;
    $\bB_{\bh(i_j), 3} \gets \bB_{\bh(i_j), 3} + \bA_{i_j, 3} v_j$\;
}
\tcp{The set $L$ is revealed to the algorithm}
\For{$b = 1,\ldots,s$}{
    \eIf{$b \notin \bh(L)$}{
    $\Est(b) \gets \frac{(|\bB_{b,1}| \cdot |\bB_{b,2}| \cdot |\bB_{b,3}|)^{p/3}}{((2/\pi) \Gamma(p/3) \Gamma(2/3) \sin(\pi p/6))^3}$\;
    }{
        $\Est(j) \gets 0$\;
    }
}
$\Phi \gets \frac{s}{s-|\bh(L)|}\sum_{b=1}^s \Est(b)$\;
\Return{$\Phi$\;}
\caption{\LightEstimator{} using Independent $p$-stable random variables}
\label{alg:light-estimator-with-independent-lp}
\end{algorithm}
As seen previously, the \HighEnd{} data structure lets us compute the $F_p$ moment of all the elements from $L$. We use the \LightEstimator{} data structure to approximate the $F_p$ moment of all the \emph{light} elements i.e., the coordinates of $x$ not in $L$. 

Assume that at the end of processing the stream we are given a set $L \subseteq [d]$, $|L| \le 2/\alpha$ and for all $i \notin L$, $|x_i|^p < \alpha\lp{x}^p$. 

Let $s = \Theta(1/\alpha) \ge 10|L|$ be a large enough power of $2$. For $i \in [d]$ and $j \in [3]$, let $\bA_{i,j}$ denote an independent $p$-stable random variable. Let $\bh : [d] \rightarrow [s]$ be a hash function drawn using the construction of Theorem~\ref{thm:pagh-and-pagh} with parameters $z = \ceil{2/\alpha+2}$ and a large enough constant $C$. For $i \in [s]$ and $j \in [3]$, initialize the counters $\bB_{i,j} = 0$. On receiving an update $(i,v)$ in the stream, for each $j \in [3]$, update $\bB_{\bh(i), j} \gets \bB_{\bh(i), j} + \bA_{i,j} \cdot v$. As argued in \cite{fast-moment-estimation-optimal-space}, we need to store the values $\bB_{i,j}$ only up to a precision of $\Theta(\log d)$ bits. Hence the space complexity of the \LightEstimator{} data structure (excluding the space required to store/generate $\bA_{i,j}$) is $O(\alpha^{-1}\log d)$ bits.

At the end of the stream, we receive the set $L \subseteq [d]$ of heavy hitters as specified in the previous section and we will then compute an estimator for $\lp{x_{[d] \setminus L}}^p$ as follows: for each $b \in [s] \setminus \bh(L)$, define 
\begin{align*}
 \Est(b) := \frac{(|\bB_{b,1}| \cdot |\bB_{b,2}| \cdot |\bB_{b,3}|)^{p/3}}{\left(\frac{2}{\pi} \Gamma\left({p}/{3}\right) \Gamma\left(2/3\right) \sin \left({\pi p}/{6}\right)\right)^3}.   
\end{align*}

Let $\bh(L) = \setbuilder{\bh(i)}{i \in L}$ denote the buckets to which the elements of $L$ are hashed into. Define the following estimator
\begin{align*}
	\Phi = \frac{s}{s - |\bh(L)|} \sum_{b \in [s] \setminus \bh(L)}\Est(b).
\end{align*}
It is shown in \cite{fast-manhattan} for $p=1$ and extended to all $0 < p < 2$ in \cite{fast-moment-estimation-optimal-space} that
\begin{align}
	\E_{\bh, \bA}[\Phi] = (1 \pm \alpha^{10})\lp{x_{[d] \setminus L}}^p.
	\label{eqn:mean-of-the-light-estimator}
\end{align}
We extend their analysis and show an upper bound on $\Var_{\bh, \bA}[\Phi]$.
\begin{remark}
Lemma~7 in \cite{fast-manhattan} and Theorem~15 in \cite{fast-moment-estimation-optimal-space} state that $\E_{\bh,\bA}[\Phi] = (1 \pm O(\varepsilon))\lp{x_{[d] \setminus L}}^p$ for $p=1$ and $0 < p < 2$ respectively. It can be seen from the proof of Lemma~7 in \cite{fast-manhattan}, we can obtain the above stronger result by just picking $C$ large enough when constructing the hash family $\calH$ using the construction in Theorem~\ref{thm:pagh-and-pagh}. A similar argument works to extend for $0 < p < 2$.
\end{remark}
\subsubsection{Analysis of the Estimator}
\begin{claim}
\begin{align*}
    \Var_{\bh, \bA}[\Phi]  \le O(\alpha)\lp{x}^{2p}.
\end{align*}
\label{claim:variance-of-light-estimator-full-independence}
\end{claim}
\begin{proof}
For simplicity, we define ${\Est}(b) = 0$ for all $b \in \bh(L)$. We have
\begin{align*}
	\E_{\bh, \bA}[\Phi^2]
 &= \E_{\bh, \bA}\left[\left(\frac{s}{s-|\bh(L)|}\sum_{b=1}^s {\Est}(b)\right)^2\right]\\
	&= \E_{\bh}\left[\left( \frac{s}{s-|\bh(L)|}\right)^2\left(\sum_{b=1}^s \E_{\bA\mid \bh} ({\Est}(b))^2 + \qquad \sum_{b \ne b'}\E_{\bA|\bh}({\Est}(b))({\Est}(b'))\right)\right]
\end{align*}
Now, for $b \notin \bh(L)$, using the variance and mean of Li's estimator \eqref{eqn:lis-estimator}, 
\begin{align}
	\E_{\bA\mid \bh}({\Est}(b))^2 &= \Var_{\bA\mid \bh}({\Est}(b)) + (\E_{\bA \mid \bh} {\Est}(b))^2 = O(1)\left(\sum_{j:\bh(j) = b}|x_j|^p\right)^2
\end{align}
and as $p$-stable random variables denoted by $\bA$ are independent, for $b \ne b'$ with $b, b' \notin \bh(L)$,
\begin{align}
    \E_{\bA \mid \bh} ({\Est}(b))({\Est}(b')) &= \E_{\bA \mid \bh}({\Est}(b)) \E_{\bA \mid \bh}({\Est}(b'))\nonumber\\
    &=\sum_{j: \bh(j) = b} |x_j|^p\sum_{j: \bh(j) = b'}|x_j|^p.
\end{align}
Hence,
\begin{align*}
    \E_{\bh, \bA}[\Phi^2]
    &= \E_{\bh} \left[\left(\frac{s}{s- |\bh(L)|}\right)^2\left( O(1)\sum_{b \notin \bh(L)}(\sum_{i : \bh(i) = b}|x_i|^p)^2  + \sum_{b \ne b' : b,b' \notin \bh(L)}(\sum_{i:\bh(i) = b}|x_i|^p)(\sum_{i:\bh(i) = b'}|x_i|^p)\right)\right]\\
    &\le \E_{\bh} \left[\left(\frac{s}{s-|\bh(L)|}\right)^2 \left(O(1)\sum_{b \notin \bh(L)}(\sum_{i : \bh(i) = b}|x_i|^p)^2  + \left(\sum_{i: \bh(i) \notin \bh(L)}|x_i|^p\right)^2\right)\right]\\
    &\le O(1)\E_{\bh}\left[\sum_{b \notin \bh(L)}(\sum_{i:\bh(i) = b}|x_i|^p)^2\right] + \E_{\bh} \left[\sum_{i\ne i'}\left(\frac{s}{s-|\bh(L)|}\right)^2\bone[\bh(i), \bh(i') \notin \bh(L)]|x_i|^p|x_{i'}|^p\right]
\end{align*}
where the last inequality follows from the fact that $|\bh(L)| \le |L| \le s/10$. We first bound the second term. For any $i, i' \notin L$, with probability $1 - \rho$, the hash function $\bh$ is drawn from a hash family that is $|L|+2$-wise independent when restricted to the set $L \cup \set{i,i'}$ when restricted to $L \cup \set{i,i'}$. We can make $\rho \le \varepsilon^C$ for any constant $C$ by setting $C$ large enough while sampling $\bh$ from the hash family in Theorem~\ref{thm:pagh-and-pagh}. Let the event that the hash family is $|L|+2$-wise independent with respect to $L \cup \set{i,i'}$ be called $\textsf{Good}$. Conditioned on this event and the size $|\bh(L)|$,
\begin{align}
    \E_{\bh \mid \textsf{Good}, |\bh(L)|}[\bone [\bh(i), \bh(i') \notin \bh(L)]] = \left(\frac{s - |\bh(L)|}{s}\right)^2
    \label{eqn:conditioned-on-goodness-and-size}
\end{align}
which gives
\begin{align}
  &\E_{\bh} \left[\sum_{i\ne i'}\left(\frac{s}{s-|\bh(L)|}\right)^2\bone[\bh(i), \bh(i') \notin \bh(L)]|x_i|^p|x_{i'}|^p\right]\nonumber\\ 
  &\le \E_{\bh \mid \textsf{Good}} \left[\sum_{i\ne i'}\left(\frac{s}{s-|\bh(L)|}\right)^2\bone[\bh(i), \bh(i') \notin \bh(L)]|x_i|^p|x_{i'}|^p\right]\nonumber\\
  &\quad + \E_{\bh \mid \lnot \textsf{Good}} \left[\sum_{i\ne i'}\left(\frac{s}{s-|\bh(L)|}\right)^2\bone[\bh(i), \bh(i') \notin \bh(L)]|x_i|^p|x_{i'}|^p\right] \times \Pr_{\bh}[\lnot \textsf{Good}]\nonumber\\
  &\le \lp{x_{[d] \setminus L}}^{2p} + 2\lp{x_{[d] \setminus L}}^{2p}\rho = (1 + 2\rho)\lp{x_{[d] \setminus L}}^{2p}. \label{eqn:variance-bound-second-term}
\end{align}
Here we used \eqref{eqn:conditioned-on-goodness-and-size} to cancel out the $s^2/(s - |\bh(L)|)^2$ factor in the expectation and we used that $s/(s - |\bh(L)|) \le 10/9$ with probability $1$ as $|\bh(L)| \le |L| \le s/10$. Thus
\begin{align}
    \E_{\bh,\bA}[\Phi^2] &\le  O(1)\E_{\bh}\left[\sum_{b \notin h(L)}(\sum_{i:h(i) = b}|x_j|^p)^2\right]\nonumber\\
    &\quad + (1 + 2\rho)\|x_{[d] \setminus L}\|_p^{2p}.
\end{align}
Now, we bound the first term.
\begin{align*}
    \E_{\bh}\left[\sum_{b \notin \bh(L)}\left(\sum_{i:\bh(i)=b}|x_j|^p\right)^2\right]
    &= \E_{\bh}\left[\sum_{b \notin \bh(L)}\sum_{i:\bh(i) = b}|x_i|^{2p}+ \sum_{b \notin \bh(L)}\sum_{i \ne i': \bh(i)=\bh(i') = b}|x_i|^{p}|x_{i'}|^p\right]\\
    &\le \|x_{[d] \setminus L}\|_{2p}^{2p} + \E_{\bh} \left[\sum_{b} \sum_{i \ne i'}\bone[b \notin \bh(L), \bh(i) = \bh(i') = b]|x_i|^p|x_{i'}|^p\right].
\end{align*}
Again, for $i, i' \notin L$, with probability $1 - \rho$, the hash function $\bh$ is drawn from a hash family that is $|L|+2$-wise independent on the set $L \cup \set{i, i'}$. Conditioning on that event, for any $b \in [s]$,
\begin{align*}
    \Pr[\bone[b \notin \bh(L), \bh(j) = b, \bh(j') = b]] \le O(1/s^2).
\end{align*}
Hence using an argument similar to the one in proving \eqref{eqn:variance-bound-second-term}, we have
\begin{align*}
    &\E_{\bh}\left[\sum_{b \notin \bh(L)}\left(\sum_{j:\bh(j)=b}|x_j|^p\right)^2\right]\\
    &\quad = \|x_{[d] \setminus L}\|_{2p}^{2p} + O(1/s + \rho )\|x_{[d] \setminus L}\|_p^{2p}.
\end{align*}
Therefore
\begin{align*}
   \Var_{\bh,\bA}[\Phi] &\le O(1)\|x_{[d] \setminus L}\|_{2p}^{2p} + O(1/s + \rho)\|x_{[d] \setminus L}\|_p^{2p} + (1 + O(\rho))\lp{x_{[d] \setminus L}}^{2p} - (1 - \alpha^{10})^2\|x_{[d] \setminus L}\|_p^{2p}\\
    &\le O(1)\|x_{[d] \setminus L}\|_{2p}^{2p}
    + O(1/s + \rho)\|x_{[d] \setminus L}\|_p^{2p} + O(\rho + \alpha^{10})\|x_{[d] \setminus L}\|_p^{2p}.
\end{align*}
Now, $\|x_{[d] \setminus L}\|_{2p}^{2p} \le \alpha \|x\|_p^p\lp{x_{[d] \setminus L}}^p$ using the fact that all coordinates in  $x_{[d] \setminus L}$ have $p$th power at most $\alpha\lp{x}^p$. As $\rho \le \varepsilon^C$ and $s = \Theta(1/\alpha)$,
\begin{align*}
    \Var_{\bh,\bA}[\Phi] \le O(\alpha)\lp{x}^{2p}. &\qedhere
\end{align*}
\end{proof}
Let $\Phi_1,\ldots,\Phi_T$ be $T$ independent copies of the estimator $\Phi$ for $T = \Theta(\log(1/\varepsilon)) = \Theta(\log d)$. As $\alpha = \varepsilon^2 \log (d)$, $\Var[(\Phi_1 + \cdots + \Phi_T)/T] \le O(\varepsilon^2 \log(d)/T)\lp{x}^{2p} \le (\varepsilon^2/100)\lp{x}^{2p}$. By Chebyshev inequality, with probability $\ge 99/100$,    $\bar{\Phi} := \frac{\Phi_1 + \cdots + \Phi_T}{T}$ lies in the interval
\begin{align*}
   [(1-\alpha^{10})\lp{x_{[d] \setminus L}}^p - (\varepsilon/10)\lp{x}^p, (1+\alpha^{10})\lp{x_{[d] \setminus L}}^p + (\varepsilon/10)\lp{x}^p].
\end{align*}
Hence, using $O(\log d)$ independent instantiations of the \LightEstimator{} data structure, we can obtain an estimate $\bar{\Phi}$ for $\lp{x_{[d] \setminus L}}$ so that with probability $\ge 7/10$, $\Psi + \bar{\Phi} \in [(1 - O(\varepsilon))\lp{x}^p, (1 + O(\varepsilon))\lp{x}^p]$. Throughout the analysis we assumed that the $p$-stable random variables $\bA_{i,j}$ are independent. To implement the algorithm in sub-linear space, we derandomize $p$-stable random variables using \FastPRG.
\subsubsection{Derandomizing \texorpdfstring{$p$}{p}-stable random variables using \FastPRG}
\begin{lemma}
Given $p \in (0,2)$, an accuracy parameter $\varepsilon$ such that $1/\sqrt{d} \le \varepsilon \le 1/d^{c}$ for a constant $c$, a parameter $\alpha$ such that $1/\alpha \le O(1/\varepsilon^2)$ and a stream of updates $(i_1,v_1),\ldots(i_m, v_m) \in [d] \times \set{-M, \ldots, M}$ for $m, M \le \poly(d)$ to the vector $x$, there is a streaming algorithm that uses $O(\alpha^{-1}\log d)$ bits of space and an update time of $O(\log d)$ per stream element in the Word RAM model. At the end of processing the stream, the algorithm takes in a set $L \subseteq [d]$ of heavy coordinates satisfying for all $i \notin L$, $|x_i|^p \le \alpha \lp{x}^p$ and outputs a value $\Phi$ satisfying
\begin{align*}
    \E[\Phi] = (1 \pm \varepsilon^{C})\lp{x_{[d] \setminus L}}^p \pm \frac{1}{\poly(d)}.
\end{align*}
and
\begin{align*}
    \Var[\Phi] \le O(\alpha)\lp{x}^p.
\end{align*}
\label{lma:light-estimator-derandomized}
\end{lemma}
The algorithm in above lemma is given by a modified version of Algorithm~\ref{alg:light-estimator-with-independent-lp}. Instead of using independent $p$-stable random variables $\bA_{i,j}$, the algorithm uses \FastPRG to obtain a pseudorandom string and uses the pseudorandom bits to compute the $p$-stable random variables.
\begin{proof}
Consider one instance of \LightEstimator{} data structure as in Algorithm~\ref{alg:light-estimator-with-independent-lp}. It consists of the following objects:
\begin{enumerate}
    \item A hash function $\bh : [d] \rightarrow [s]$ for $s = O(1/\alpha)$ drawn from a hash family as stated in Theorem~\ref{thm:pagh-and-pagh} with parameter $z = O(1/\alpha)$ and $C$ being a large enough constant.
    \item There is a table of $s = O(1/\alpha)$ counters each maintained with a precision of $O(\log d)$ bits.
\item The $p$-stable random variables $\bA_{i,j}$ for $i \in [d]$ and $j = 1,2,3$.
\end{enumerate}
The hash function $\bh$ can be stored using $O(\alpha^{-1}\log d)$ bits by Theorem~\ref{thm:pagh-and-pagh}. The counters can be maintained using $O(\alpha^{-1} \log d)$ bits as well. So, we are left with derandomizing the $p$-stable random variables.

The algorithm overall needs $O(d\log d)$ uniform random bits to generate three $p$-stable random variables for each of the coordinates. We use \FastPRG to obtain the pseudorandom bits and use them to generate $p$-stable random variables. We critically use the fact that our analysis of the estimator $\Phi$ as described in the previous section needs to use only the mean and variance of $\Phi$ to show that we only have to ``fool'' multiple $O(\log d)$ space algorithms and hence using \FastPRG as described in Theorem~\ref{thm:nisan-main-theorem} is enough to generate the pseudorandom bits to compute the $p$-stable random variables. Further, for each update in the stream, the necessary block of pseudorandom bits can be generated in $O(1)$ time for each update. The space required to store the randomness necessary for the pseudo random generator is $O(d^{\epsilon}) = O(1/\varepsilon^2)$ bits when $\epsilon \le 2c$ where $\varepsilon \le 1/d^c$.

So, the overall algorithm on each update $(i, v)$ is as follows: We use $\bh$ to hash $i$ into one of the $O(1/\alpha)$ buckets. Note that $\bh(i)$ can be computed in $O(1)$ time. Using the value of $i$, we generate a block of $O(\log d)$ pseudorandom bits from the pseudorandom  generator and then use the bits to compute $3$ samples from a $p$-stable distribution. Let the samples be $\bA_{i,1}, \bA_{i,2}, \bA_{i,3}$. We update the counters $\bB_{h(i), j} \gets \bB_{h(i), j} + \bA_{i,j}v$ for $j \in [3]$. As discussed, generating the pseudorandom bits and updating the counters can be performed in $O(1)$ time (assuming the pseudorandom bits can be converted to samples from a $p$-stable distribution in $O(1)$ time).

Let $\bgamma \sim \FastPRG$ be a string sampled from the pseudorandom generator. Let $\bA_{\bgamma}$ denote the $p$-stable random variables generated using $\bgamma$. Let $\bA$ denote $p$-stable random variables generated using a uniform random string of bits. Hence the random variables $\bA_{i,j}$ are independent. Fix a hash function $\bh$, a vector $x$ and a set of heavy elements $L$. Now consider the estimator we use to estimate the $F_p$ moments of the light elements $[d] \setminus L$:
\begin{small}
\begin{align*}
    \Phi_{\bh, \bA}=  \frac{s}{(s-|\bh(L)|)\theta}\, \times\sum_{b \in [s] \setminus \bh(L)} (|\sum_{i: \bh(i) =b} \bA_{i,1}x_i|\cdot|\sum_{i: \bh(i) =b} \bA_{i,2}x_i|\cdot|\sum_{i: \bh(i) =b} \bA_{i,3}x_i|)^{p/3}
\end{align*}    
\end{small}
where we use $\theta$ to denote the denominator in \eqref{eqn:lis-estimator}. Fix some bucket $b \in [s] \setminus \bh(L)$. Define
\begin{small}
 \begin{align*}
\Phi^{(b)}_{\bh, \bA} \coloneqq \frac{s}{(s-|\bh(L)|)\theta}  (|\sum_{i: \bh(i) =b} \bA_{i,1}x_i|\cdot|\sum_{i: \bh(i) =b} \bA_{i,2}x_i|\cdot|\sum_{i: \bh(i) =b} \bA_{i,3}x_i|)^{p/3}.    
\end{align*}   
\end{small}

For $b \in \bh(L)$, we define $\Phi_{\bh, \bA}^{(b)} = 0$.
The quantity $\Phi^{(b)}_{\bh, \bA}$ can be computed by an $O(\log d)$ space algorithm in a single pass over the uniform random string of bits used to generate the $p$-stable random variables $\bA$ by going over the values $i=1,\ldots,d$ and ignoring the random bits that correspond to all $i$ such that $\bh(i) \ne b$. We formalize the algorithm by constructing an FSM $Q_{x, L, \bh, b}$ with $\poly(d)$ states. Note that we fixed $x, L, \bh$ and $b$. Let the state of automaton be of the form $(i, c_1, c_2, c_3)$ where $i \in [d]$ and $c_1,c_2, c_3$ denote the counters. The FSM being in state $(i, c_1, c_2, c_3)$ denotes that it has processed $x_1,\ldots,x_i$ and found that for $j \in [3]$
\begin{align*}
	c_j = \sum_{i' \le i : h(i') = b} \bA_{i',j}x_{i'}.
\end{align*}
When in state $(i, c_1, c_2, c_3)$, if $\bh(i+1) \ne b$, the FSM directly transitions to the state $(i+1, c_1, c_2, c_3)$ ignoring the alphabet in the input to the FSM. If $\bh(i) = b$, then the FSM reads the alphabet in the input string. Uses the $\set{0,1}^{O(\log d)}$ size bit string that it reads to construct three $p$-stable random variables $\bA_{i+1, 1}, \bA_{i+1, 2}, \bA_{i+1, 3}$ and then transitions to the state $(i+1, c_1', c_2', c_3')$ where for $j \in [3]$
\begin{align*}
	c_j' = c_j + \bA_{i+1, j} x_{i+1}.
\end{align*}
Note that all the above operations are performed only with a precision of $O(\log d)$ bits. Hence the Finite State Machine has only $\poly(d)$ states. From the state $(d, c_1, c_2, c_3)$, the algorithm transitions to $(\text{final}, (s/(s-|\bh(L)|)\theta)(|c_1||c_2||c_3|)^{p/3})$.

Given a uniform random string as input, the final state of FSM $Q_{c,L,\bh, b}$ encodes the value $\Phi^{(b)}_{\bh, \bA}$ and given $\bgamma \sim \FastPRG$ as input, the final state of FSM $Q_{c, L, \bh, b}$ encodes the value of $\Phi^{(b)}_{\bh, \bA_{\bgamma}}$.

Let $d_{\bh}^b$ be the distribution of the value of $\Phi^{(b)}_{\bh, \bA}$ conditioned on $\bh$. Now define $(d_{\bh}^{b})'$ to be the distribution of $\Phi_{ \bh, \bA_{\bgamma}}^{(b)}$ i.e., the value of the estimator for $b$th bucket computed using $p$-stable random variables generated from a random $\bgamma ~ \FastPRG$. As FSM $Q_{x, L \bh, b}$ has only $\poly(d)$ states, we obtain using Theorem~\ref{thm:nisan-main-theorem} that
\begin{align*}
    d_{\TV}(d_{\bh}^b, ((d_{\bh})^{b})') \le 1/\poly(d).
\end{align*}
The above is true for any fixing of $\bh$, $x$, $L$ and $b$. As for any values of $\bh$ and $\bA$, we have $|\Phi_{\bh, \bA}^{(b)}| \le \poly(d)$, we obtain that for any fixed $\bh$, $x$, $L$ and $b$,
\begin{align*}
    |\E_{\bA}[\Phi_{\bh, \bA}^{(b)}] - \E_{\bA_{\bgamma}}[\Phi_{\bh, \bA_{\bN}}^{(b)}]| \le \frac{1}{\poly(d)}
\end{align*}
which implies that
\begin{align*}
    |\E_{\bA}[\Phi_{\bh, \bA}] - \E_{\bA_{\bgamma}}[\Phi_{\bh, \bA_{\bgamma}}]| \le s \cdot \frac{1}{\poly(d)} \le \frac{1}{\poly(d)}.
\end{align*}
Therefore, using \eqref{eqn:mean-of-the-light-estimator} we have
\begin{align}
	\E_{\bh, \bA_{\bgamma}}[\Phi_{\bh, \bA_{\bgamma}}] &= \E_{\bh, \bA}[\Phi_{\bh, \bA}] \pm 1/\poly(d)= (1 \pm \varepsilon^C)\lp{x_{[d] \setminus L}}^p \pm 1/\poly(d).
	\label{eqn:mean-with-nisan-p-stable}
\end{align}
Similarly, we have
\begin{align*}
    (\Phi_{\bh, \bA})^2 =  \sum_{b,b' \in [s]} \Phi_{\bh, \bA}^{(b)} \cdot \Phi_{\bh, \bA}^{(b')}.
\end{align*}
Again, for any fixed pair $b,b'$, we can compute the product $\Phi_{\bh, \bA}^{(b)} \cdot \Phi_{\bh, \bA}^{(b')}$ in $O(\log d)$ space using one pass over the uniform random string used to generate the $p$-stable random variables. We can construct a Finite State Machine very similar to the one above to show that the value $\Phi_{\bh, \bA}^{(b)} \cdot \Phi_{\bh, \bA}^{(b')}$ can be computed by a machine with $\poly(d)$ states. Now, we have that the total variation distance between the distributions of the product when using a uniform random string to generate $p$-stable random variables and \FastPRG to generate the $p$-stable random variables is at most $1/\poly(d)$ and hence we obtain that
\begin{align*}
    |\E_{\bA}(\Phi_{\bh, \bA}^{(b)} \cdot \Phi_{\bh, \bA}^{(b')}) - \E_{\bA_{\bgamma}}(\Phi_{\bh, \bA_{\bgamma}}^{(b)} \cdot \Phi_{\bh, \bA_{\bgamma}}^{(b')})| \le 1/\poly(d).
\end{align*}
Summing over all the pairs $(b,b')$, we obtain for any $x, L, \bh$ that,
\begin{align*}
    |\E_{\bA}[(\Phi_{\bh, \bA})^2] - \E_{\bA_{\bgamma}}[(\Phi_{\bh, \bA_{\bgamma}})^2]| \le \frac{1}{\poly(d)}
\end{align*}
which implies
\begin{align*}
    \E_{\bh,\bA_{\bgamma}}[(\Phi_{\bh, \bA_{\bgamma}})^2] \le \E_{\bh, \bA}[(\Phi_{\bh, \bA})^2] + \frac{1}{\poly(d)}.
\end{align*}
We then have
\begin{align*}
    \Var_{\bh, \bA_{\bgamma}}[\Phi_{\bh, \bA_{\bgamma}}] &\le \E_{\bh,\bA}[(\Phi_{\bh, \bA})^2] + \frac{1}{\poly(d)} - (\E_{\bh, \bA}[\Phi_{\bh, \bA}])^2 +  \frac{\E_{\bh,\bA}[\Phi_{\bh, \bA}]}{\poly(d)}.
\end{align*}
As the $\poly(d)$ term can be made $\ge d^C$ for a large enough constant $C$, we obtain that
\begin{align*}
    \Var_{\bh,\bA_{\bgamma}}[\Phi_{\bh, \bA_{\bN}}] \le \Var_{\bh, \bA}[\Phi_{\bh, \bA}] + \frac{1}{\poly(d)} \le O(\alpha)\lp{x}^{2p}.
\end{align*}
by Claim~\ref{claim:variance-of-light-estimator-full-independence} and using the fact that $x$ is a nonzero vector with integer coordinates and $\alpha \ge 1/\poly(d)$.
\end{proof}
\subsection{Wrap-up}
\begin{theorem}
  Given $p \in (0,2)$, an accuracy parameter $\varepsilon$ such that $1/\sqrt{d} \le \varepsilon \le 1/d^{c}$ for a constant $0 < c < 1/2$ and a stream of updates $(i_1,v_1),\ldots(i_m, v_m) \in [d] \times \set{-M, \ldots, M}$ for $m, M \le \poly(d)$ to the vector $x$, there is a streaming algorithm that uses $O(\varepsilon^{-2}\log d)$ bits of space and has an update time of $O(\log d)$ per stream element that outputs with probability $\ge 7/10$, a value $v$ such that
  \begin{align*}
      v = (1 \pm \varepsilon)\lp{x}^p.
  \end{align*}
\end{theorem}
\begin{proof}
    Setting $\alpha = \varepsilon^2 \log(d)$, we have that the set of $\alpha$ heavy hitters $L$ can be computed in $O(\alpha^{-1}\log^2(d)) = O(\varepsilon^{-2}\log d)$ bits of space using Lemma~\ref{lma:heavy-hitters} and has an update time of $O(\log d)$ per stream element. The set $L$ satisfies all the properties in Lemma~\ref{lma:heavy-hitters} with probability $\ge 9/10$. By Lemma~\ref{lma:high-end}, the \HighEnd{} data structure can be maintained in $O(\alpha^{-1}\log^2(d)) = O(\varepsilon^{-2}\log d)$ bits of space and has an update time of $O(\log d)$ per stream element. Conditioned on $L$ satisfying all the properties, we have that the value $\Psi$ output by \HighEnd{} data structure satisfies with probability $\ge 9/10$,
\begin{align*}
    \Psi = (1 \pm \varepsilon)\lp{x_{L}}^p.
\end{align*}

By Lemma~\ref{lma:light-estimator-derandomized}, the \LightEstimator{} data structure can be maintained in $O(\alpha^{-1}\log d) = O(\varepsilon^{-2})$ bits of space. We also have that the data structure can be update in $O(1)$ time per stream element and conditioned on the set $L$ having all the properties, the value $\Phi$ output by the algorithm satisfies
\begin{align*}
    \E[\Phi] = (1 \pm \varepsilon^C)\lp{x_{[d] \setminus L}}^{p} + 1/\poly(d)
\end{align*}
and
\begin{align*}
    \Var[\Phi] = O(\varepsilon^2 \log d)\lp{x}^{2p}.
\end{align*}
Maintaining $r = O(\log d)$ independent copies of \LightEstimator{} in the stream and considering their outputs $\Phi_1,\ldots,\Phi_r$, conditioned on $L$ having all the properties, we obtain using Chebyshev's inequality that with probability $\ge 99/100$,
\begin{align*}
   \bar{\Phi} = \frac{\Phi_1 + \cdots + \Phi_r}{r} = (1 \pm \varepsilon^C)\lp{x_{[d] \setminus L}}^p + \frac{1}{\poly(d)} + \varepsilon\lp{x}^p.
\end{align*}
Thus, by a union bound, with probability $\ge 7/10$,
\begin{align*}
    \Psi + \bar{\Phi} = (1 \pm O(\varepsilon))\lp{x}^p + \frac{1}{\poly(d)} = (1 \pm O(\varepsilon))\lp{x}^p
\end{align*}
using the fact that $\varepsilon > (1/\sqrt{d})$ and there is at least one nonzero integer coordinate in $x$.
\end{proof}

\section{Derandomizing CountSketch with \FastPRG}\label{sec:derandomizing-countsketch}

CountSketch~\cite{Charikar2004finding} is a random linear map of a vector $x \in \mathbb{R}^d$ to $Ax \in \mathbb{R}^D$. For parameters $\countsketchrepetitions, \countsketchrange$ such that $D = \countsketchrepetitions\countsketchrange$, the CountSketch $\CS(x)$ is defined by two sequences of random independent hash functions: $\countsketchhash_1, \ldots, \countsketchhash_\countsketchrepetitions: \set{0, \ldots, d-1} \rightarrow [\countsketchrange]$ and $\bs_1, \ldots, \bs_\countsketchrepetitions: \set{0, \ldots, d-1} \rightarrow \{-1,+1\}$. 
To simplify our exposition we will assume that $\countsketchrange$ is a power of two and that $\countsketchrepetitions$ is odd. We additionally assume that the coordinates of $x$ are $0$-indexed so that $x = (x_0, \ldots, x_{d-1})$.
Indexing $\CS(x) \in \mathbb{R}^D$ by $(i,j) \in [\countsketchrepetitions] \times [\countsketchrange]$ the entries are defined as: 
$$\CS(x)_{i,j} \coloneqq \sum_{\ell = 0}^{d-1} \bs_i(\ell)\, x_{\ell}\, [\countsketchhash_i(\ell) = j].$$
For $x\in\mathbb{R}^d$ and $\ell \in \set{0, \ldots, d-1}$ we use $\CS(x)$ to approximate $x_\ell$ with the following estimator: 
\begin{equation}\label{eq:countsketch-estimator}
\hat{x}_{\ell} = \textrm{median}(\{\bs_i(\ell) \cdot \CS(x)_{i, \countsketchhash_i(\ell)} \mid i \in [\countsketchrepetitions]\})\enspace.
\end{equation}

Charikar, Chen, and Farach-Colton~\cite{Charikar2004finding} upper bounded the estimation error $|\hat{x}_{\ell} - x_{\ell}|$ in terms of the norm of the vector $x$ and the parameters $\countsketchrepetitions$, $\countsketchrange$.
Their analysis only relies on using \emph{pairwise independent} hash functions (independently for each repetition), which require $O(\countsketchrepetitions \log d)$ bits of storage and allow the estimator to be computed in $O(\countsketchrepetitions)$ time assuming constant time arithmetic operations.

Minton and Price~\cite{Minton2014improved} presented a tighter analysis of the distribution of the estimation error, focusing on $\countsketchrepetitions = \Theta(\log d)$ repetitions, under the assumption that the hash function values of $\countsketchhash_1, \ldots, \countsketchhash_\countsketchrepetitions$ and $\bs_1, \ldots, \bs_\countsketchrepetitions$ are \emph{fully independent}.
This assumption is used in order to argue about the Fourier transform of the error distribution.
In our notation they show the following lemma:
\begin{lemma}
    \label{lemma:count-sketch}
    For every $\alpha\in [0,1]$, $\ell\in [d]$,
    $\Pr\left[ |\hat{x}_\ell - x_\ell| > \alpha\,\Delta\right] < 2\exp\left(-\Omega\left(\alpha^2 \countsketchrepetitions \right)\right)$,
            where $\Delta=\|\text{tail}_{\countsketchrange}(x)\|_2 / \sqrt{\countsketchrange}$.
\end{lemma}

Literally storing fully random hash functions would require $O(\countsketchrepetitions d \log \countsketchrange)$ bits, so this is not attractive when $d\gg\countsketchrange$, which is the setting where using CountSketch is of interest.
Minton and Price note that for integer vectors $x\in \{-M,\dots,+M\}^d$ where $M$ is polynomial in $d$, it is possible to use the pseudorandom generator of Nisan~\cite{nisan} to replace the fully independent hash functions, keeping the tail bound of Lemma~\ref{lemma:count-sketch} up to a $1/\poly(d)$ additive term.
However, this comes with considerable overhead:
The space complexity increases by an $\Omega(\log (d\countsketchrange))$ factor, and the time per update/query increases by a factor $\Omega(\countsketchrepetitions \countsketchrange)$.
Jayaram and Woodruff~\cite{Jayaram2018perfect} later considered a modification of CountSketch (with the same space and error guarantees) and showed that the multiplicative space overhead can be reduced to $O((\log \log d)^2)$ using a pseudo-random generator for fooling halfspaces.
The time complexity increases by an unspecified polylogarithmic factor compared to the fully random setting.
Though it is technically not accurate we will still refer to their sketch as CountSketch.

\medskip

In this section we present an alternative derandomization of~\cite{Minton2014improved} using hash functions computed using \FastPRG.
Specifically, for $i\in [\countsketchrepetitions]$ and $\rho \in \set{0, \ldots, d-1}$ we use block number $(i-1) d + \rho$ from the output of \FastPRG to get the random bits for $\bs_i(\rho)$ and $\countsketchhash_i(\rho)$.
Since a given output block of \FastPRG can be computed efficiently, these hash functions can be efficiently evaluated.

\begin{theorem}\label{theorem:fastprg-count-sketch}
Let $d$ be the dimension of the vectors and $M$ be the maximum absolute value of coordinates in the vector. Let $\countsketchrange$ and $\countsketchrepetitions$ be the parameters of the CountSketch map as defined above. Let $b \ge 2$ be an integer denoting the branching factor of $\FastPRG$. Let $w = \Omega(\log d + \log M)$. There exists a randomized linear sketch $\CS_{\FastPRG}: \set{-M, \ldots, M}^d \rightarrow \set{-2^{w}, \ldots, 2^w}^{tr}$ that can be implemented on a word RAM with word size $w$ with the following properties:
\begin{itemize}
    \item The parameters required to defined the map $\CS_{\FastPRG}$ can be stored in $O(b \log_b d)$ words of space and given a vector $x \in \set{-M, \ldots, M}^d$, the resulting vector $\CS_{\FastPRG}(x)$ can be stored in $O(\countsketchrepetitions\countsketchrange)$ words of space.
    \item Given $\CS_{\FastPRG}(x)$ and an update $(\ell,u_\ell)$ corresponding to a vector $u$ with a single nonzero entry $u_\ell$, we can compute $\CS_{\FastPRG}(x + u)$ in time $O(\countsketchrepetitions \log_b d)$.
    \item For every $x\in \{-M,\dots,M\}^d$, $\alpha\in [0,1]$, and $\ell\in [d]$, we can compute an estimator $\hat{x}_\ell$ from $\CS_{\FastPRG}$ in time $O(\countsketchrepetitions \log_b d)$ such that
        $\Pr\left[ |\hat{x}_\ell - x_\ell| > \alpha\,\Delta\right] < 2\exp\left(-\Omega\left(\alpha^2 \countsketchrepetitions \right)\right) + 2^{-Cw}$, 
         where $\Delta=\|\text{tail}_{\countsketchrange}(x)\|_2 / \sqrt{\countsketchrange}$.
\end{itemize}
    
\end{theorem}

Figure~\ref{fig:countsketch-overview} compares Theorem~\ref{theorem:fastprg-count-sketch} to previously known ways of choosing the hash functions for CountSketch.
Our construction is the first one that is able to match the space usage of CountSketch with pairwise independent hash functions (for $\countsketchrepetitions = O(\log d)$ repetitions) while showing the strong concentration known for fully random hash functions.
With CountSketch table size $\countsketchrange = d^{\Omega(1)}$ and word length $w=O(\log d)$ we also match the update time of pairwise independence on the Word RAM.

\begin{figure*}
    \begin{center}
        \resizebox{\textwidth}{!}{\begin{tabular}{lccc} 
        \toprule
        \textbf{Hash function} & \textbf{Space in words} & \textbf{Bounds small error} & \textbf{Update time}\\ [0.5ex]
         \midrule
         Pairwise independent~\cite{Charikar2004finding}  & $D$ & No & $\log d$ \\ 
         Fully random~\cite{Minton2014improved} &  $\textcolor{red}{d} \log d$ & Yes & $\log d$\\
         Nisan's generator~\cite{Minton2014improved, nisan} & $D \log(d)$ & Yes & $\textcolor{red}{\countsketchrange} \log^3(d)$ \\
         Halfspace Fooling PRGs~\cite{Jayaram2018perfect} & $D (\log\log d)^2$  & Yes & $(\log d)^{O(1)}$\\
          \textbf{\FastPRG} ($b=t$) & $D$ & Yes & $\log^2(d) / \log t$ \\
          \textbf{\FastPRG} ($b=d^{\Omega(1)}$) & $D + b$ & Yes & $\log d$ \\
          \bottomrule
        \end{tabular}}
        \caption{Overview of CountSketch guarantees with different kinds of random hash functions.
        For simplicity we focus on the case of $\countsketchrepetitions = O(\log d)$ repetitions and $d$-dimensional input vectors that contain $O(\log d)$-bit integers such that the CountSketch itself (without hash functions) uses space $D= O(t\log d)$ words.
        With pairwise independence we can only tightly bound the probability of exceeding error $\Delta = \|\text{tail}_{\countsketchrange}(x)\|_2 / \sqrt{\countsketchrange}$, while the other hash functions allow us to bound the probability of smaller errors.
        Time bounds are for implementation on a Word RAM with word size $w = O(\log d)$.
        Parameters with a particularly bad impact on space or time are highlighted in red color.}\label{fig:countsketch-overview}
    \end{center}
    \end{figure*}

\subsection{PRGs for Space-bounded Computation and CountSketch}

Like Minton and Price~\cite{Minton2014improved} we will consider vectors $x\in \{-M,\dots,M\}^d$ for a positive integer $M$.
For concreteness we consider CountSketch with entries that are $w$-bit machine words.
We can relax the requirement from~\cite{Minton2014improved} that $M$ is polynomial in $d$, and instead assume that $M < 2^{w-1}/d$, which is also necessary to ensure that there are no overflows when computing $\CS(x)$.



To derandomize CountSketch, we describe a small-space algorithm for any fixed input vector $x$, query $\ell$ and threshold $\alpha\Delta$.
The algorithm makes a single pass over the output from \FastPRG and determines whether the estimator $\hat{x}_\ell$ computed using \FastPRG has error exceeding $\alpha\Delta$.
This is done \emph{without} computing $\CS(x)$, and in fact even without computing $\hat{x}_\ell$.
The algorithm makes critical use of the \emph{symmetry} of \FastPRG, namely, that the distribution of hash values is unchanged by permuting the inputs using a mapping of the form $\rho \mapsto \rho\oplus\ell$.
We stress that Nisan's generator does not have this symmetry property, and that we are not aware of an equally space-efficient finite state machine for evaluating the error of CountSketch using Nisan's generator.

\paragraph{The finite state machine.}
Consider $x \in \{-M,\dots,M\}^d$, $\ell\in \set{0, \ldots, d-1}$, and a given error threshold $\alpha\Delta \in \mathbb{R}$.
A choice of hash functions $\countsketchhash_1, \ldots, \countsketchhash_\countsketchrepetitions: \set{0, \ldots, d-1} \rightarrow [\countsketchrange]$ and $\bs_1, \ldots, \bs_\countsketchrepetitions: \set{0, \ldots, d-1} \rightarrow \{-1,+1\}$ can be represented as a binary string $\bgamma\in \{0,1\}^{\countsketchrepetitions d w}$ where a block of $w$ consecutive bits encodes hash values $\bs_i(\rho)$ and $\countsketchhash_i(\rho)$ for $i\in [\countsketchrepetitions]$ and $\rho\in \set{0, \ldots, d-1}$.
We order the blocks such that hash values with $i=1$ come first, then $i=2$ and so on.
Concretely we may take $\bs_i(\rho) = 2\bgamma_{(i-1)dw + \rho w} - 1$ and $\countsketchhash_i(j) = 1 + \sum_{k=1}^{\log(\countsketchrange)} \bgamma_{(i-1)dw + (\rho+1)w - k} 2^{k-1}$ such that the hash values can be extracted from a block in constant time.
We consider two ways of sampling the string $\bgamma$:
\begin{itemize}
\item First, we may choose $\bgamma \sim (U_2)^{\countsketchrepetitions d w}$ with independent, random bits.
It is easy to see that this is the same as choosing the hash functions with full independence, so Lemma~\ref{lemma:count-sketch} holds for this choice of $\bgamma$.
\item Second, for every $b$, $k$ such that $d \countsketchrepetitions \leq b^k < 2^{cw}$ we can use \FastPRG with block size $n = w$ to generate $\bgamma\sim G_k(*,\bh_1,\dots,\bh_k)$.
This corresponds to the hash functions we use to derandomize CountSketch.
\end{itemize}
As we noted in the introduction, the distribution of $\bgamma^{\oplus \ell}$ is the same as the distribution of $\bgamma$. Thus, we can assume that an algorithm that makes a single pass over the string $\bgamma$ reads the $rd$ blocks of bits in the order $0 \oplus \ell, 1 \oplus \ell, \ldots, (rd-1)\oplus \ell$ so that for each repetition $i \in [\countsketchrange]$, the algorithm gets to know the value $\bg_i(\ell)$ before reading the pseudorandom bits corresponding to other blocks. Thus, for each repetition $i \in [t]$, we can assume that an algorithm making a single pass over the string $\bgamma$ \emph{sees} the values $\bg_i(0 \oplus \ell), \bg_i (1 \oplus \ell), \ldots, \bg_i((d-1) \oplus \ell)$ in that order and similarly the values of $\bs_i$.
\medskip

To analyze the properties of the CountSketch data structure constructed using the string $\bgamma$, we consider an FSM $Q = Q_{x,\ell,\alpha\Delta}$ with states 
$$\{0,\dots,\countsketchrepetitions + 1\}^3 \times \{-1, +1\} \times [\countsketchrange] \times \set{0, \ldots, d} \times \{-2^w,\dots,2^w\},$$
plus a special start state $\bot$. 
We use $\countsketchhash_1, \ldots, \countsketchhash_\countsketchrepetitions$ and $\bs_1, \ldots, \bs_\countsketchrepetitions$ to refer to the hash functions encoded by $\bgamma$.
The idea is that the FSM $Q$ computes the $\countsketchrepetitions$ simple estimators in (\ref{eq:countsketch-estimator}) one at a time, and keeps track of the number of these estimators that deviate from $x_\ell$ by more than $\alpha\Delta$ (with separate accounting for overestimates and underestimates).
More precisely, when $Q$ is in state $(\beta_1,\beta_2,\beta_3,\beta_4,\beta_5,\beta_6,\beta_7)$ it signifies that:
\begin{itemize}
\item It has fully processed the hash functions $\bs_i$, $\countsketchhash_i$ for $i < \beta_1$, that $\beta_2$ of these hash function pairs produced an estimate $\bs_i(\ell) \cdot \CS(x)_{i, \countsketchhash_i(\ell)} < x_\ell - \alpha\Delta$, and $\beta_3$ pairs produced an estimate $\bs_i(\ell) \cdot \CS(x)_{i, \countsketchhash_i(\ell)} > x_\ell + \alpha\Delta$,
\item $\bs_{\beta_1}(\ell)=\beta_4$ and $\countsketchhash_{\beta_1}(\ell) = \beta_5$, and
\item $\beta_4 \sum_{0 \le z < \beta_6} x_{z\oplus\ell} \bs_{\beta_1}(z \oplus \ell)[\countsketchhash_{\beta_1}(z \oplus \ell) = \countsketchhash_{\beta_1}(\ell)] = \beta_7$.
\end{itemize}
%
From state $\bot$, the FSM $Q$ transitions to $(1,0,0,\bs_1(\ell),\countsketchhash_1(\ell),0,0)$, where the values $\bs_1(\ell),\countsketchhash_1(\ell)$ are determined by the zeroth block since we assume that the FSM sees the blocks as they are ordered in the string $\bgamma^{\oplus \ell}$.
From this point on, when in state $(\beta_1,\beta_2,\beta_3,\beta_4,\beta_5,\beta_6,\beta_7)$:
\begin{itemize}
    \item If $\beta_6 = d$ we have $\beta_7 = \bs_{\beta_1}(\ell) \CS(x)_{\beta_1, \countsketchhash_{\beta_1}(\ell)}$, so we can decide whether simple estimator number $\beta_1$ has an error above $\alpha\Delta$ and update $\beta_2$ or $\beta_3$ accordingly.
    Finally we can increment $\beta_1$, set $\beta_6 = 0$, and update $\beta_4$, $\beta_5$ to reflect the values of the new hash values, available in the next block.
    \item Otherwise, when $\beta_1 \leq r$, $Q$ has access to $\bs_{\beta_1}(\beta_6 \oplus \ell)$ and $\countsketchhash_{\beta_1}(\beta_6 \oplus \ell)$ from the next block of bits it reads.
    This allows us to increment $\beta_6$ when simultaneously increasing $\beta_7$ by $\beta_4 \bs_{\beta_1}(\beta_6 \oplus \ell) x_{\beta_6 \oplus \ell}$ if $\bg_{\beta_1}(\beta_6 \oplus \ell) = \beta_5$.
    \item Finally, when $\beta_1 = r+1$ we ignore the rest of the input, remaining in the same state.
\end{itemize}
After reaching the end, the values $\beta_2$, $\beta_3$ determine whether the estimator $\hat{x}_\ell$ in $(\ref{eq:countsketch-estimator})$ has an error of more than $\alpha\Delta$:
    If $\beta_2 \geq \lceil \countsketchrepetitions / 2\rceil$ then $\hat{x}_\ell < x_\ell - \alpha\Delta$,
    if $\beta_3 \geq \lceil \countsketchrepetitions / 2\rceil$ then $\hat{x}_\ell > x_\ell + \alpha\Delta$, and
    otherwise $|\hat{x}_\ell - x_\ell| \leq \alpha\Delta$.

The number of states in $Q$ is $O(\countsketchrepetitions^3\countsketchrange d 2^w)$ and the number of blocks of bits read is $\countsketchrepetitions d$, both of which are $2^{O(w)}$.
Theorem~\ref{lma:generalized-nisan} with $n = O(w)$ now implies that the total variation distance $\|Q(G_k(*, \bh_1,\ldots,\bh_k)) - Q((U_n)^{2^k})\|$ is at most $2^{-cw}$ for \emph{some} constant $c>0$.
The additive term $2^{-cw}$ can be made smaller than $2^{-C w}$ for \emph{any} constant $C>1$ by adjusting the parameter $n$ since a Word RAM with a constant factor larger word size can be simulated with a constant factor overhead in time.
In particular, error probabilities grow by at most $2^{-Cw}$ when switching from fully random hash functions to hash functions defined by to \FastPRG.
Finally, note that the space usage of \FastPRG is $O(w^2 b/\log b)$ bits, and that we can compute the hash function values $\countsketchhash_1(\ell), \ldots, \countsketchhash_\countsketchrepetitions(\ell)$ and $\bs_1(\ell), \ldots, \bs_\countsketchrepetitions(\ell)$ in time $O(\countsketchrepetitions w / \log b)$.

\subsection{Alternative Derandomizations of CountSketch with Nisan's PRG}\label{subsec:alternate-derandomizations-nisan}
We note that there are alternative derandomizations using Nisan's PRG that do not incur the na\"ive $O(\log d)$ factor overhead in terms of space in some regimes. We can use the fact that Nisan's PRG also ``fools'' small space algorithms that make multiple passes over the pseudorandom string. Specifically, Lemma~2.3 of \cite{david2011strong} shows that Nisan's PRG fools a small space algorithm that makes $2$ passes over the pseudorandom string. We derandomize each repetition $i \in [\countsketchrepetitions]$ separately. Fix an index $\ell \in [d]$ and repetition $i \in [\countsketchrepetitions]$ and consider the string used to compute hash functions $\bg_i$ and $\bs_i$ for repetition $i$. An algorithm in the first pass over the string computes the bucket into which the index $\ell$ gets hashed into and in the second pass over the string computes the value, denoted by $\hat{x}_{i, \bg_i(\ell)}$, of the bucket into which the $\ell$-th coordinate gets hashed into in the $i$-th repetition. Finally the algorithm terminates with the value $\bs_i(\ell)\hat{x}_{i, \bg_i(\ell)}$. As this algorithm overall takes $O(\log d)$ space, it is ``fooled'' by Nisan's PRG with a seed length of $O(\log^2 d)$. Thus, we obtain $\Pr[|\bs_i(\ell)\hat{x}_{i, \bg_i(\ell)} - x_{\ell}| \le \alpha\Delta] \gtrsim \alpha$ as in proof of Theorem~4.1 of \cite{Minton2014improved} even when each repetition of CountSketch is independently derandomized using a string sampled from Nisan's PRG. Overall, if each repetition is derandomized using an independent pseudorandom string, we obtain that $\Pr[|\hat{x}_{\ell} - x_{\ell}| > \alpha \Delta] \le 2\exp(-\Omega(\alpha^2 r))$. While this derandomization has a fast update time, a drawback is that the overall seed length is $O(r \log^2 d)$ (a string of $O(\log ^2 d)$ bits for each $i \in [r]$) which can be larger than the space complexity of CountSketch ($O(\countsketchrepetitions \countsketchrange \log d)$ bits) when $\countsketchrange = o(\log d)$.

There is also another derandomization using Nisan's PRG which avoids space blow-up in the case of $rt = \omega(\log d)$. Instead of using independent samples from Nisan's PRG for each repetition of CountSketch, we derandomize the construction all-at-once. Consider the following algorithm that uses a single sample from Nisan's PRG to construct all hash functions $\bg_i$ and $\bs_i$. Fix a vector $x$ and coordinate $\ell$. The algorithm makes a first pass over the string to determine into which bucket the index $\ell$ gets hashed into in each of the repetitions.
The information can be stored using $O(r\log t)$ bits. In the second pass over the random string, the algorithm can then determine if the estimate $\hat{x}_{\ell}$ satisfies $|\hat{x}_{\ell} - x_{\ell}| \le \alpha \Delta$. Overall the algorithm uses a space of $O(\log d + r\log t)$ bits. Using this fact, one can obtain a derandomization of CountSketch with the above estimation error guarantee and the CountSketch derandomized using Nisan's generator can be stored in $O(\countsketchrepetitions \cdot \countsketchrange + \log d)$ words of space, which is asymptotically the same as the space complexity of CountSketch derandomized using \FastPRG. However, Nisan's generator needs to fool an $O(\countsketchrepetitions \cdot \log\countsketchrange + \log d)$ space algorithm which makes the update time much slower, on a machine with $O(\log d)$ word size, compared to the derandomization using \FastPRG which only has to fool an $O(\log d)$ space algorithm.

\section{Private CountSketch}
Pagh and Thorup~\cite{Pagh2022improved} recently analyzed the estimation error of CountSketch data structure made private using the Gaussian Mechanism. Their analysis requires the hash functions $\bg_1,\ldots,\bg_{\countsketchrepetitions} : \set{0, \ldots, d-1} \rightarrow [\countsketchrange]$ and $\bs_1,\ldots,\bs_{\countsketchrepetitions} : \set{0, \ldots, d-1} \rightarrow \set{+1, -1}$ to be fully random. They define
\begin{align*}
    \PCS(x) = \CS(x) + \bnu
\end{align*}
where $\bnu$ is a $D = \countsketchrepetitions \cdot \countsketchrange$ dimensional vector with independent Gaussian random variables of mean $0$ and variance $\sigma^2$. By taking $\sigma$ to be appropriately large, we obtain that $\PCS(x)$ is $(\varepsilon, \delta)$-differentially private. For $\ell \in \set{0, 1,\ldots, d-1}$, we can define the estimator $\hat{x}_{\ell}$ as
\begin{align*}
    \hat{x}_{\ell} = \text{median}(\setbuilder{\bs_{i}(\ell) \cdot \PCS(x)_{i, \bg_i(\ell)}}{i \in [\countsketchrepetitions]}).
\end{align*}
As we saw in Section~\ref{sec:derandomizing-countsketch}, if $\sigma$ were $0$, then the hash functions $\bg_i$ and $\bs_i$ can be derandomized using \FastPRG while obtaining tail bounds on the estimation error $|x_{\ell} - \hat{x}_{\ell}|$. A similar argument which crucially uses the symmetry property of \FastPRG shows that PCS can also be derandomized using \FastPRG. For the case of Private CountSketch with fully random hash functions the following theorem is shown in~\cite{Pagh2022improved}:
\begin{theorem}[\cite{Pagh2022improved}]
    For every $\alpha \in [0,1]$ and every $\ell \in [d]$, the estimation error of private CountSketch with $\countsketchrepetitions$ repetitions, table size $\countsketchrange$ and $\bnu \sim N(0,\sigma^{2})^{D}$, then
    \begin{align*}
        \Pr[|\hat{x}_{\ell} - x_{\ell}| \ge \alpha \max(\Delta, \sigma)] \le 2\exp(-\Omega(\alpha^2 \countsketchrepetitions))
    \end{align*}
    where $\Delta = \opnorm{\text{tail}_{\countsketchrange}(x)}/\sqrt{\countsketchrange}$.
    \label{thm:pagh-thorup-pcs}
\end{theorem}
We now derandomize the requirement that the hash functions $\bg_i$ and $\bs_i$ be fully random.
The proof is an extension of the proof in Section~\ref{sec:derandomizing-countsketch} and proceeds very similarly. We detail it below for completeness. Instead of constructing an FSM, we describe a small space algorithm which makes a single pass over the randomness while updating its state after reading a block of random bits. The algorithm can be easily converted to an FSM.

Fix a vector $x \in \R^d$ (corresponds to the final value of the vector in the stream), $\nu \in \R^{D}$ (corresponds to the Gaussian vector we add to CountSketch to make it private), a coordinate $\ell \in [d]$ and a parameter $\alpha$.
We initialize three variables $\textsf{acc} = \textsf{deficit} = \textsf{excess} = 0$.
The algorithm reads a block of $w$ bits from the input string. Using the symmetry of $\FastPRG$ we again assume that the FSM sees the blocks of $\bgamma$ as they are in $\bgamma^{\oplus \ell}$. Thus, for each repetition, the block that the FSM sees first can be used to determine $\bg_i(\ell)$ and $\bs_i(\ell)$. For the first repetition, using the zeroth block of bits, the 
FSM computes and stores the values $\bg_1(\ell)$ and $\bs_1(\ell)$ locally and update the accumulator \textsf{acc} to $\bs_1(\ell) x_{0 \oplus \ell}$ (note that $x_{0 \oplus \ell} = x_\ell$) and move to the next block of bits in the input string.

Suppose we are reading the $j$-th block of random bits.
Again, we can determine $\bg_i(j \oplus \ell)$ and $\bs_i(j \oplus \ell)$.
If $\bh_1(j \oplus \ell) \ne \bg_1(\ell)$, we move to reading the $(j+1)$-th block.
If $\bg_1(j \oplus \ell) = \bg_1(\ell)$, we add $\bs_1(j \oplus \ell)x_{j \oplus \ell}$ to the accumulator and move to the $(j+1)$-th block. 
After reading $d$ blocks, if $\bs_1(\ell) (\textsf{acc} + \nu_{1, \bg_1(\ell)}) \ge x_{\ell} + \alpha \max(\Delta, \sigma)$ we increase the variable $\textsf{excess}$  by $1$. 
If $\bs_1(\ell) (\textsf{acc} + \nu_{1, \bg_1(\ell)}) \le x_{\ell} - \alpha \max(\Delta, \sigma)$, we increase $\textsf{deficit}$ by $1$. 
Then we zero out the variable $\textsf{acc}$, remove the stored values $\bg_1(\ell)$ and $\bs_1(\ell)$ and repeat the process by reading the next block of bits.
We again determine $\bg_2(\ell)$ and $\bs_2(\ell)$ and set $\textsf{acc}$ to $\bs_2(\ell) x_{\ell}$.
We then move to the next block and so on. 
We do the process in total for $\countsketchrepetitions$ of times.
Finally, if $\textsf{excess} < \countsketchrepetitions/2$ and $\textsf{deficit} < \countsketchrepetitions/2$, we set the variable \textsf{Status} to \textsf{Success} and otherwise set \textsf{Status} to \textsf{Failure}. 

The algorithm, at any point of time needs to store only $O(\log d)$ bits. As the algorithm needs only ``read'' access to $x$, the entire algorithm can be converted to an FSM, which we call $Q_{x,\nu, \ell,\alpha}$, that has $\poly(d)$ states over the alphabet $\set{0,1}^{w}$. The variable $\textsf{Status}$ determines if the FSM ends in a state $\textsf{Success}$ or \textsf{Failure}.

If the input string to the FSM is uniformly random, then the hash functions $\bg_i$ and $\bs_i$ are fully random. Therefore by Theorem~\ref{thm:pagh-thorup-pcs}, we obtain that
\begin{align*}
    \E_{\bnu}[\Pr_{\bgamma \sim U}[\text{Final State} = \textsf{Success}]] \ge 1 - 2\exp(-\Omega(\alpha^2 \countsketchrepetitions)).
\end{align*}
Now consider \FastPRG with a block size $w = \Omega(\log d)$. By Theorem~\ref{thm:nisan-main-theorem}, for any fixed $x, \nu, \ell$ and $\alpha$ that
\begin{align*}
    &\Pr_{\bgamma \sim \FastPRG}[\text{Final State of FSM $Q_{x, \nu, \ell, \alpha}$ on input $\bgamma$}]\\
    &\quad \ge \Pr_{\bgamma \sim U}[\text{Final State of FSM $Q_{x, \nu, \ell, \alpha}$ on input $\bgamma$}] - O(2^{-cw}).
\end{align*}
By taking an expectation over $\bnu \sim N(0, \sigma^2)^{D}$, we have
\begin{align*}
    \E_{\bnu}\Pr_{\bgamma \sim \FastPRG}[\text{Final State of FSM $Q_{x, \bnu, \ell, \alpha}$ on input $\bgamma$}]\ge  1 - 2\exp(-\Omega(\alpha^2 r)) -  O(2^{-cw}).
\end{align*}
We therefore obtain that with probability $\ge 1 - 2\exp(-\Omega(\alpha^2r)) - 2 \cdot 2^{-cw}$ over $\bnu \sim N(0,\sigma^2)^D$ and $\bgamma \sim \FastPRG$, if $\bg_i$ and $\bs_i$ are hash functions constructed as a function of $\bgamma$ as described in the above algorithm, then with probability $\ge 1 - \exp(-\Omega(\alpha^2 r)) - O(2^{-cw})$ over $\bnu$ and $\bgamma$,
\begin{align*}
     |\text{median}_{i \in [r]}(\bs_i(\ell)(\sum_{j = 0}^{d-1} [\bg_i(j) = \bg_i(\ell)]\bs_j(i)x_{j} + \nu_{i, \bg_i(\ell)})) - x_{\ell}| \le \alpha \max(\Delta, \sigma).
\end{align*}
Thus we have the following theorem.
\begin{theorem}
    For every $\alpha \in [0,1]$ and every $\ell \in \set{0, \ldots, d-1}$, the estimation error of private CountSketch with $\countsketchrepetitions$ repetitions, table size $\countsketchrange$ derandomized using \FastPRG with a block size of $w$ and $\bnu \sim N(0,\sigma^2)^{D}$, then
    \begin{align*}
        \Pr[|\hat{x}_{\ell} - x_{\ell}| \ge \alpha \max(\Delta, \sigma)] \le 2 \exp(-\Omega(\alpha^2 r)) + O(2^{-cw})
    \end{align*}
    where $\Delta = \opnorm{\text{tail}_t(x)}/\sqrt{t}$.
\end{theorem}
\section{Estimating \texorpdfstring{$\linf{x}$}{l-inf-x}}
Let $x \in \R^d$ be the underlying vector we are maintaining in a turnstile stream. Assume that the coordinates of $x$ are integers bounded in absolute value by $\poly(d)$. We give a simple algorithm that uses only $O(\varepsilon^{-2}\log d \log 1/\varepsilon)$ bits of space and approximates $\linf{x}$ up to an additive error of $\varepsilon\opnorm{x}$. We give a matching lower bound and show that our algorithm is tight up to constant factors.

Let $t \le d$ be a parameter that we set later. Let $\bL : \R^d \rightarrow \R^t$ be a randomized linear map defined as
\begin{align*}
    (\bL x)_{i} = \sum_{j: \bh(j) = i} \bs(i) x_i.
\end{align*}
Here we assume that the hash function $\bh$ is drawn from a 2-wise independent hash family $\calH = \set{h : [d] \rightarrow [t]}$ and the sign function $\bs$ is drawn from a 4-wise independent hash family $\calS = \set{s : [d] \rightarrow \set{+1, -1}}$. We note that there exist families $\calH$ and $\calS$ such that $\bh$ and $\bs$ can be sampled from their respective families and stored using $O(\log d)$ bits. We prove the following lemma:
\begin{lemma}
    Given a parameter $\alpha$, if $t \ge 1/2\alpha^4\delta$, then with probability $\ge 1 - 3\delta$, the following simultaneously hold:
    \begin{enumerate}
        \item $\linf{\bL} = \linf{x} \pm (2\sqrt{\alpha}/\delta^{1/4})\opnorm{x}$ and
        \item $\opnorm{\bL x}^2 \le (1+2\alpha^2)\opnorm{x}^2$.
    \end{enumerate}
    \label{lemma:first-level-hashing}
\end{lemma}
\begin{proof}
Define
$
    \textLarge \coloneqq \setbuilder{j \in [d]}{|x_j| \ge \alpha\opnorm{x}}
$
and $\textsc{Small} \coloneqq [d] \setminus \textLarge$. Note that $|\textLarge| \le 1/\alpha^2$. Let $x_{\textLarge} \in \R^d$ be the $d$-dimensional vector with only the \textLarge{} coordinates of vector $x$ and define $x_{\textSmall} = x - x_{\textLarge}$. The following result is now a simple consequence of the 2-wise independence of $\bh$.
\begin{lemma}
If $t \ge |\textLarge|^2/(2\delta)$, then with probability $\ge 1 - \delta$, for all $j,j' \in \textLarge$ with $j \ne j'$, we have $\bh(j) \ne \bh(j')$.
\end{lemma}
\begin{proof}
For $j, j' \in \textLarge$ with $j \ne j'$, let $\bX_{j,j'} = 1$ if $\bh(j) = \bh(j')$ and $0$ otherwise. Using the 4-wise independence of $\calH$, we have $\E[\bX_{j,j'}] = 1/t$. Hence, $\E[\sum_{j < j'}\bX_{j,j'}] \le |\textLarge|^2/(2t)$. By Markov's inequality, with probability $\ge 1 - \delta$, $\sum_{j < j'}\bX_{j,j
'} \le |\textLarge|^2/(2t\delta) \le 1/2$ if $t \ge |\textLarge|^2/\delta$. By definition, the random variable $\sum_{j < j'}\bX_{j,j'}$ takes only non-negative integer values. Hence, we obtain that with probability $\ge 1 - \delta$, $\sum_{j < j'}\bX_{j,j'} = 0$ which implies that the hash function $\bh$ hashes each of the coordinates in the set $\textLarge$ to distinct locations.  
\end{proof}
As $t \ge 1/2\alpha^4\delta \ge |\textLarge|^2/2\delta$, we get that all the coordinates in the set $\textLarge$ are hashed to distinct buckets by the hash function $\bh$ with probability $1 - \delta$. 
We now bound $\linf{\bL x_{\textSmall}}$. By definition of the set \textSmall, we have $\linf{x_{\textSmall}} \le \alpha\opnorm{x}$ and $\opnorm{x_{\textSmall}} \le \opnorm{x}$. Let $i \in [t]$ be arbitrary. We have
$
    (\bL x_{\textSmall})_i = \sum_{j \in \textSmall} [\bh(j) = i] \bs(j)x_j
$
and
\begin{align*}
    \E[(\bL x_{\textSmall})_i^2] = \sum_{j,j' \in \textSmall}\Pr[\bh(j) = i, \bh(j') = i]\E[\bs(j)\bs(j')]x_j x_{j'}
\end{align*}
using the independence of $\bh$ and $\bs$. For $j \ne j' \in \textSmall$, using the 4-wise independence of $\bs$, we get $\E[\bs(j)\bs(j')] = 0$ and therefore, the above expression simplifies to
\begin{align*}
\E[(\bL x_{\textSmall})_i^2] = \sum_{j \in \textSmall} \Pr[\bh(j) = i]x_j^2 = {\opnorm{x_{\textSmall}}^2}/{t}.
\end{align*}
Similarly, we have
\begin{align*}
    \E[(\bL x_{\textSmall})_i^4] = \sum_{j_1,j_2,j_3,j_4 \in \textSmall}\Pr[\bigwedge_{k=1}^4 (\bh(j_k)=i)] \E[\prod_{k=1}^4\bs(j_k)]\prod_{k=1}^4 x_{j_k}.
\end{align*}
We now see using the 4-wise independence of $\bs$ that $\E[\bs(j_1)\bs(j_2)\bs(j_3)\bs(j_4)]$ is nonzero only when all the indices $j_1,j_2,j_3,j_4$ are equal or we can pair the indices $j_1,j_2,j_3,j_4$ into two groups taking same values. Hence,
\begin{align*}
    \E[(\bL x_{\textSmall})_i^4]
    &= \frac{\|x_{\textSmall}\|_4^4}{t} + \frac{6}{t^2}\sum_{j_1 < j_2 \in \textSmall} x_{j_1}^2x_{j_2}^2\\
    &= \frac{\|x_{\textSmall}\|_4^4}{t} + \frac{3}{t^2}\left(\opnorm{x_{\textSmall}}^4 - \|x_{\textSmall}\|_4^4\right)
\end{align*}
which then implies \[\Var[(\bL x_{\textSmall})_i^2] \le (1/t)\|x_{\textSmall}\|_4^4 + (2/t^2)\|x_{\textSmall}\|_2^4.\] Since $\|x_{\textSmall}\|_4^4 \le \linf{x_{\textSmall}}^2\opnorm{x_{\textSmall}}^2 \le \alpha^2\opnorm{x}^4$, we get $\Var[(\bL x_{\textSmall})_i^2] \le (2/t^2+\alpha^2/t)\opnorm{x}^4$. By Chebyshev's inequality,
\begin{align*}
    \Pr[(\bL x_{\textSmall})_i^2 \ge \opnorm{x_{\textSmall}}^2/t + \gamma] \le \frac{(2/t^2 + \alpha^2/t)\opnorm{x}^4}{\gamma^2}.
\end{align*}
By a union bound over all $t$ values of $i$, we get that $\linf{\bL x_{\textSmall}}^2 \le \opnorm{x_{\textSmall}}^2/t + \gamma$ with probability $\ge 1 - (2/t + \alpha^2)\opnorm{x}^4/\gamma^2$. For $t \ge 1/\alpha^2$ and $\gamma = (2\alpha/\sqrt{\delta})\opnorm{x}^2$, we have that with probability $\ge 1 - \delta$, $\linf{\bL x_{\textSmall}}^2 \le (\alpha^2  + 2\alpha/\sqrt{\delta})\opnorm{x}^2 \le (3\alpha/\sqrt{\delta})\opnorm{x}^2$. We thus finally have that if $t \ge 1/\alpha^2$, then with probability $\ge 1 - \delta$,
\begin{align*}
    \linf{\bL x_{\textSmall}} \le \frac{2\sqrt{\alpha}}{\delta^{1/4}}\opnorm{x}.
\end{align*}
Hence, by a union bound, with probability $\ge 1 - 2\delta$, $\linf{\bL x_{\textLarge}} = \linf{x_{\textLarge}}$ and $\linf{\bL x_{\textSmall}} \le (2\sqrt{\alpha}/\delta^{1/4})\opnorm{x}$. Condition on this event. If $\linf{x} \ge \alpha\opnorm{x}$, then $\linf{x_{\textLarge}} = \linf{x}$ and by triangle inequality, we get
\begin{align*}
    \linf{\bL x} &= \linf{\bL x_{\textLarge} + \bL x_{\textSmall}}\\
    &= \linf{\bL x_{\textLarge}} \pm \linf{\bL x_{\textSmall}}\\
    &= \linf{x} \pm (2\sqrt{\alpha}/\delta^{1/4})\opnorm{x}.
\end{align*}
If $\linf{x} < \alpha\opnorm{x}$, then $\textLarge = \emptyset$ and $\linf{\bL x} = \linf{\bL x_{\textSmall}} \le (2\sqrt{\alpha}/\delta^{1/4})\opnorm{x}$ which clearly satisfies $\linf{\bL x} = \linf{x} \pm (2\sqrt{\alpha}/\delta^{1/4})\opnorm{x}$. 

We will now bound $\opnorm{\bL x}^2$. First we have,
\begin{align*}
    \opnorm{\bL x}^2 &= \sum_{i \in [t]}\left(\sum_{j \in [d] : \bh(j) = i}\bs(j)x_j\right)^2\\
    &= \sum_{i \in [t]}\left(\sum_{\substack{j \in [d]\\\bh(j)=i}}x_j^2 + \sum_{\substack{j_1 \ne j_2 \in [d]:\\\bh(j_1) = \bh(j_2) = i}}\bs(j_1)\bs(j_2)x_{j_1}x_{j_2}\right)\\
    &=\opnorm{x}^2 + 2\sum_{i \in [t]}\sum_{j_1 < j_2 \in [d]}[\bh(j_1)=\bh(j_2) =i]\bs(j_1)\bs(j_2)x_{j_1}x_{j_2}.
\end{align*}
By 4-wise independence of $\bs$, we get $\E[\bs(j_1)\bs(j_2)] = 0$ for $j_1 \ne j_2$ and therefore get $\E[\opnorm{\bL x}^2] = \opnorm{x}^2$. We now bound $\Var(\opnorm{\bL x}^2)$.
\begin{small}
\begin{align*}
    &\Var(\opnorm{\bL x}^2) = \E[(\opnorm{\bL x}^2 - \opnorm{x}^2)^2]\\
    &=4\sum_{i_1, i_2}\sum_{\substack{j_1 < j_2\\ j_3 < j_4}}\Pr[\bh(j_1) = \bh(j_2) = i_1, \bh(j_3)=\bh(j_4)=i_2] \times \E[\bs(j_1)\bs(j_2)\bs(j_3)\bs(j_4)]x_{j_1}x_{j_2}x_{j_3}x_{j_4}.
\end{align*}    
\end{small}
Note that by 4-wise independence of the hash family from which the function $\bs$ is drawn, we get that $\E[\bs(j_1)\bs(j_2)\bs(j_3)\bs(j_4)]$ is $0$ unless $j_1 = j_3$ and $j_2 = j_4$ (since $j_1 < j_2$ and $j_3 < j_4$). If $j_1 = j_3$ and $j_2 = j_4$, we additionally have that $\Pr[\bh(j_1) = \bh(j_2) = i_1, \bh(j_3) = \bh(j_4) = i_2] = 0$  unless $i_1 = i_2$. Thus, the above expression simplifies to
\begin{align*}
    \Var(\opnorm{\bL x}^2) &= 4\sum_{i}\sum_{j_1 < j_2}\Pr[\bh(j_1) = \bh(j_2)= i]x_{j_1}^2x_{j_2}^2\\
    & = 4\sum_i (1/t^2)\sum_{j_1 < j_2}x_{j_1}^2x_{j_2}^2\\
    &= (2/t)(\opnorm{x}^4 - \|x\|_4^4).
\end{align*}
For $t \ge 1/2\alpha^4\delta$, we have $\Var(\opnorm{\bL x}^2) \le 4\alpha^4\delta\opnorm{x}^4$ and by Chebyshev inequality, we get that $\Pr[\opnorm{\bL x}^2 \ge \opnorm{x}^2 + 2\alpha^2\opnorm{x}^2] \le \delta$. By the union bound, we obtain the result.
\end{proof}
We now prove the following theorem.
\begin{theorem}
There is a turnstile stream algorithm using $O(\varepsilon^{-2}\log(1/\varepsilon)\log d)$ bits of space and outputs an estimate to $\linf{x}$ up to an additive error of $\varepsilon\opnorm{x}$ with probability $\ge 9/10$.
\end{theorem}
\begin{proof}
    In the above lemma, setting $\delta = 1/100$, $\alpha = \varepsilon^2/160$, we get that for $t = C/\varepsilon^8$ for a large enough constant $C$, with probability $\ge 97/100$, $\linf{\bL x} = \linf{x} \pm (\varepsilon/2)\opnorm{x}$ and $\opnorm{\bL x}^2 \le (1+4\varepsilon^4)\opnorm{x}^2$. Condition on this event. From \cite{Charikar2004finding}, if a vector $y \in \R^{m}$ is being updated in a turnstile stream, the CountSketch data structure with parameters $\countsketchrange = O(1/\varepsilon^2)$ and $\countsketchrepetitions = O(\log m)$ can be used to recover a vector $\hat{y}$ such that with probability $\ge 99/100$,
    \begin{align*}
        \linf{y - \hat{y}} \le (\varepsilon/3)\opnorm{y}
    \end{align*}
    and therefore have that $\linf{\hat{y}} = \linf{y} \pm (\varepsilon/3)\opnorm{y}$. Note that the hash functions for the CountSketch data structure can be stored using $O(\countsketchrepetitions \cdot \log m) = O(\log^2 m)$ bits. If the vector $y$ has entries bounded by $\poly(d)$, then the CountSketch data structure can be stored in $O(\countsketchrange \cdot \countsketchrepetitions \cdot \log d)$ bits. We now note that $\bL x$ is a $C/\varepsilon^8$ dimensional vector with entries bounded by $\poly(d)$. Hence using a CountSketch data structure, we can obtain an estimate $\Est$ that satisfies \[\Est = \linf{\bL x} \pm (\varepsilon/3)\opnorm{\bL x} = \linf{x} \pm \varepsilon\opnorm{x}.\]

    The two stage sketching procedure can be implemented in a turnstile stream as follows: when we receive an update $(i, \Delta)$ to the vector $x$, we supply the update $(\bh(i), \bs(i)\Delta)$ to the CountSketch data structure. The overall space usage of the algorithm is $O(\log d + (\log 1/\varepsilon)^2 + \varepsilon^{-2}\log(1/\varepsilon)\log d)$ bits where we use $O(\log d)$ bits to store the hash functions corresponding to $\bL$, $O(\log 1/\varepsilon)^2$ bits to store the hash functions corresponding to the CountSketch data structure and $O(\varepsilon^{-2}\log(1/\varepsilon)\log d)$ bits to store the CountSketch table itself. 

    To process each update in the stream, we require $O(\log 1/\varepsilon)$ time in the Word RAM model with a word size $O(\log d)$ as each update only involves evaluating $O(\log 1/\varepsilon)$ constant wise independent hash functions.
\end{proof}
We will now show that the above is tight up to constant factors.
\subsection{Space Lower Bound for estimating \texorpdfstring{$\linf{x}$}{l-infty-norm of x} in a turnstile stream}
To lower bound the space complexity of the turnstile streaming algorithm, we reduce from the Augmented Sparse Set-Disjointness problem. We define this communication problem as a combination of the so-called Augmented INDEX problem \cite{CW09} and Sparse Set-Disjointness problem \cite{dasgupta2012disjoint}. In this problem, Alice is given sets $A_1, \ldots, A_t \subseteq [n]$ and Bob is given the sets $B_1,\ldots,B_t \subseteq [n]$. Assume that for all $j \in [t]$, $|A_j|=|B_j| = k$. Bob is given an index $j$ and the sets $A_{1},\ldots,A_{j-1}$ and has to output using a one-way message $M$ from Alice if $A_j \cap B_j = \emptyset$ or not. Suppose that Alice and Bob have access to a shared random string. We say that a randomized one-way protocol has $\delta$ error if for any instance of the Augmented Sparse Set-Disjointness problem, when Alice and Bob run the protocol $\Pi$, Bob outputs the correct answer with probability $\ge 1 - \delta$. Note that in the one-way protocol, Alice can only send a single message $M$ (possibly randomized using the shared random string) to Bob. The $\delta$-error communication complexity of the Augmented Sparse Set-Disjointness problem is then defined as the minimum over all $\delta$-error protocols of the maximum length, measured in terms of number of bits, of the message sent by Alice over all the inputs. By Yao's minimax principle, we can lower bound the communication complexity of the problem by exhibiting a distribution over the inputs such that any \emph{deterministic} protocol must have a large communication complexity for Bob to output correct answer with probability $\ge 1 - \delta$ (over the distribution of inputs). We now show that there is a small enough constant $\delta$ for which the $\delta$-error communication complexity of the Augmented Sparse Set-Disjointness problem is $\Omega(tk\log k)$.

\begin{theorem}
Let $t$ be arbitrary. If $n \ge k^2$, there exists a small enough universal constant $\delta$ (independent of $t$, $k$ and $n$) such that the $\delta$ error randomized communication complexity of the Augmented Sparse Set-Disjointness problem is $\Omega(tk\log k)$.
\end{theorem}
\begin{proof}
Given parameters $k$ and $\alpha < 1/2$, \cite{dasgupta2012disjoint} show that there exists a family $\calX$ of $2^{\alpha k\log k}$ subsets of $[k^2]$ such that (i) $|X| = k$ for all $X \in \calX$ and (ii) for all $X \ne X' \in \calX$, we have $|X \cap X'| \le \alpha k$. They also show that corresponding to the family $\calX$, there is a family $\calY$ of subsets of $[k^2]$ such that (i) $|Y| = k$ for all $Y \in \calY$, (ii) $|\calY| \le ak\log k$ for an absolute constant $a = a(\alpha)$ (independent of $k$) and (iii) for $X \ne X' \in \calX$, there is at least one set $Y \in \calY$ such that exactly one of the sets $X \cap Y$ and $X' \cap Y$ is non-empty. The last property implies that the sequence $(\Disj(X, Y))_{Y \in \calY}$ is distinct for each $X \in \calX$. Here $\Disj(X,Y) = 1$ if $X \cap Y = \emptyset$ and $0$ otherwise.

We will now define a distribution over the instances of the Augmented Set-Disjointness problem such that any deterministic protocol must have Alice sending a message with $\Omega(tk\log k)$ bits which will prove the theorem by Yao's minimax principle. Let $\bX_1, \ldots, \bX_t \sim \calX$ and $\bY_1, \ldots, \bY_t \sim \calY$ be drawn independently. We let $\bX = (\bX_1,\ldots,\bX_t)$ denote the sets given to Alice and $\bY = (\bY_1,\ldots,\bY_t)$ denote the sets given to Bob. Let $\bi \sim [t]$ drawn uniformly at random be the index given to Bob. Let $M(\bX)$ be the message sent by Alice when running an arbitrary deterministic protocol with an error $\delta$ over the distribution defined by the random variables $\bX$, $\bY$,  and $\bi$. By the chain rule of entropy,
\begin{align*}
    H(\bX \mid M(\bX)) = \sum_{i \in [t]} H(\bX_i \mid M(\bX), \bX_{< i}).
\end{align*}
As $\bX_i$ is uniquely identifiable by the sequence $(\Disj(\bX_i, Y))_{Y \in \calY}$, we have
\begin{align*}
    H(\bX \mid M(\bX)) &= \sum_{i \in [t]}H((\Disj(\bX_i, Y))_{Y \in \calY} \mid M(\bX), \bX_{< i})\\
    &\le \sum_{i \in [t]}\sum_{Y \in \calY} H(\Disj(\bX_i, Y) \mid M(\bX), \bX_{< i}) && \text{(sub-additivity)}\\
    &= \sum_{i \in [t]}\sum_{y \in \calY}H(\Disj(\bX_i, \bY_i) \mid M(\bX), \bX_{< i}, \bY_i = y)\\
    &= \sum_{i \in [t]}|\calY| H(\Disj(\bX_i, \bY_i) \mid M(\bX), \bX_{< i}, \bY_i) &&  \text{(since $\bY_i$ is uniform over $\calY$)}\\
    &= |\calY|\sum_{i \in [t]} H(\Disj(\bX_i, \bY_i) \mid M(\bX), \bX_{< i}, \bY).
\end{align*}
Here the last equality follows from the fact that $\bY_1,\ldots,\bY_t$ are mutually independent and are also independent of $\bX_{< i}$ and $M(\bX)$. As the \emph{deterministic} protocol has $\delta$ error over the distribution of the instances we defined above, we have that there is a deterministic function $f$ (which Bob runs to output his answer) such that
\begin{align*}
    \Pr_{\bX, \bY, \bi}[f(M(\bX), \bi, \bY, \bX_{<\bi}) = \Disj(\bX_{\bi}, \bY_{\bi})] \ge 1 - \delta.
\end{align*}
We then have that with probability $1 - \sqrt{\delta}$ over $\bi$ that
\begin{align*}
    \Pr_{\bX, \bY}[f(M(\bX), \bi, \bY, \bX_{< \bi}) = \Disj(\bX_{\bi}, \bY_{\bi})] \ge 1 - \sqrt{\delta}. 
\end{align*}
Let $\textsc{Good} \subseteq [t]$ denote all the indices $i$ for which the above holds. We have $|\textsc{Good}| \ge (1 - \sqrt{\delta})t$. We now note that $\Disj(\bX_i, \bY_i)$ is a $0$/$1$ random variable. By Fano's inequality, for all $i \in \textsc{Good}$, we obtain
\begin{align*}
    H(\Disj(\bX_i, \bY_i) \mid M(\bX), \bY, \bX_{< i}) \le H(\sqrt{\delta}). 
\end{align*}
For $i \notin \textsc{Good}$, we simply have $H(\Disj(\bX_i, \bY_i) \mid M(\bX), \bY, \bX_{<i}) \le H(\Disj(\bX_i, \bY_i)) \le 1$. Thus,
\begin{align*}
    H(\bX \mid M(\bX))
    &\le |\calY|\sum_{i \in [t]} H(\Disj(\bX_i, \bY_i) \mid M(\bX), \bX_{<i}, \bY)\\
    &=|\calY|\sum_{i \in \textsc{Good}} H(\Disj(\bX_i, \bY_i) \mid M(\bX), \bX_{<i}, \bY)\\
    &\qquad \qquad + |\calY|\sum_{i \notin \textsc{Good}}H(\Disj(\bX_i, \bY_i) \mid M(\bX), \bX_{<i}, \bY)\\
    &\le |\calY| t(1-\sqrt{\delta})H(\sqrt{\delta}) + |\calY|t\sqrt{\delta}\\
    &= |\calY|t(\sqrt{\delta} + (1-\sqrt{\delta})H(\sqrt{\delta})).
\end{align*}
Now, $H(\bX \mid M(\bX)) \ge H(\bX) - H(M(\bX))$ by the chain rule and sub-additivity which implies from the above inequality that 
\begin{align*}
    H(M(\bX)) \ge H(\bX) - |\calY|t (\sqrt{\delta} + (1-\sqrt{\delta})H(\sqrt{\delta})).
\end{align*}
As $H(\bX) = H((\bX_1,\ldots,\bX_t)) = tH(\bX_1) = t\alpha k\log k$, we have
\begin{align*}
    H(M(\bX)) &\ge t\alpha k\log k - tak\log k(\sqrt{\delta} + (1-\sqrt{\delta})H(\sqrt{\delta}))\\
    &\ge tk\log k(\alpha - a(\sqrt{\delta} + (1-\sqrt{\delta})H(\sqrt{\delta}))).
\end{align*}
Now we note that $H(\sqrt{\delta}) \le 2\delta^{1/4}$ and get $H(M(\bX)) \ge tk\log k(\alpha - a(\sqrt{\delta} + 2\delta^{1/4}))$. As $a = a(\alpha)$ is purely a function of $\alpha$ independent of $k$, by picking $\delta$ to be a small enough function of $\alpha$, we get $H(M(\bX)) \ge (\alpha/2)tk\log k$. Finally, this implies that $\max_{\bX}|M(\bX)| \ge H(M(\bX)) \ge (\alpha tk \log k)/2$. Picking $\alpha = 1/2$, we obtain that there is a small enough constant $\delta$ and a product distribution $\calD$ over $\calX \otimes \calY \otimes [t]$ such that any \emph{deterministic} one-way protocol that solves the Augmented Sparse Set-Disjointness problem with probability $\ge 1 - \delta$ over the distribution $\calD$ must have a communication complexity of $\Omega(tk\log k)$ bits. 
\end{proof}
Using the above lower bound, we can now show that any turnstile streaming algorithm that approximates the $\ell_{\infty}$ norm of a $d$ dimensional vector $x$ with integer coordinates bounded in absolute value by $\poly(d)$ up to an additive error of $\varepsilon\opnorm{x}$ must use $O(\varepsilon^{-2}\log(1/\varepsilon)\log d)$ bits of space.
\begin{theorem}
There exists a small enough constant $\delta$ such that if $\varepsilon \ge 6((\log d)/d)^{1/4}$, any turnstile streaming algorithm that estimates  $\linf{x}$ up to an additive error of $\varepsilon\opnorm{x}$ with probability $\ge 1 - \delta$, of a $d$-dimensional vector $x$ with integer entries bounded in absolute value by $\poly(d)$, must use $\Omega(\varepsilon^{-2}\log(1/\varepsilon)\log(d))$ bits of space.
\end{theorem}
\begin{proof}
    Let $t = \log d$, $k = 1/\varepsilon^2$ and $n = 1/\varepsilon^{4}$. From the above theorem, the Augmented Sparse Set-Disjointness problem with these parameters has a randomized communication complexity of $\Omega(\varepsilon^{-2}\log(1/\varepsilon)\log d)$ bits. 

    Suppose given an instance of the Augmented Sparse Set-Disjointness problem, Alice computes a vector $x \in \R^{(\log d) \cdot 1/\varepsilon^{4}}$ as follows: the vector $x$ is divided into $\log d$ blocks---one for each of the sets $A_1,\ldots,A_{\log d}$. The $i$-th block of vector $x$, for $i=1,\ldots,t$, is defined to be $10^{t-i} \cdot a_i$ where $a_i$ is a binary-vector representation of the set $A_i$. As $t = \log d$, we obtain that $\linf{x} \le \poly(d)$. Let $\bM$ be a randomized turnstile streaming algorithm for estimating $\linf{x}$. Let $\bM(x)$ be the state of the turnstile streaming algorithm after feeding the coordinates of the vector $x$ to $\bM$. Alice transmits the state $\bM(x)$ to Bob. As Bob has an index $i$ and the sets $A_1,\ldots, A_{i-1}$, Bob can construct a vector $y \in \R^{(\log d) \cdot 1/\varepsilon^4}$ such that the $j$-th block of vector $y$ for $j=1,\ldots,i-1$ is the same as the $j$-th block of the vector $x$. The rest of the blocks of $y$ are set to be $0$. We now note that the only non-zero blocks of the vector $x-y$ are $i,i+1,\ldots, t$.

    Bob feeds the updates corresponding to the vector $-y$ to the streaming algorithm $\bM$ starting with the state $\bM(x)$ to obtain $\bM(x-y)$. Finally, Bob defines a vector $z$ with the $i$-th block being the vector $10^{t-i} \cdot b_i$ where $b_i$ is the binary vector corresponding to the set $B_i$. The rest of the blocks of $z$ are set to be $0$. Bob finally updates the state of the streaming algorithm to obtain $\bM(x-y+z)$. If $A_i \cap B_i = \emptyset$, we have $\linf{x - y + z} = 10^{t-i}$ and if $A_i \cap B_i \ne \emptyset$, then $\linf{x - y + z} = 2 \cdot 10^{t-i}$. Additionally, we have $\opnorm{x - y + z}^2 \le 4 \cdot 10^{2(t-i)} \cdot k + \sum_{j = i+1}^t 10^{2(t-j)} \cdot k \le 5 \cdot 10^{2(t-i)} \cdot k$. Thus, $\opnorm{x-y+z} \le 3 \cdot 10^{t-i} \cdot (1/\varepsilon)$ since $k = 1/\varepsilon^2$. Thus an approximation of $\linf{x - y +z}$ up to an additive error of $(\varepsilon/6)\opnorm{x-y+z}$ lets Bob output the correct answer for the instance of Augmented Sparse Set-Disjointness problem. 
    
    By the lower bound on communication complexity of the Augmented Sparse Set-Disjointess problem, we obtain that any turnstile streaming algorithm that outputs, with probability $1 - \delta$ for a small enough constant $\delta$, an approximation to $\linf{x}$ up to an additive error of $(\varepsilon/6)\linf{x}$, for a $(1/\varepsilon^4)\log d$ dimensional vector $x$ with coordinates of absolute values bounded by $\poly(d)$, must use $\Omega(\varepsilon^{-2}\log(1/\varepsilon)\log d)$ bits. As $\varepsilon \ge ((\log d)/d)^{1/4}$ implies $d \ge \varepsilon^{-4}\log d$, we obtain the result.
\end{proof}
Note that the above lower bound crucially uses that the algorithm is a turnstile streaming algorithm and does not hence lower bound the space complexity of the algorithms in the insertion-only streams where only nonnegative updates are allowed to the vector $x$ being maintained in the stream.
\subsection{Tighter bounds for vectors with large \texorpdfstring{$\linf{x}$}{l-inf-x}}
The lower bound in the previous section shows that $\Theta(\varepsilon^{-2}\log(d)\log(1/\varepsilon))$ bits of space is both necessary and sufficient to approximate $\linf{x}$ up to an additive error $\varepsilon\opnorm{x}$. The hard instance in the lower bound has the property that $\linf{x} = O(\varepsilon\opnorm{x})$. We show that assuming the vector $x$ satisfies, $\linf{x} \ge c\opnorm{x}$ for a constant $c$, we can beat the lower bound and obtain a better than $\varepsilon\opnorm{x}$ additive error using $O(\varepsilon^{-2}\log d)$ bits of space. Note that the condition $\linf{x} \ge c\opnorm{x}$ is natural in certain settings where the coordinates of the vector $x$ follow the Zipf's law. We prove the following theorem.
\begin{theorem}
    Suppose a $d$ dimensional vector $x$ is being maintained in a turnstile stream. Assume that the coordinates of $x$ are integers and are bounded in absolute value by $\poly(d)$. If $x$ is such that $\linf{x} \ge c\opnorm{x}$ for a universal constant $c$, then there is an algorithm that uses $O(\varepsilon^{-2}\log d)$ bits of space and outputs an estimate $\Est$ that satisfies
    \begin{align*}
        \Est = \linf{x} \pm \varepsilon\opnorm{x}
    \end{align*}
    with probability $\ge 9/10$. The update time of the algorithm is $O(\log 1/\varepsilon)$ in the Word RAM model with a word size $\Omega(\log d)$.
\end{theorem}
\begin{proof}
Without loss of generality, assume $\varepsilon \le c/10$. Let $\textLarge = \setbuilder{i}{|x_i| \ge (c/2)\opnorm{x}}$. We have $|\textLarge| \le 4/c^2$. If $\bL : \R^d \rightarrow \R^{d'}$ is a randomized linear map to $d' = \poly(1/\varepsilon)$ dimension constructed using a 2-wise independent hash function $\bh$ and a 4-wise independent sign function $\bs$ as in Lemma~\ref{lemma:first-level-hashing}, then we have that
\begin{enumerate}
    \item all the coordinates in $\textLarge$ are hashed to different coordinates, 
    \item $\opnorm{\bL x}^2 \le (1+\varepsilon^8)\opnorm{x}^2$,
    \item for all $i \notin \textLarge$, $|(\bL x)_{\bh(i)}| \le (c/2 +  \varepsilon^2/8)\opnorm{x} \le (3c/5)\opnorm{x}$, and
    \item for all $i \in \textLarge$, $|(\bL x)_{\bh(i)}| = |x_i| \pm (\varepsilon^2/8)\opnorm{x}$.
\end{enumerate}
Now, let $\countsketchrepetitions = O(\log 1/\varepsilon)$, $\countsketchrange = O(1/\varepsilon^2)$ and $b = 1/\varepsilon$. Instantiate a CountSketch data structure $\CS: \R^{d'} \rightarrow \R^{rt}$ with these parameters and derandomized using \FastPRG as in Theorem~\ref{theorem:fastprg-count-sketch} with a word size $w = \Omega(\log d)$. From Theorem~\ref{theorem:fastprg-count-sketch}, the parameters (random seed for \FastPRG) of the map $\CS$ and the value $\CS(\bL x)$ can be stored using $O(rt + b \log_b d') = O(\varepsilon^{-2}\log 1/\varepsilon + \varepsilon^{-1}\log_{\varepsilon^{-1}}\poly(1/\varepsilon)) = O(\varepsilon^{-2}\log 1/\varepsilon)$ \emph{words} of space. We also have that the update time of the CountSketch data structure instantiated with these parameters is $O(\log 1/\varepsilon)$ in the Word RAM model with a word size $\Omega(\log d)$.
If $x$ receives a turnstile update $(i, \Delta)$, then 
\begin{align*}
    \bL(x + \Delta e_i) = \bL x + \Delta \cdot (\bL e_i).
\end{align*}
By definition of the map $\bL$, the vector $\bL e_i$ is nonzero in exactly one coordinate $\bh(i)$. Thus, we further obtain
\begin{align*}
    \CS( \bL(x + \Delta e_i)) = \CS(\bL x + \bs(i)\Delta e_{\bh(i)}).
\end{align*}
Now, by Theorem~\ref{theorem:fastprg-count-sketch}, the vector $\CS(\bL(x+\Delta e_i))$ can be computed using the value of $\CS(\bL x)$ in time $O(r \log_b d') = O(\log 1/\varepsilon)$ time in Word RAM model. Thus, the randomized two-level sketch $\CS \circ \bL$ can be applied to the underlying vector $x$ in a turnstile stream using a total space of $O(\varepsilon^{-2} \log 1/\varepsilon)$ \emph{words} of space and each turnstile update can be processed in $O(\log 1/\varepsilon)$ time on a Word RAM machine with a word size $\Omega(\log d)$. 

Theorem~\ref{theorem:fastprg-count-sketch} also gives the following recovery guarantees: for any $i \in [d']$ and $\alpha < 1$, we can recover a value $\widehat{(\bL x)_i}$ such that
\begin{align*}
    \Pr_{\CS}[|\widehat{(\bL x)_i} - (\bL x)_{i}| \ge \alpha\varepsilon\opnorm{\bL x}] \le \exp(-\alpha^2 r) + 1/\poly(d).
\end{align*}
Setting $\alpha = 1/\sqrt{\log 1/\varepsilon}$, a union bound over the indices in the set $h(\textLarge) \subseteq [d']$ gives that with probability $\ge 99/100$ over $\CS$, for all $i \in h(\textLarge)$, $|\widehat{(\bL x)_{i}} - (\bL x)_i| \le (\varepsilon/\sqrt{\log 1/\varepsilon})\opnorm{\bL x}$ and setting $\alpha = 1$ and a union bound over all the coordinates $i \in [d']$ gives that with probability $\ge 99/100$, for all $i \in [d]$, $|\widehat{(\bL x)_i} - (\bL x)_i| \le \varepsilon\opnorm{\bL x}$. Conditioned on the properties of the map $\bL$ above, overall, we obtain that with a probability $\ge 9/10$,
\begin{align*}
    \linf{\bL x} = \linf{x} \pm \frac{2\varepsilon}{\sqrt{\log 1/\varepsilon}}\opnorm{x}.
\end{align*}
Hence, $\linf{x}$ can be estimated to an additive error of $(\varepsilon/\sqrt{\log 1/\varepsilon})\opnorm{x}$ using only $O(\varepsilon^{-2}\log 1/\varepsilon)$ words of space and the time to update the state in a turnstile stream is $O(\log 1/\varepsilon)$ in the Word RAM model with a word size $\Omega(\log d)$. Setting $\varepsilon' = 2\varepsilon/\sqrt{\log 1/\varepsilon}$, we obtain the result.
\end{proof}
\section*{Acknowledgements}
Praneeth Kacham and David P. Woodruff were supported by National Institute of Health (NIH) grant 5R01 HG 10798-2. Rasmus Pagh was supported by the VILLUM foundation grant 16582 and by a Providentia, a Data Science Distinguished Investigator grant from Novo Nordisk Fonden. Mikkel Thorup was supported by the VILLUM foundation grant 16582.

\bibliographystyle{alpha}
\bibliography{main}


\appendix

\section{Nisan's Pseudorandom Generator}\label{sec:nisan}

We say that a \emph{randomized} program uses space $w$ with a block size $n$ if it accepts its random bits as an $n$ bit block at a time and uses at most space $w$ between the different blocks of random bits. 
Such programs can be modeled as a finite state machine over at most $2^w$ states, taking an input string over the alphabet $\set{0,1}^n$. Nisan \cite{nisan} constructed a pseudorandom generator which requires only a small uniform random seed that ``fools'' a space $w$ program with a block size $n$.

Let $\bh_1,\ldots,\bh_k$ be independent hash functions drawn from a $2$-wise independent hash family $\calH = \set{h : \set{0,1}^n \rightarrow \set{0,1}^n}$. These hash functions together with $\bx\in\set{0,1}^n$, sampled uniformly at random, serve as the \emph{seed} of the generator $G_k : \set{0,1}^n \rightarrow \set{0,1}^{2^k \cdot n}$, defined recursively as follows:
\begin{align*}
    G_0(x) &:= x\\
    G_k(x, \bh_1,\ldots, \bh_k) &:= G_{k-1}(x, \bh_1,\ldots,\bh_{k-1}) \circ G_{k-1}(\bh_k(x), \bh_1, \ldots, \bh_{k-1}),
\end{align*}
where $\circ$ denotes the string concatenation. 
For a given choice of $\bh_1,\ldots,\bh_k$, define the distribution $G_k(*, \bh_1,\ldots,\bh_k)$ over bitstrings of length $2^k \cdot n$ to be the distribution of $G_k(\bx, \bh_1,\ldots,\bh_k)$ for random $\bx\in\set{0,1}^n$.
Nisan showed that for any fixed FSM with at most $2^w$ states over alphabet $\set{0,1}^n$, with high probability over the hash functions $\bh_1,\ldots,\bh_k$, the distribution $G_k(*, \bh_1,\ldots,\bh_k)$ is indistinguishable from the uniform distribution over $\set{0,1}^{2^k \cdot n}$. 
The power of Nisan's generator is summarized by the following lemma (using notation from section~\ref{sec:hashprg}):
\begin{lemma}
There exists a constant $c > 0$ such that given integers $n$ and $w \le cn$ and parameter $k \le cn$, for any FSM $Q$ with $2^w$ states, if $\bh_1,\ldots,\bh_k : \set{0,1}^n \rightarrow \set{0,1}^n$ are drawn independently from a $2$-wise independent hash family, then with probability $\ge 1 - 2^{-cn}$,
\begin{align*}
    \|Q(G_k(*, \bh_1,\ldots,\bh_k)) - Q((U_n)^{2^k})\| \le 2^{-cn}
\end{align*}
where $\|M\| := \max_i \sum_j |M_{ij}|$.
\end{lemma}
Note that  $\|M\|= \max_{x \ne 0}\|Mx\|_{\infty}/\|x\|_{\infty}$ where $\|x\|_{\infty} = \max_i |x_i|$. We therefore have that for any two matrices $A$ and $B$, $\|A+B\| \le \|A\|+\|B\|$ and $\|AB\| \le \|A\|\|B\|$. Therefore, we obtain that with probability $\ge 1 - 2^{-cn}$ over the hash functions $\bh_1,\ldots,\bh_k$, we have that the total variation distance between the distribution of final state using a random string drawn from $(U_n)^{2^k}$ and a random string drawn from $G_k(*, \bh_1,\ldots,\bh_k)$ is at most $2^{-cn}$.

Specifically, for a $w = O(\log d)$ space algorithms using $\poly(d)$ random bits, we have that we can use Nisan's Generator with $n, k = O(\log d)$. The time to evaluate a block of $n$ random bits is then $\Omega(k) = \Omega(\log d)$ in the Word RAM model as the $k$ hash functions have to be applied sequentially to the random seed. Using our new pseudorandom generator, which we call \FastPRG, we show that we can set $k = O(1)$ at the expense of using more space to store the hash functions.

\section{Finding Heavy Entries}\label{sec:heavy-entries}
See Section~\ref{sec:derandomizing-countsketch} for the definition of CountSketch data structure. Note that $[\countsketchrange]$ denotes the range of locations the coordinate gets hashed into and $\countsketchrepetitions$ denotes the number of repetitions. Further for each $\ell \in [d]$, $\hat{x}_\ell$ defined in \eqref{eq:countsketch-estimator} denotes our estimate for the value of coordinate $x_{\ell}$.

Jowhari, Sa{\u{g}}lam and Tardos \cite{jowhari2011tight} show that if $\countsketchrepetitions = O(\log d)$, then with probability $\ge 1-1/\poly(d)$, for all $\ell$,
\begin{align*}
    |x_{\ell} - \hat{x}_{\ell}| \le \frac{\lp{x}}{\countsketchrange^{1/p}}.
\end{align*}
By picking $\countsketchrange = (\phi/10)^{-p}$ we obtain that with probability $1 - 1/\poly(d)$, for all $\ell \in [d]$, $|x_{\ell} - \hat{x}_{\ell}| \le (\phi/10)\lp{x}$. The algorithm uses $O((\phi/10)^{-p}\log^2 d)$ bits of space and has an update time of $O(\log d)$ per stream element. Condition on the event that for all $\ell \in [d]$, $|x_{\ell}-\hat{x}_{\ell}| \le (\phi/10)\lp{x}$ for all $\ell \in [d]$. 

Concurrently, run the algorithm of \cite{kane2010exact} with $\varepsilon = 1/4$ to obtain a value $v$ such that with probability $\ge 99/100$
\begin{align*}
    (9/10)\lp{x} \le v \le (11/10)\lp{x}.
\end{align*}
Note that for constant $\varepsilon$, their algorithm uses $O(\log d)$ bits of space and has an update time of $O(1)$ per stream element in the Word RAM model. Condition on this event as well. 

Now, let $L'$ be the set returned by heavy-hitters algorithm of \cite{heavy-hitters-fast-query} with parameter $\phi$. Their algorithm uses $O(\phi^{-p}\log ^2 (d))$ bits of space and has an update time of $O(\log d)$ per stream element. At the end of processing the stream, in time $O(\phi^{-p}\poly(\log d))$, they return a set $L'$ satisfying $|L'| = O(\phi^{-p})$ and
\begin{align*}
    L' \supseteq \setbuilder{\ell}{|x_{\ell}| \ge \phi\lp{x}}.
\end{align*}
The set $L'$ contains all the heavy-hitters and may contain additional coordinates as well. To filter the list $L'$, we use the estimates ${\hat x}_{\ell}$ given by the CountSketch data structure. Define
\begin{align*}
    L = \setbuilder{i \in L'}{|\hat x_i| \ge (8/10)\phi v}.
\end{align*}
Conditioned on the correctness of $L', \hat{x}_{\ell}$ and the estimate $v$, we prove properties about the set $L$. If $|x_{\ell}| \ge \phi\lp{x}$, then $|\hat{x}_{\ell}| \ge (9\phi/10)\lp{x} \ge (9\phi/11)v \ge (8\phi/10)v$. Therefore, $\ell \in L$. On the other hand, if $\ell \in L$ then
\begin{align*}
    |x_{\ell}| \ge |\hat{x}_{\ell}| - (\phi/10)\lp{x} \ge (8/10)\phi v - (\phi/10)\lp{x} \ge (6\phi/10)\lp{x}.
\end{align*}
As $|x_{\ell} - \hat{x}_{\ell}| \le (\phi/10)\lp{x}$ for all $\ell$ and $|x_{\ell}| \ge (6\phi/10)\lp{x}$ for all $\ell \in L$, we also obtain that for all $\ell \in L$, $\text{sign}(x_{\ell}) = \text{sign}(\hat{x}_{\ell})$. As the list $L'$ has size at most $O(\phi^{-p})$, the post processing can be performed in time $O(\phi^{-p}\poly(\log d))$. Thus, we over all have the following lemma.
\begin{lemma}
    Given a stream of updates $(i_1, v_1), \ldots, (i_m, v_m) \in [d] \times \set{-M, \ldots, M}$ for $m, M \le \poly(d)$, a parameter $\phi$ and $p \in (0,2)$, there is a streaming algorithm that uses $O(\phi^{-p}\log^2(d))$ bits of space and has an update time of $O(\log d)$ per stream element and outputs a set $L \subseteq [d]$ at the end of the stream that with probability $\ge 9/10$ satisfies:
    \begin{enumerate}
        \item $L \supseteq \setbuilder{\ell \in [d]}{|x_{\ell}| \ge \phi\lp{x}}$.
        \item For all $\ell \in L$, $|x_{\ell}| \ge (6\phi/10)\lp{x}$.
        \item For all $\ell \in L$, the algorithm also outputs $\text{sign}(x_{\ell})$.
    \end{enumerate}
    At the end of the stream, the algorithm takes only $O(\phi^{-p}\poly(\log d))$ time to compute the set $L$.
    \label{lma:heavy-hitters}
\end{lemma}
\end{document}